\documentclass[12pt]{article}
\usepackage{eurosym}
\usepackage{amsmath}
\usepackage{float}
\usepackage[footnotesize]{caption}
\usepackage{subcaption}
\usepackage[margin=0.95in]{geometry}
\usepackage{color}
\usepackage[pdftex]{graphicx}
\usepackage{graphicx}
\usepackage{psfrag}
\usepackage{overpic}
\usepackage{Here}
\usepackage{multirow}
\usepackage{pgf}
\usepackage{color}
\usepackage{pgf,pgfarrows,pgfnodes,pgfautomata,pgfheaps,pgfshade}
\usepackage{amsmath,amssymb,amsthm,bm,booktabs,array,microtype,placeins,needspace}
\usepackage[colorlinks=true,linkcolor=black,citecolor=black,urlcolor=black,filecolor=black]{hyperref}
\usepackage[T1]{fontenc}
\usepackage{lmodern}
\usepackage{amsfonts}
\usepackage{lscape}
\usepackage{natbib}
\usepackage{enumerate}
\usepackage{lscape}
\usepackage{diagbox}
\usepackage{comment}
\usepackage{makecell}
\usepackage{titling}
\usepackage{setspace}
\usepackage{footmisc}
\renewcommand{\footnotelayout}{\doublespacing}

\usepackage{rotating}
\usepackage{mathtools}
\usepackage{afterpage}
\usepackage{lipsum}
\newcommand{\ind}{\mathbf 1}
\DeclareMathOperator{\Var}{Var}
\DeclareMathOperator{\E}{\textnormal{E}}
\newtheorem{theorem}{Theorem}
\newtheorem{proposition}{Proposition}
\newtheorem{assumption}{Assumption}
\newtheorem{assumptionbc}{Assumption}
\newtheorem{remark}{Remark}
\newtheorem{example}{Example}
\DeclareMathOperator{\Cov}{Cov}
\newtheorem{lemma}{Lemma}[section]
\newtheorem{extensionprop}{Proposition}[section]

\title{Quantiles of Heterogeneous Coefficients in Observed Populations}
\author{Antonio F. Galvao\thanks{Department of Economics, Michigan State University, East Lansing,
MI 48824, USA. E-mail: \texttt{agalvao@msu.edu}.} \and Ulrich Hounyo\thanks{Department of Economics, University at Albany - State University
of New York, Albany, NY 12222, USA. E-mail: \texttt{khounyo@albany.edu}.} \and Jiahao Lin\thanks{Department of Economics, University at Albany - State University
of New York, Albany, NY 12222, USA. E-mail: \texttt{jlin28@albany.edu}.} }
\date{\today}

\usepackage{abstract}
\begin{document}
\begingroup
\singlespacing
\setlength{\droptitle}{-4.2em}
\maketitle
\endgroup
\vspace{-1.5em}

\begin{abstract}
\normalsize\onehalfspacing 

We study estimation and inference for quantiles of heterogeneous coefficients among the observed units in a panel. Treating the coefficients as deterministic, we estimate each individual coefficient using time-series observations, then use the estimates and their standard errors to construct a smoothed quantile estimator and estimate its variance. Under suitable restrictions on the relative growth rates of the cross-sectional and time dimensions, the estimator is asymptotically normal around a deterministic center. A feasible bias correction relaxes these restrictions, accommodating panels with large cross-sectional dimensions. Normal inference extends to the observed-population quantile under a local balance condition requiring the difference between the deterministic center and the target to be negligible relative to the standard error. Without this condition, we establish an impossibility result: the distribution of the studentized error relative to the target cannot be estimated uniformly consistently. To address this limitation, we construct a conservative confidence interval whose coverage tends uniformly to one and whose length attains the optimal worst-case rate in probability. We provide primitive conditions for smooth GMM estimation. An application to mortgage data reveals widespread but uneven declines in delinquency across credit scores.
\vspace{0.75cm}


\medskip{}
\end{abstract}

\vspace*{-0.5cm}

\vfill{}

\thispagestyle{empty}

\pagebreak{}

\doublespacing
\section{Introduction}\label{sec:intro}

Does a policy or business intervention benefit most units in the observed population, or are its gains concentrated among a few? Consider, for instance, a retailer that offers a discount coupon to randomly assigned customers in every store it operates. Higher sales across the chain could reflect modest gains in most stores or large gains in just a few. A positive lower quartile of store effects tells the retailer that the coupon increases average spending in at least three quarters of stores.

In practice, these quantiles may be difficult to estimate precisely because each individual intervention effect is itself estimated. The estimation error changes both their values and their ordering. In a large population, many coefficients near the target quantile may differ by less than their estimation errors, so repeated estimation can change which unit occupies the target rank. Quantile inference must then account for estimation errors across these nearby units, rather than those of a single unit at the target rank.

This paper develops estimation and inference for quantiles of heterogeneous individual-specific coefficients. As the main structural model, we assume a potentially nonlinear moment condition panel data model containing unknown parameters that are allowed to be individual-specific. This model incorporates many semi-parametric models as special cases. We treat individual parameters as deterministic. The main target of interest is the $\tau$-quantile of the distribution of the individual coefficients, which are also deterministic under repeated sampling. This formulation applies when interest focuses on a specific collection of units or when the entire population is observed. In either case, population membership is known, but the coefficients must be estimated. 

We make four main contributions to the literature. First, we propose a two-step smoothed estimator of the observed-population quantile.  
In the first step, each coefficient is estimated from repeated observations on its unit, and both the number of units, $n$, and the observations per unit, $T$, may grow.
We also estimate the corresponding standard error per unit. In the second step, we combine the first-step individual coefficient estimates, and their standard errors, to estimate the population quantile of interest. 
 Our primary interest is the observed-population quantile, which characterizes coefficient heterogeneity among the units being studied, treating their coefficients as fixed. In contrast, the superpopulation quantile characterizes the distribution of coefficients in an underlying population from which the observed units are randomly drawn. Although the same estimator can estimate both quantiles, inference for the observed-population quantile reflects uncertainty from estimating the individual coefficients, whereas inference for the superpopulation quantile additionally reflects uncertainty from sampling units, leading to different convergence rates and asymptotic variances.

As a second contribution, we construct a consistent variance estimator using the same first-step estimates and standard errors. We establish a standard normal limit for the studentized two-step estimator around a deterministic center when $T^{1/2}\ll n\ll T^{3/2}$. A local balance condition on coefficient spacing makes the difference between this center and the target negligible, transferring the normal approximation to the observed-population quantile. The sample growth restriction ensures that the first-step bias is negligible relative to the standard error, whose order is $n^{-1/2}T^{-1/4}$. This rate reflects the information provided by many coefficients near the quantile, without requiring identification of the unit occupying the target rank.

Third, we derive a feasible correction for the leading bias. When the number of units grows sufficiently quickly, relative to the observations per unit, the first-step bias can become comparable to the quantile's standard error. We account for the first-step coefficient-estimation bias, skewness, estimator curvature,
and dependence between the coefficient and standard-error
estimates. Under strengthened first-step conditions, estimating and removing the leading bias relaxes the sample growth restriction to $n/T^{3/2}\to c\in[0,\infty)$ and extends inference to empirical settings with larger numbers of units.
 The correction leaves the first-order variance unchanged. We also establish primitive conditions for smooth GMM estimators, providing a class of first-step estimators
to which the procedures apply.

Finally, we characterize the trade-off involved in removing the local balance condition. We establish an impossibility result showing that, without local balance, the sampling distribution of the studentized estimation error relative to the target quantile cannot be estimated uniformly consistently. Despite this impossibility, we construct a data-dependent conservative confidence interval whose coverage tends uniformly to one and whose length is of order $T^{-1/2}$. A matching lower bound establishes that this rate is optimal in a worst-case sense among all uniformly conservative intervals.
Thus, inference remains possible without local balance, with an explicit cost in precision. We highlight that, for empirical work, the choice of interval depends on two questions:
whether local balance is plausible for the population being studied, and how large $n$ is relative to $T$. Under local balance, when $n/T^{3/2}\to 0$, we provide a normal interval using only the first-step coefficient estimates and their standard errors. When $n/T^{3/2}$
approaches a positive finite constant, a leading bias may be comparable to the quantile's standard error, motivating the
bias-corrected interval under stronger first-step conditions. When the design does not justify local balance, we can still provide a conservative interval. Reporting it alongside the normal
interval shows which conclusions remain supported without local balance.

Simulations illustrate the improvement from bias correction and the coverage of the conservative procedure when balance fails. The empirical illustration compares mortgage repayment performance across origination cohorts at fixed credit scores, and measures
how broadly performance changes across that population. 

This paper connects several strands of the literature. Quantile regression panel methods usually incorporate individual heterogeneity through fixed effects, transformations, or nonlinear specifications
(see, e.g., \citet{Koenker04}, \citet{canay2011}, \citet{arellano2016quantile}, \citet{machado2019}), with further results on smoothing and normal approximation
(see, e.g., \citet{galvao2016smoothed}, \citet{GalvaoGuVolgushev20}).
Our quantiles summarize heterogeneity in the coefficients themselves. Closer in this respect, \citet{chernozhukov2018sorted} develop inference for sorted
effects, obtained by ordering an estimated effect function across a distribution of covariates. We instead study a growing collection of separately estimated unit-specific coefficients. \citet{paddock2011percentile} also construct intervals
for percentiles of effects among observed providers, using a hierarchical Bayesian model. Our analysis holds the coefficient array fixed and establishes how local spacing relative to estimation
uncertainty affects frequentist inference, yielding normal
approximations, feasible variance estimation, and bias correction.

The noisy-coefficient literature studies distributions of heterogeneous coefficients 
under sampling across units. \citet{okui2019,okui2020} develop
distributional and kernel methods for heterogeneous time-series
characteristics. \citet{jochmans2024} derive bias corrections for
distributions and quantiles estimated from noisy draws, while
\citet{barras2022} study distributions of fund measures. Our target
is the quantile of the deterministic coefficient array. Conditioning
on this array removes variation from sampling units, producing
a different leading variance and making its local arrangement
relevant to centering. These features lead to a different variance
estimator and bias analysis for the quantile of the units being studied. \citet{sancibrian2025} develops finite-population variance corrections
for averages and regression-type estimands from noisy repeated measurements.
We study coefficient quantiles, which fall outside his baseline framework
of moment conditions affine in the latent attributes, and develop
inference as both $n$ and $T$ grow. Our focus on the observed population also relates to \citet{abadie2020}, who distinguish sampling uncertainty from uncertainty arising through treatment assignment, with potential outcomes held fixed. Our analysis holds the unit-specific coefficients fixed and studies uncertainty from estimating them using repeated within-unit observations. The target is a quantile of these coefficients, and no treatment-assignment framework is required.

Related work addresses uncertainty about ranks
(see, e.g., \citet{klein2020}, \citet{mogstad2024}), the effect of a selected option
\citep{andrews2024}, and ranking and selection as compound
decisions \citep{gu2023}. Our focus is the coefficient value at
a population quantile, and how its estimation error depends on
many nearby coefficients as the population grows.
Simultaneous coefficient intervals can
also be projected onto that value; we include this construction
in our comparisons. 

It is also useful to distinguish the proposed quantile of coefficients from other
distributional objects used to study heterogeneous effects.
For example, when the coefficients represent unit- or segment-specific
average treatment effects, our quantiles describe how these average effects vary across units. Conventional quantile treatment effects,
by contrast, compare quantiles of potential outcome distributions
\citep{firpo2007}. Intercept comparisons also
require a common, meaningful covariate reference: changing the covariate
origin changes the intercept, and a reference outside the observed
support entails extrapolation \citep{melly2025minimum}. Our formulation
allows a specified coefficient contrast or a fitted response at a common
reference within the relevant support. Relatedly,
\citet{fernandezval2026dynamic} study heterogeneous distribution regression,
including predicted outcome distributions and counterfactual quantile
effects. We study the quantile of the coefficient values themselves,
with inference determined by their local spacing and estimation errors.

The remainder of the paper is organized as follows. Section~\ref{sec:model} introduces the model, target parameter, and estimation.
Sections~\ref{sec:inference}-\ref{sec:conservative} develop normal inference, bias correction, and conservative inference, respectively.
Sections~\ref{sec:simulation} and \ref{sec:empirical} present the
numerical results, and Section~\ref{sec:conclusion} concludes.
Appendix~\ref{sec:primitive} gives primitive GMM conditions.
Appendix~\ref{sec:mainthmproofs} proves the three main theorems;
the Supplementary Materials contain supporting results and extensions.

\section{Model, estimation, and the population of interest}\label{sec:model}

\subsection{Model and target}
\label{sec:target}

For each unit $i=1,\ldots,n$, the researcher observes a sample
$\{\bm Z_{it}:t=1,\ldots,T\}$. The index $t$ labels observations within a unit and may represent time.
The unknown coefficient vector
$\bm\beta_i\in\mathbb R^p$, with fixed dimension $p$, satisfies
\begin{equation}\label{eq:model}
 \E\left[\bm g_i(\bm Z_{it},\bm\beta_i)\right]=\bm0,
\end{equation}
where the function $g_i(\cdot)$ is known to the researcher. The moment condition in equation \eqref{eq:model} encompasses many semi-parametric models as special cases, and can be derived from numerous economic theories. Such restrictions arise from equilibrium and optimality conditions in structural models of consumer and firm behavior, asset pricing, and dynamic optimization, as well as from orthogonality conditions in causal models involving randomized experiments and instrumental variables.

The parameters, $\bm\beta_i$, in \eqref{eq:model} are allowed to vary across units, but are treated as deterministic. For example, the linear model
$Y_{it}=\bm X_{it}^{\top}\bm\beta_i+u_{it}$, with
$\E[\bm X_{it}u_{it}]=\bm0$, yields the moment condition
$\E[\bm X_{it}(Y_{it}-\bm X_{it}^{\top}\bm\beta_i)]=\bm0$.
The formulation in \eqref{eq:model}  accommodates nonlinear models whose coefficients can be estimated by the Generalized Method of Moments (GMM). We study quantiles of the distribution across units of a scalar
coefficient or contrast as following,
\begin{equation}\label{eq:contrast}
 \theta_i=\bm a^{\top}\bm\beta_i,
\end{equation}
where the fixed nonzero vector $\bm a$ can select one coefficient, a difference between coefficients, or the response to a specified
covariate change. The target is a $\tau$-quantile of this scalar contrast, not of the joint coefficient vector.

For each $n$, we take the observed units to be the population
of interest and treat their coefficient array
$\{\theta_i\}_{i=1}^n$ as deterministic. We study a sequence of
such populations as $n,T\to\infty$, with uncertainty arising
from the within-unit observations used to estimate their
coefficients. The distribution within each population is
\begin{equation}
\begin{aligned}\label{eq:cdftarget}
 F_n(x)=\frac1n\sum_{i=1}^n\ind\{\theta_i\le x\}.
\end{aligned}
\end{equation}
For a specified quantile level $\tau\in(0,1)$, the corresponding
target is
\begin{equation}
\begin{aligned}\label{eq:target}
 q_n=F_n^{-1}(\tau)
     =\inf\{x:F_n(x)\ge\tau\}
     =\theta_{(\lceil n\tau\rceil)},
\end{aligned}
\end{equation}
where $\theta_{(k)}$ denotes the $k$-th smallest latent coefficient. For example, $q_n$ may be a lower quartile of responses among schools participating in a program, the median elasticity across a firm's products, or an upper quantile of exposures among funds in a specified portfolio. In each case, the target describes the particular population observed: its membership is known, but its coefficients must be estimated.

\begin{example}\label{ex1}
We use the following treatment experiment example throughout to illustrate the proposed methods for quantiles of deterministic heterogeneous effects. Consider a retailer testing a discount coupon in every store it operates, randomly assigning customers
within each store to receive the offer or no offer. Let $D_{it}$ indicate whether customer $t$ in store $i$ receives the offer, and let $Y_{it}$ denote that customer's spending at the store during the following week,
including zero if the customer makes no purchase. The linear model becomes
\[
 Y_{it}=\alpha_i+\theta_i D_{it}+u_{it},
 \qquad \E[u_{it}\mid D_{it}]=0.
\]
By random assignment, $\theta_i$ measures the average spending gain
from offering the coupon in store $i$. The retailer asks whether these
gains extend across most of its stores or are concentrated in a few,
which motivates studying quantiles of the store-specific effects.
\end{example}

\subsection{Estimation}\label{sec:estimation}
Next, we turn to estimation of the parameters of interest, which is simple and proceeds in two steps. In the first step, we estimate each $\bm\beta_i$ using the $T$ observations for unit $i$ using the moment condition in \eqref{eq:model}, for example
by GMM, and denote the resulting estimate
by $\widehat{\bm\beta}_i=\widehat{\bm\beta}_i(\{\bm Z_{it}\}_t)$.
The corresponding estimator of $\theta_i$ and its estimated
standard error are
\begin{equation}
\begin{aligned}\label{eq:thetahat}
 \widehat\theta_i & =\bm a^{\top}\widehat{\bm\beta}_i,
 \qquad
 \widehat s_i=\frac{\widehat\sigma_i}{\sqrt T},
\end{aligned}
\end{equation}
where $\widehat\sigma_i=\widehat\sigma_i(\{\bm Z_{it}\}_t)$ estimates the first-order standard deviation
$\sigma_i$ of $\sqrt T(\widehat\theta_i-\theta_i)$, and let
$s_i=\sigma_i/\sqrt T$.
In the second step, we combine the first-step coefficient estimates in \eqref{eq:thetahat}  to estimate the population quantile in \eqref{eq:target}. Simply substituting $\widehat\theta_i$ for $\theta_i$ in $F_n$ in equation \eqref{eq:cdftarget} produces a step function, whose lack of differentiability complicates the
asymptotic analysis of the resulting quantile estimator. We instead smooth each comparison at the scale
of that unit's estimation uncertainty. Our distribution estimator and quantile estimator are, respectively,
\begin{equation}\label{eq:Fhat}
 \begin{aligned}
 \widehat F(x)
   &=\frac1n\sum_{i=1}^n
     \Phi\left(\frac{x-\widehat\theta_i}{\widehat s_i}\right),\qquad \widehat q
   =\widehat F^{-1}(\tau)
     =\inf\{x:\widehat F(x)\ge\tau\},
 \end{aligned}
\end{equation}
where $\Phi$ is the standard normal cumulative distribution function. For positive $\widehat s_i$, $\widehat F$ is continuous and strictly
increasing, so $\widehat q$ is equivalently the unique solution to
$\widehat F(\widehat q)=\tau$, which can be computed by a
one-dimensional root search.

A precisely estimated coefficient (small $\widehat s_i$)
contributes almost an indicator; a less precisely estimated
coefficient (larger $\widehat s_i$) contributes a more gradual
transition. Using $\widehat s_i$ determines the smoothing scale from
the precision of each coefficient estimate, without a separate
bandwidth choice.
Smoothing changes how the quantile is estimated, while the target
remains the latent order statistic $q_n$ in \eqref{eq:target}.
Estimating this coefficient value does not require identifying
the unit occupying the corresponding rank.

\subsection{Distinguishing population targets}
\label{sec:settings}

Although the estimator in \eqref{eq:Fhat} is computed from the observed units, the population it is intended to describe must be specified separately. This distinction explains why the same point estimate
can lead to different uncertainty calculations.

\paragraph{Observed population and superpopulation.}
When units are sampled from a broader population with coefficient
distribution $F$, the researcher may target either the
observed-population quantile $q_n$, or the superpopulation quantile $q=F^{-1}(\tau)$. The same estimator $\widehat q$ can be used for both targets, but its convergence rate and the appropriate inference depend on the question being asked. To see why, write
\begin{equation}\label{eq:targets}
 \widehat q-q
 =\underbrace{\widehat q-q_n}_{\text{estimation for the observed units}}
  +\underbrace{q_n-q}_{\text{variation from sampling $n$ units}}.
\end{equation}
The first term on the right hand side arises because the coefficients of the observed units must be estimated. The second arises because a different collection of
sampled units would generally have a different quantile.

Under independent sampling from $F$ with a positive, continuous density $f$ at $q$, the second term has the usual $\sqrt n$ quantile limit. If first-step estimation and smoothing contribute
$\widehat q-q_n=o_p(n^{-1/2})$, then
\begin{align}\label{eq:stochastic}
 \sqrt n(\widehat q-q)
 =\sqrt n(q_n-q)+o_p(1)
 \ \Rightarrow\
 N\left(0,\frac{\tau(1-\tau)}{f(q)^2}\right).
\end{align}
Increasing $T$ improves the estimation of each coefficient but
does not eliminate variation from sampling units. Even if every
observed coefficient were known, estimation of $q$ would retain
uncertainty of order $n^{-1/2}$. The superpopulation perspective underlies the distributional
analysis of heterogeneous time-series characteristics in
\citet{okui2019,okui2020}, inference from noisy draws in
\citet{jochmans2024}, and the random-coefficient analysis of fund
measures in \citet{barras2022}. 

When the target is $q_n$, we condition on the coefficient array
$\{\theta_i\}_{i=1}^n$, so the second term in
\eqref{eq:targets} contributes no sampling uncertainty. Our joint-asymptotic theory provides conditions on this array and the first-step estimation under which we obtain
\[
 \widehat q-q_n=O_p\left(n^{-1/2}T^{-1/4}\right).
\]
Thus, the same estimator can have rate $\sqrt n$ for the
superpopulation quantile and rate $\sqrt{n\sqrt T}$ for the
observed-population quantile, under the respective conditions.

\addtocounter{example}{-1}
\begin{example}[continued]
The interpretation of the observed-
population quantile $q_n$ in Example \ref{ex1} is simple. 
If the retailer repeated the experiment, different customers would receive the offer, but the same stores would enter $q_n$, since the population consists of all stores the retailer operates. Moreover, the observed-population formulation applies when the observed units
exhaust the population of interest or when interest focuses on a specified collection of units. In the retailer experiment, the retailer
may study every store in its chain or focus on the stores included in a planned rollout. In either case, the target is $q_n$ for those
stores, and uncertainty comes from estimating their effects.
Random assignment of coupons within stores does not introduce
variation from sampling stores.

If, instead, the participating stores are randomly sampled from
the whole market and the researcher seeks the effect quantile
in that population, the target is $q$. Observing a different set
of stores would then change the estimated quantile even if each
store's effect were known (i.e., $T=\infty$). The appropriate setting therefore depends
on which population the researcher wants to describe.
\end{example}

\paragraph{Fixed $n$ and growing observed populations.}
Within the observed-population perspective, both fixed-$n$ and
growing-$n$ analyses hold the coefficient array
$\{\theta_i\}_{i=1}^n$ fixed under repeated sampling and target
$q_n$. They differ in how first-step uncertainty compares with
the differences between nearby coefficients.

With fixed $n$, write $k=\lceil n\tau\rceil$ and let $i_k$ denote the unit with
$\theta_{i_k}=q_n$.  Suppose that
the target coefficient differs from every other coefficient by
a fixed, positive amount, i.e., $\min_{j\ne i_k}|\theta_j-\theta_{i_k}|>0$. As $T\to\infty$, estimation errors
shrink, and the $k$th smallest estimate belongs to unit $i_k$
with probability approaching one. Under first-step asymptotic
normality, the raw estimated order statistic therefore satisfies
\[
 \sqrt T(\widehat\theta_{(k)}-q_n)
 \ \Rightarrow\ N(0,\sigma_{i_k}^2),
\]
where $\sigma_{i_k}^2$ is that unit's first-step asymptotic variance.

When $n$ grows, many coefficients can instead satisfy
$|\theta_j-\theta_{i_k}|\lesssim T^{-1/2}$, so their differences from
the target remain comparable to first-step estimation errors.
Their estimated ordering can change across samples, and joint asymptotics
capture how these nearby coefficients jointly affect quantile
estimation.

\addtocounter{example}{-1}
\begin{example}[continued]
The two perspectives are also complementary. For the retailer experiment, these perspectives describe different levels of difficulty in ordering the store effects. With a fixed number
of stores and distinct effects, observing more customers eventually makes their ordering clear. In a large retail chain, however, many stores may have similar
effects that remain difficult to distinguish even with large
numbers of customers in each store. Growing-$n$ asymptotics describe how these
nearby effects jointly determine quantile estimation. Both
perspectives retain the effects of the retailer's own stores as
the population of interest.
\end{example}

\section{Variance estimation and asymptotics}\label{sec:inference}
Even with the coefficient array fixed, the first-step estimation in equation \eqref{eq:thetahat} makes
$\widehat q$ random. Its variance is driven mainly by coefficients close to $q_n$ relative to their first-step estimation uncertainty $\widehat{s}_i$. We first construct a variance estimator for $\widehat{q}$, then  show that the resulting studentized statistic
converges to a standard normal distribution.

\subsection{Variance and its estimation}
\label{sec:variance_logic}

Recall that $\widehat F(\widehat q)=\tau$. To study the variation
in $\widehat q$, we seek to expand this equation around the quantile
of an approximation to the mean empirical CDF. Define
\begin{equation}\label{eq:H}
 H_n(x)=\frac1n\sum_{i=1}^n
 \Phi\left(\frac{x-\theta_i}{\sqrt2s_i}\right),
 \qquad
 \tilde q_n=H_n^{-1}(\tau).
\end{equation}
Under Assumptions~\ref{ass:first}-\ref{ass:counts} below,
\(
 \E\left[\widehat F(x)\right]=H_n(x)+O(T^{-1}),
\)
uniformly within any fixed multiple of $T^{-1/2}$ of $q_n$.
The approximation retains a possible first-step bias of order
$T^{-1}$, which is negligible for the normal approximation
in Theorem~\ref{thm:wald}. Notice that both $H_n$ and $\tilde q_n$ are deterministic.
Thus, $\tilde q_n$ serves as the expansion point, while $q_n$
remains the target. The difference between them is discussed below. 

Expanding the estimating equation \eqref{eq:Fhat} around $\tilde q_n$, and replacing the estimated slope by its first-order approximation gives
\(
 0=\widehat F(\widehat q)-\tau
 \simeq
 \widehat F(\tilde q_n)-\tau
 +H_n'(\tilde q_n)(\widehat q-\tilde q_n).
\)
Since $H_n(\tilde q_n)=\tau$, rearranging yields
\begin{equation}\label{eq:rootintuition}
 \widehat q-\tilde q_n
 \simeq
 -\frac{\widehat F(\tilde q_n)-H_n(\tilde q_n)}
        {H_n'(\tilde q_n)}.
\end{equation}
Because $\tilde q_n$ and $H_n'(\tilde q_n)$ are deterministic,
the corresponding first-order variance is
\[
 \Var(\widehat q)
 =\Var(\widehat q-\tilde q_n)
 \simeq
 \frac{\Var\left(\widehat F(\tilde q_n)\right)}
      {\{H_n'(\tilde q_n)\}^2}.
\]
The numerator measures fluctuations in the estimated distribution. The denominator reflects the slope of the distribution at its quantile. Next, we discuss the estimation of the numerator and denominator of $\Var(\widehat q)$ above.

\paragraph{Estimating the numerator.}
For independent units,
\begin{equation}\label{eq:cdfvariance}
 \Var\left(\widehat F(x)\right)
 =
 \frac1{n^2}\sum_{i=1}^n
 \Var\left(
 \Phi\left(\frac{x-\widehat\theta_i}{\widehat s_i}\right)
 \right).
\end{equation}
In the retailer experiment in Example \ref{ex1}, $x$ represents a candidate spending gain.
Stores with effects far above or below $x$ contribute little
uncertainty. Most of the
uncertainty comes from stores whose estimated effects could fall
on either side of $x$ in another experiment. The variance estimator
combines these uncertain comparisons. To estimate these variance contributions, we measure local uncertainty using
 \begin{align}\label{eq:pscore}p_i(x)=
 \Phi\left(-\left|\frac{x-\widehat\theta_i}{\widehat s_i}\right|\right).\end{align} 
 This quantity is largest when $\widehat\theta_i=x$ and decreases
as their distance grows relative to the first-step standard error.
It therefore gives greater weight to coefficients whose position
relative to $x$ is difficult to determine.
Its square has a useful variance interpretation: when
$\widehat s_i=s_i$ and $\widehat\theta_i\sim N(\theta_i,s_i^2)$,
we have 
 $\E\left[p_i(x)^2\right]
 =
 \Var\left(
 \Phi\left(\frac{x-\widehat\theta_i}{s_i}\right)
 \right).$
Summing the squared uncertainty measures at $\widehat q$ estimates the numerator
\begin{equation}\label{eq:variance}
 \widehat v_F=\frac1{n^2}\sum_{i=1}^n p_i(\widehat q)^2.
\end{equation}
Exact normality is used
here to explain the construction; the theory below controls
the approximation error for general coefficient estimators and
estimated standard errors.

\paragraph{Estimating the denominator.}
The denominator is the squared slope $\{H_n'(\tilde q_n)\}^2$.
We estimate the slope by differentiating $\widehat F$:
\begin{equation}\label{eq:denoslope}
 \widehat F'(x)=\frac1n\sum_{i=1}^n\frac1{\widehat s_i}
 \phi\left(\frac{x-\widehat\theta_i}{\widehat s_i}\right).
\end{equation}
Here $\phi$ is the standard normal density. With $\widehat s_i=s_i$ and $\widehat\theta_i\sim N(\theta_i,s_i^2)$, we have $\E\left[\widehat F'(x)\right]=H_n'(x)$. More generally, the theory controls the approximation error
and the random fluctuations, and justifies evaluation at $\widehat q$.
Combining the estimated numerator and denominator in \eqref{eq:variance} and \eqref{eq:denoslope} gives the variance estimator and standard error
\begin{equation}\label{eq:se}
\widehat{V}=\frac{\widehat v_F}{(\widehat F'(\widehat q))^2},\qquad \widehat{\mathrm{se}}(\widehat q)
 =
 \sqrt{\widehat{V}}
 =
 \frac{\sqrt{\sum_{i=1}^n p_i(\widehat q)^2}}
 {\sum_{i=1}^n\widehat s_i^{-1}
   \phi\{(\widehat q-\widehat\theta_i)/\widehat s_i\}}.
\end{equation}
The calculation requires only the coefficient estimates
$\widehat\theta_i$ and their standard errors $\widehat s_i$
from the first step, since $\widehat q$ is computed from
these same quantities. The order of the standard error reflects the number of coefficients
close to the quantile.
About $n/\sqrt T$ coefficients lie within a $T^{-1/2}$ neighborhood
of the quantile and contribute nontrivially to the variance.
Independence across units gives
\[
 \Var\left(\widehat F(\tilde q_n)\right)
 =
 \frac1{n^2}\sum_{i=1}^n
 \Var\left(
 \Phi\left(\frac{\tilde q_n-\widehat\theta_i}{\widehat s_i}\right)
 \right)
 \asymp
 \frac1{n^2}\frac{n}{\sqrt T}
 =
 \frac1{n\sqrt T}.
\]
Provided the slope $H_n'(\tilde q_n)$ remains bounded above
and away from zero, the standard error of $\widehat q$ is
therefore of order $n^{-1/2}T^{-1/4}$.

\subsection{Centering at the target quantile}

For the observed-population target $q_n$, the relevant estimation
error is the first term in \eqref{eq:targets}, $\widehat q-q_n$.
The second term, $q_n-q$, enters only when the target is the
superpopulation quantile $q$. To study $\widehat q-q_n$, recall
that the expansion in \eqref{eq:rootintuition} is centered at
$\tilde q_n$, the $\tau$-quantile of $H_n$. Adding and subtracting
this center gives
\begin{equation}\label{eq:centering}
 \widehat q-q_n
 =
 \underbrace{\widehat q-\tilde q_n}_{\text{estimation error around $\tilde q_n$}}
 +
 \underbrace{\tilde q_n-q_n}_{\text{deterministic centering error}}.
\end{equation}
The standard error in \eqref{eq:se} measures the variability of the first term in \eqref{eq:centering}. The second term is deterministic conditional on the coefficient array. Thus, the normal approximation
around $\tilde q_n$ carries over to $q_n$ when the two differ by
less than the order of the standard error:
\[
 |\tilde q_n-q_n|=o\left(n^{-1/2}T^{-1/4}\right).
\]
When the local CDF slope is bounded above and away from zero,
this requirement can equivalently be expressed in terms of
$H_n(q_n)-\tau$. The following assumption states this condition
on the CDF scale.
\begin{assumption}[Local balance]\label{ass:balance}
 $|H_n(q_n)-\tau|=o(n^{-1/2}T^{-1/4}).$
\end{assumption}

Assumption~\ref{ass:balance} concerns the local balance of the
deterministic coefficient array relative to the estimation error. Under the growth conditions below, the assumption holds for a deterministic quantile grid with common precision across units, provided the generating distribution has a continuously differentiable, strictly positive density near the target quantile. A smooth coefficient function evaluated
on a fixed grid provides another sufficient case. This case is developed
in Section SE.1 of the Supplementary Materials and motivates the
interpretation in Section~\ref{sec:empirical}. These are sufficient examples: the assumption restricts aggregate
local imbalance without prescribing individual coefficient positions.
Local balance requires positive and negative smoothing errors near the target quantile to offset each other sufficiently within the
fixed coefficient array. Even with i.i.d.\ coefficient draws, however, the realized array need not
satisfy local balance with probability approaching one. Without additional
assumptions, local balance can be difficult to assess from noisy coefficient estimates. Section~\ref{sec:conservative}
discusses this difficulty, establishes the trade-off involved
in removing local balance, and develops a conservative alternative.

\subsection{Conditions and main theory}\label{sec:conditions}
We consider joint asymptotics with $n,T\to\infty$. Let $c$ and $C$ denote positive, finite constants, uniform over
units and sample sizes.
The next two conditions describe the accuracy of each coefficient
estimate and how many coefficients lie near $q_n$.

\begin{assumption}[A first-order expansion]\label{ass:first}
$\{\bm Z_{it}\}_{t=1}^T$
are independent across $i$. For each $i$,
\begin{equation}\label{eq:linearization}
 \widehat\theta_i-\theta_i=\frac1T\sum_{t=1}^T\psi_{it}+R_i,
 \qquad \E\left[R_i^2 \right]\le C/T^2.
\end{equation}
For each unit $i$, the contributions of the influence function, $\psi_{it}=\psi_i(\bm Z_{it})$, are independent across $t$, with
$\E[\psi_{it}]=0$, $\E|\psi_{it}|^4\le C$, and
$\sigma_i^2=T^{-1}\sum_t \E\left[\psi_{it}^2 \right]$ satisfies $c\le\sigma_i\le C$. Its estimator satisfies
$c\le\widehat\sigma_i\le C$ and
$\E\left[(\widehat\sigma_i-\sigma_i)^2 \right]\le C/T$.
\end{assumption}

Assumption~\ref{ass:first} requires an influence-function expansion
with a remainder of root-mean-square order $T^{-1}$ and a
 consistent estimate of the first-step standard deviation.
These are some regularity conditions, although the uniform mean-square remainder bound is
stronger than ordinary asymptotic linearity, and the fixed bounds
on $\widehat\sigma_i$ may require truncation.
Proposition~\ref{prop:gmm} in Appendix~\ref{sec:primitive} verifies
these conditions for a class of GMM.

\begin{assumption}[Two-sided local counts]\label{ass:counts}
$n/\sqrt T\to\infty$. For all sufficiently large samples,
\begin{equation}\label{eq:twosides}
 \begin{split}
 c\frac n{\sqrt T}\le
 \#\{i:q_n-T^{-1/2}\le\theta_i<q_n\}\le C\frac n{\sqrt T},\\
 c\frac n{\sqrt T}\le
 \#\{i:q_n<\theta_i\le q_n+T^{-1/2}\}\le C\frac n{\sqrt T}.
 \end{split}
\end{equation}
For some fixed $\epsilon>0$,
\begin{equation}\label{eq:uppercounts}
 \#\{i:|\theta_i-q_n|\le x\}\le Cnx,
 \qquad T^{-1/2}\le x\le\epsilon.
\end{equation}
\end{assumption}

Condition~\eqref{eq:twosides} imposes the familiar local concentration
of coefficients associated with a positive, continuous density
near the target quantile. 
Irregular spacing, local gaps, and ties of controlled multiplicity
are therefore allowed. The upper bound in \eqref{eq:uppercounts}
prevents excessive concentration in larger neighborhoods of $q_n$.
For i.i.d. coefficient draws with a positive and continuous density, Assumption \ref{ass:counts} holds with probability tending to one when
$n/\sqrt T\to\infty$, as shown in
Section SE.2 of the Supplementary Materials.

\begin{proposition}[Centering and consistency]\label{prop:root}
Under Assumptions~\ref{ass:first}-\ref{ass:counts},
\begin{align}\label{eq:consistencyrate}
 \widehat q-\tilde q_n=O_p(n^{-1/2}T^{-1/4}+T^{-1}),\qquad \tilde q_n-q_n=O(T^{-1/2}).
\end{align}
\end{proposition}

Proposition~\ref{prop:root} bounds the estimation error
$\widehat q-\tilde q_n$, and
the first term matches the order of the standard error
in \eqref{eq:se}.
The additional $T^{-1}$ term reflects first-step approximation
errors, including bias, and is negligible relative to the standard
error when $n/T^{3/2}\to0$. When $n/T^{3/2}\to c\in(0,\infty)$,
however, this term is of the same order as the standard error,
so the leading bias can affect the normal approximation, which is studied below.

Proposition~\ref{prop:root} also bounds the deterministic centering
error by $O(T^{-1/2})$, establishing consistency of $\widehat q$
for $q_n$. This bound alone is insufficient for normal inference,
since it allows centering errors larger than the standard error. Assumption~\ref{ass:balance} strengthens
this bound by requiring the centering error to be negligible. The following theorem first establishes
normality around $\tilde q_n$, then uses local balance to center the
limit at $q_n$.

\begin{theorem}[Normal approximation]\label{thm:wald}
Under Assumptions~\ref{ass:first}-\ref{ass:counts} and
$n/T^{3/2}\to0$,
\[
 \frac{\widehat q-\tilde q_n}
      {\widehat{\mathrm{se}}(\widehat q)}
 \Rightarrow N(0,1).
\]
If, in addition, Assumption~\ref{ass:balance} holds, then
\(
 \frac{\widehat q-q_n}
      {\widehat{\mathrm{se}}(\widehat q)}
 \Rightarrow N(0,1).
\)
\end{theorem}

The normal approximation around $\tilde q_n$ does not require
local balance. Assumption~\ref{ass:balance} provides the additional
conclusion that the same approximation holds around the target
$q_n$. Theorem~\ref{thm:wald} requires $T^{1/2}\ll n\ll T^{3/2}$.
The condition $n/\sqrt T\to\infty$, imposed in
Assumption~\ref{ass:counts}, ensures sufficiently many nearby
coefficients for a central limit theorem. The condition
$n/T^{3/2}\to0$ makes the $T^{-1}$ first-step bias negligible relative to the standard error.

\section{Bias Correction}\label{sec:bias}
\subsection{The leading bias}

The previous section provided conditions to establish an asymptotically unbiased normal approximation for the parameter of interest. Nevertheless, when $n/T^{3/2}$ converges to a positive finite constant, a first-step bias
of order $T^{-1}$ is of the same order as the estimation error
$n^{-1/2}T^{-1/4}$. We now identify its leading
contribution and construct a feasible correction, extending the growth
range while retaining the variance estimator developed above.

The correction accounts for three features of first-step estimation:
mean bias, departures from normality, and dependence between the
coefficient estimate and its estimated standard deviation. Let
$b_i/T$ denote the leading mean bias of the first-step estimator $\widehat\theta_i$, and
let $\gamma_i$ summarize the effects of skewness, estimator
curvature, and estimation of its standard deviation. Their
formulas are given in \eqref{eq:bcconstants} below. To express how
these effects change the mean of $\widehat F$, write the
contribution of unit $i$ to the slope of $H_n$ as
\[
 h_i(x)=\frac1{\sqrt2s_i}
 \phi\left(\frac{x-\theta_i}{\sqrt2s_i}\right),
 \qquad H_n'(x)=\frac1n\sum_i h_i(x).
\]
Under the strengthened first-step condition below, we obtain
\begin{equation}\label{eq:biascdf}
 \E\left[\widehat F(x)\right]
 =H_n(x)-\frac1{nT}\sum_i
 \left\{b_i h_i(x)+\frac{\gamma_i}{6T}h_i''(x)\right\}
 +o(T^{-1}),
\end{equation}
uniformly in $x$ within any fixed multiple of $T^{-1/2}$ of $q_n$. The expansion describes the bias in the estimated distribution.
Dividing by its slope, as in \eqref{eq:rootintuition}, converts
this bias into a shift in the quantile. The leading bias of
$\widehat q-\tilde q_n$ is therefore $B_n/T$, where
\begin{equation}\label{eq:BN}
 B_n=
 \frac{n^{-1}\sum_i
 \{b_i h_i(\tilde q_n)+\gamma_i h_i''(\tilde q_n)/(6T)\}}
 {H_n'(\tilde q_n)}.
\end{equation}
The numerator is the leading term in
$T\left[H_n(\tilde q_n)-\E\left[\widehat F(\tilde q_n)\right]\right]$; the denominator converts
it into a quantile shift. Local counts in Assumption \ref{ass:counts} imply
$B_n=O(1)$. 

\subsection{Strengthened first-step conditions}

Assumption~\ref{ass:first} controls the magnitude of the first-step
remainder but does not identify its leading contribution.
Estimating the bias requires a more explicit expansion of both
the coefficient estimator and its estimated variance.

\begin{assumptionbc}[A second-order first-step expansion]
\label{ass:second}
$\{\bm Z_{it}\}_{t=1}^T$ are i.i.d. within $i$ and
 independent across $i$.
There exist deterministic symmetric matrices $\bm D_i$ with
$\sup_i\|\bm D_i\|\le C$ and fixed-dimensional vectors
\(\bm v_{it}=\bm v_i(\bm Z_{it})\) satisfying 
$\E[\bm v_{it}]=\bm0$ and $\|\bm v_{it}\|\le C$ almost surely
for all $i,t$, whose first two coordinates are
$\psi_{it}$ and $\eta_{it}$. With
$\bm S_i=T^{-1/2}\sum_t\bm v_{it}$ and
$Q_i=\tfrac12\bm S_i^{\top}\bm D_i\bm S_i$, the estimators satisfy
\begin{align}
 \widehat\theta_i-\theta_i
 &=\frac1T\sum_t\psi_{it}+\frac{Q_i}{T}+r_i,
 & \E[r_i^2]&\le C/T^3,\label{eq:secondtheta}\\
 \widehat\sigma_i^2-\sigma_i^2
 &=\frac1T\sum_t\eta_{it}+w_i,
 & \E[w_i^2]&\le C/T^2.\label{eq:secondscale}
\end{align}
The variance is $\sigma_i^2=\E[\psi_{it}^2]$, and the positive
bounds on $\sigma_i$ and $\widehat\sigma_i$ in
Assumption~\ref{ass:first} hold. Define $\bm\Sigma_i=\E\left[\bm v_{it}\bm v_{it}^{\top}\right]$, 
 $\bm d_i=\E\left[\bm v_{it}\psi_{it}\right]$, 
\begin{equation}\label{eq:bcconstants}
 \begin{split}
 b_i&=\tfrac12\operatorname{tr}(\bm D_i\bm\Sigma_i),\qquad \gamma_i=\E\left[\psi_{it}^3\right]
          +3\bm d_i^{\top}\bm D_i\bm d_i
          +3\E\left[\psi_{it}\eta_{it}\right].
 \end{split}
\end{equation}
Bounded sample analogues $\widehat b_i$ and $\widehat\gamma_i$
satisfy
\(
 \sup_i \E\left[
 (\widehat b_i-b_i)^2+
 (\widehat\gamma_i-\gamma_i)^2
 \right]\le C/T.
\)
\end{assumptionbc}

The quadratic term $Q_i$ records the curvature of a smooth
first-step estimator, with $\E[Q_i]=b_i$.
The coefficient $\gamma_i$ also has a direct interpretation.
Its first term is the third moment of the influence contribution,
its second captures the effect of estimator curvature, and its
third accounts for dependence between coefficient estimation
and variance estimation. 

These conditions strengthen Assumption~\ref{ass:first} sufficiently
to identify and estimate the leading bias. Uniform boundedness
is sufficient for the bounds used here; it is not required of
every estimator that admits a bias correction. Higher-order
expansions for GMM are studied by \citet{newey2004}. Here we
derive how these terms affect the estimated population quantile
and show that the remaining error is negligible relative to its
standard error.
Proposition~\ref{prop:gmmbias} in Appendix~\ref{sec:primitive}
gives sufficient conditions for these expansions and for
estimating $b_i$ and $\gamma_i$ in  GMM with
deterministic weights.


\addtocounter{example}{-1}
\begin{example}[continued]
Suppose each customer independently receives the coupon with
probability $1/2$. Within store $i$, let $d=1$ denote customers
offered the coupon and $d=0$ those receiving no offer. Write
$\bar Y_{id}$, $\widehat v_{id}$, and $\widehat\mu_{3,id}$ for
the sample mean, variance, and third central moment of spending
among customers with treatment status $d$.
The average spending gain from offering the coupon in store \(i\)
can be estimated by \(\widehat\theta_i=\bar Y_{i1}-\bar Y_{i0}\),
with first-step asymptotic variance estimated by
\(\widehat\sigma_i^2=2(\widehat v_{i1}+\widehat v_{i0})\).
The scale \(\widehat\sigma_i\) is truncated to satisfy the bounds in
Assumption~\ref{ass:first}. Write $\mu_{3,id}$ for the population
third central moment in group $d$. Equation~\eqref{eq:bcconstants}
then gives
\begin{equation}\label{eq:meanbiasexample}
 \begin{aligned}
 b_i&=0,
 &\gamma_i&=16(\mu_{3,i1}-\mu_{3,i0}),& \widehat b_i&=0,
 &\widehat\gamma_i&=
 16(\widehat\mu_{3,i1}-\widehat\mu_{3,i0}).
 \end{aligned}
\end{equation}
A coupon may induce a few customers to make unusually large purchases,
while most change their spending little. This response can make
spending more skewed among customers offered the coupon. Although
the difference in means remains unbiased for each store's average
gain, the difference in skewness can bias the estimated quantiles
used to assess how broadly the promotion benefits stores.
\end{example}

\subsection{Feasible correction and inference}

The population expression in \eqref{eq:BN} involves latent
coefficients. Its components can nevertheless be estimated
by aggregating observable smoothing weights.
Let $z_i=(\widehat q-\widehat\theta_i)/\widehat s_i$ and define
\begin{equation}\label{eq:Bhat}
 \widehat B_n=
 \frac{\displaystyle\sum_i\left\{
 \widehat b_i\widehat s_i^{-1}\phi(z_i)
 +\frac{\widehat\gamma_i}{6T\widehat s_i^3}
                (z_i^2-1)\phi(z_i)\right\}}
 {\displaystyle\sum_i\widehat s_i^{-1}\phi(z_i)},
 \qquad
 \widehat q_{\mathrm{BC}}
 =\widehat q-\frac{\widehat B_n}{T}.
\end{equation}
The density weight is the same one used to estimate the CDF
slope. The second weight comes from differentiating that
kernel twice, using $\phi''(z)=(z^2-1)\phi(z)$.
Although individual weights are noisy, averaging over the
many locally relevant units consistently estimates the
quantities in \eqref{eq:BN}. The bias correction uses the same coefficient estimates and standard errors used to estimate the variance. The only additional
inputs are the unit-specific estimates $\widehat b_i$ and
$\widehat\gamma_i$.

For nominal coverage $1-\alpha$, the resulting interval is
\begin{equation}\label{eq:bcinterval}
 I_{\mathrm{BC}}=
 \left[
 \widehat q_{\mathrm{BC}}
 -z_{1-\alpha/2}\widehat{\mathrm{se}}(\widehat q),\,
 \widehat q_{\mathrm{BC}}
 +z_{1-\alpha/2}\widehat{\mathrm{se}}(\widehat q)
 \right],
 \qquad z_p=\Phi^{-1}(p).
\end{equation}
The correction changes the center only. The following theorem
shows that estimating $B_n$ introduces no additional
first-order variance in the enlarged growth range.

\begin{theorem}[Bias-corrected inference]\label{thm:bc}
Suppose Assumptions~\ref{ass:second} and \ref{ass:counts}
hold, and $n/T^{3/2}\to c\in[0,\infty)$.
Then $\widehat B_n-B_n\to_p0$, $|I_{\mathrm{BC}}|\asymp_p n^{-1/2}T^{-1/4}$, and
\(
 \frac{\widehat q_{\mathrm{BC}}-\tilde q_n}
      {\widehat{\mathrm{se}}(\widehat q)}
 \Rightarrow N(0,1).
\)
If, in addition, Assumption~\ref{ass:balance} holds, then
 $\frac{\widehat q_{\mathrm{BC}}-q_n}
      {\widehat{\mathrm{se}}(\widehat q)}
 \Rightarrow N(0,1)$ and $P(q_n\in I_{\mathrm{BC}})\to1-\alpha$.
\end{theorem}

By estimating the leading first-step bias, Theorem~\ref{thm:bc}
extends normal inference from
$n/T^{3/2}\to0$ to $n/T^{3/2}\to c\in[0,\infty)$.

\section{Conservative Inference}\label{sec:conservative}

Theorems~\ref{thm:wald} and~\ref{thm:bc} above established normal approximations around $\tilde q_n$. Local balance in Assumption~\ref{ass:balance} allows us to
use these approximations for the target $q_n$ by making the
difference $\tilde q_n-q_n$ negligible relative to the standard error. A natural question is what we can conclude about $q_n$ without this condition.
Section~\ref{sec:impossibility} shows that the sampling
distribution centered at
$q_n$ cannot
be estimated uniformly consistently, even when first-step bias is negligible. Section~\ref{sec:bound_algorithm} constructs a conservative interval for $q_n$ without estimating this distribution.

\subsection{Inference without local balance: an impossibility result}\label{sec:impossibility}

Let $\mathcal F$ denote the class of all data-generating processes satisfying Assumptions~\ref{ass:first}-\ref{ass:counts}
uniformly, with constants $c,C,\epsilon$ that do not
depend on $F\in\mathcal F$.
For each $F\in\mathcal F$, let $P_{nT}$ denote probability
under $F$ and $q_n(F)$ the observed-population quantile of
the latent coefficients.  Let $\|G-H\|_\infty=\sup_{z\in\mathbb R}|G(z)-H(z)|$ and write the sampling distribution of the studentized estimator as
\begin{equation}\label{eq:samplinglaw}
 L_{nT,F}(z)
 =P_{nT}\!\left\{
 \frac{\widehat q-q_n(F)}{\widehat{\mathrm{se}}(\widehat q)}
 \le z\right\}.
\end{equation}
Under each $F$, the coefficient array is held fixed, and
sampling variation arises from repeated observations on
the same units.

\begin{proposition}[Impossibility of uniform distribution estimation]
\label{prop:impossibility}
Fix $\tau\in(0,1)$ and let $n,T\to\infty$ with
$ n/ T^{3/2}\to 0$.
There exists $\varepsilon>0$ such that every data-based estimator
$\widehat L_{nT}$ of $L_{nT,F}$ satisfies
\begin{equation}\label{eq:distributionimpossibility}
 \liminf_{n,T\to\infty}\sup_{F\in\mathcal F}
 P_{nT}\!\left\{
 \|\widehat L_{nT}-L_{nT,F}\|_\infty>\varepsilon
 \right\}\ge\frac12.
\end{equation}
\end{proposition}

The proof compares two coefficient arrays whose observation laws
become indistinguishable, but whose studentized errors converge
to $N(0,1)$ and $N(-1,1)$, respectively. Moving a single coefficient
within a local gap changes the target by one quantile standard
error, while changing the data distribution by a vanishing amount.
One array satisfies local balance; the other does not.
Thus no analytical or resampling method can consistently recover
the sampling distribution uniformly over $\mathcal F$. 

\begin{remark}
A direct implication of the proof is that, under Assumptions~\ref{ass:first}
and~\ref{ass:counts} alone, no data-based rule can uniformly distinguish
negligible local imbalance from imbalance of the same order
as the standard error. Therefore, without additional restrictions
on the deterministic coefficient array, a pretest cannot
provide a uniformly reliable assessment of whether
Assumption~\ref{ass:balance} holds.
\end{remark}

\paragraph{Practical guidance.} For empirical work, this result suggests choosing the procedure
according to the restrictions supported by the design or a substantive
model. When these restrictions justify Assumption~\ref{ass:balance},
normal inference is available under the respective conditions of
Theorems~\ref{thm:wald} and \ref{thm:bc}: use $I_{\mathrm U}$ when
the first-step contribution is negligible relative to the standard
error, and $I_{\mathrm{BC}}$ when $n$ is large and the
strengthened first-step conditions hold.

When the justification for local balance is uncertain, inference can
instead be based on a bound for the centering error, without estimating
its contribution to the sampling law. Reporting the resulting
conservative interval $I_{\mathrm{C}}$ developed below alongside the normal interval also identifies
which conclusions remain supported without local balance, under the
remaining conditions. 
The next subsection develops this conservative approach.

\subsection{A conservative interval}
\label{sec:bound_algorithm}

The quantity $p_i(x)$ in \eqref{eq:pscore} measures uncertainty
about whether a coefficient lies below $x$. We use its average,
$\bar p(x)=n^{-1}\sum_i p_i(x)$, to bound the difference between
the mean estimated CDF and the latent distribution.
Under Assumptions~\ref{ass:first}-\ref{ass:counts}, adding and
subtracting twice this average yields the following bounds for all
$x$ within $K/\sqrt{T}$ of $q_n$, and any fixed $K>0$:
\[
 \begin{aligned}
 \E\left[\widehat F(x)-2\bar p(x)\right]&\le F_n(x-)+C/T,
 \qquad \E\left[\widehat F(x)+2\bar p(x)\right]\ge F_n(x)-C/T,
 \end{aligned}
\]
where $F_n(x-)=n^{-1}\sum_i\ind\{\theta_i<x\}$.
These bounds account for the mean discrepancy without requiring
positive and negative smoothing errors to cancel. To relate
these expectation bounds to their observable counterparts,
we also establish, for every fixed $K$ and both signs,
\begin{equation}\label{eq:conservative_uniform}
 \sup_{|x-q_n|\le K/\sqrt T}
 \left|\{\widehat F(x)\pm2\bar p(x)\}
       -\E\left[\widehat F(x)\pm2\bar p(x)\right]\right|
 =o_p(T^{-1/2}),
\end{equation}
uniformly over $\mathcal F$. Since
$F_n(q_n-)<\tau\le F_n(q_n)$, these results motivate
retaining a candidate value $x$ for $q_n$ whenever
\begin{equation}\label{eq:accept}
 \widehat F(x)-2\bar p(x)\le\tau\le\widehat F(x)+2\bar p(x),
 \quad\text{equivalently}\quad
 |\widehat F(x)-\tau|\le2\bar p(x).
\end{equation}

The estimate $\widehat q$ always satisfies \eqref{eq:accept}
because $\widehat F(\widehat q)=\tau$. Starting from $\widehat q$,
we extend the accepted range in both directions until the
condition first fails. Formally, define
 \begin{align}\label{eq: conservativeinterval}
 L_{\mathrm C}&=\sup\{x<\widehat q:\widehat F(x)+2\bar p(x)<\tau\},
 \qquad U_{\mathrm C}=\inf\{x>\widehat q:\widehat F(x)-2\bar p(x)>\tau\}.
 \end{align}
The resulting confidence interval is
 $I_{\mathrm C}=[L_{\mathrm C},U_{\mathrm C}]$.
Continuity gives
$\widehat F(L_{\mathrm C})+2\bar p(L_{\mathrm C})=\tau$ and
$\widehat F(U_{\mathrm C})-2\bar p(U_{\mathrm C})=\tau$.
We solve the first equation to the left of $\widehat q$
and the second to its right. In each direction, the endpoint
is the first root beyond which \eqref{eq:accept} no longer holds.
The resulting interval need not be symmetric around $\widehat q$.
To establish coverage, it suffices to show that, with probability
tending to one, \eqref{eq:accept} holds at every point between
$\widehat q$ and $q_n$. The search therefore cannot stop before
reaching $q_n$, ensuring that the interval contains the target.

\begin{theorem}[Uniform conservative coverage]
\label{thm:bound}
For the class $\mathcal F$ used in
Proposition~\ref{prop:impossibility}, we have
$\inf_{F\in\mathcal F}P_{nT}\{ q_n(F)\in I_{\mathrm C}\}\to1$.
Moreover, $|I_{\mathrm C}|=O_p(T^{-1/2})$ uniformly over $\mathcal F$.
\end{theorem}

Theorem~\ref{thm:bound} establishes that $I_{\mathrm C}$ covers the target
with probability tending to one uniformly over $\mathcal F$.
Thus, removing local balance still permits conservative inference,
with interval length of order $T^{-1/2}$. This is a slower
length rate than the $n^{-1/2}T^{-1/4}$ rate of $I_{\mathrm{BC}}$,
since $T^{1/2}\ll n$.  Section~\ref{sec:lengthlower} of the Supplementary Materials
provides a matching lower bound: uniform coverage over
$\mathcal F$ requires interval length of at least order
$T^{-1/2}$ with positive probability.\footnote{The lower
bound requires only asymptotic coverage of at least the nominal
level uniformly over $\mathcal F$. It therefore also applies
to $I_{\mathrm C}$, whose coverage tends uniformly to one.} Together, these results establish that $I_{\mathrm C}$
attains the optimal worst-case length rate in probability.
Section SE.3 of the Supplementary Materials extends conservative
inference to serially dependent observations.

\begin{remark}
Proposition~\ref{prop:root} and $n/\sqrt T\to\infty$ imply
$\widehat q-q_n=O_p(T^{-1/2})$. Thus, for any fixed $c>0$,
a straightforward interval like
\(
\left[\widehat q-\frac{c\log T}{\sqrt T},
      \widehat q+\frac{c\log T}{\sqrt T}\right]
\)
requires no additional estimation and has coverage tending to one. However, its radius requires an
arbitrary choice of $c$ and does not adapt to the estimated
uncertainty. Its length also contains an additional $\log T$
factor relative to the $O_p(T^{-1/2})$ length of $I_{\mathrm C}$.
\end{remark}

\section{Monte Carlo evidence}\label{sec:simulation}
\subsection{Design and methods}
The experiments examine how first-step bias and local
imbalance affect coverage of the target $q_n$. We first
vary $n$ and $T$ under balance to assess bias correction. We then compare
rejection frequencies under balanced and unbalanced arrays. The target
is the fixed-population median $q_n=\theta_{(\lceil n/2\rceil)}$.
Each design uses 5,000 repetitions, holding the coefficient array fixed.

We generate observations from
\[
 Y_{it}=\theta_i+\varepsilon_{it},\qquad
 \varepsilon_{it}=\frac{B_{it}-0.2}{0.4},\qquad
 B_{it}\stackrel{\mathrm{iid}}{\sim}\operatorname{Bernoulli}(0.2).
\]
Thus $\E[\varepsilon_{it}]=0$, $\E[\varepsilon_{it}^2]=1$, and
$\E[\varepsilon_{it}^3]=1.5$. We estimate $\theta_i$ by the sample mean
$\bar Y_i=T^{-1}\sum_{t=1}^T Y_{it}$ and its standard error
by $\widehat s_i=\widehat\sigma_i/\sqrt T$, where
\(
 \widehat\sigma_i
 =\min\left\{2,\max\left\{0.5,
 \sqrt{\frac1T\sum_{t=1}^T(Y_{it}-\bar Y_i)^2}
 \right\}\right\}.
\)
Thus, $\widehat\sigma_i$ is the sample standard deviation truncated to $[0.5,2]$, as required by the theory. The correction uses
$\widehat b_i=0$ and
$\widehat\gamma_i=4T^{-1}\sum_t(Y_{it}-\bar Y_i)^3$.

For odd $n$, put $k=(n-1)/2$ and
$\theta_i=\Phi^{-1}\{0.005+0.99(i-1)/(n-1)\}, \  i=1,\ldots,n.$
The \emph{balanced} array is $(\theta_i)_{i=1}^n$.
The \emph{reflected} array consists of zero and the pairs
$(-|Z_j|,|Z_j|)$, $j=1,\ldots,k$, where the standard normal draws
are generated once for each population size and then held fixed.
The \emph{unbalanced} array replaces each positive $\theta_i$ by $2\theta_i$, while keeping the negative $\theta_i$ unchanged. 
All three have $q_n=0$, but only the first two satisfy exact median
balance.  The reflected and unbalanced designs are hence examined to assess the role of local balance.

A fourth, \emph{balanced array with local concentrations} makes the first-step
correction visible without violating local balance. Let  $m=\lfloor n/\sqrt T+1/2\rfloor$.
Set one coefficient equal to zero, $m$ coefficients equal to
$-1/(2\sqrt T)$, and another $m$ equal to $1/(2\sqrt T)$.
The remaining $2(k-m)$ coefficients form the symmetric pairs
\[
 \left(-\frac{j}{k-m+1},\,\frac{j}{k-m+1}\right),
 \qquad j=1,\ldots,k-m.
\]
Thus, the population remains symmetric around zero, with
$2m$ coefficients concentrated close to the median and the
remaining coefficients evenly spaced on either side. The local concentrations  make the first-step bias more pronounced and
the effect of bias correction easier to assess.
We use $T\in\{50,100,200\}$ and
$n=2\lfloor cT^{3/2}/2\rfloor+1$ for $c\in\{1/4,1,4\}$,
so that $n$ is odd and the median corresponds to a single order
statistic. These designs examine the enlarged growth range of
Theorem~\ref{thm:bc}.

We compare five procedures. The uncorrected interval $I_{\mathrm U}$
and the bias-corrected interval $I_{\mathrm{BC}}$ use the standard error in equation \eqref{eq:se}, with centers $\widehat q$ and $\widehat q_{\mathrm{BC}}$, respectively. The conservative interval $I_{\mathrm C}$ is given in \eqref{eq: conservativeinterval}.
The projection interval $I_{\mathrm P}$ uses Bonferroni-adjusted
coefficient intervals: if all coefficients are covered, their
ordered lower and upper endpoints also cover the population
quantile. Related simultaneous-region methods for rank inference
are studied by \citet{klein2020,mogstad2024}.
The stochastic-variance reference $I_{\mathrm S}$ uses the
superpopulation CLT in \eqref{eq:stochastic}.
Specifically, with a significance level $\alpha=0.05$,
\begin{align}\label{eq:simulationreferences}
 I_{\mathrm U}
 &=\bigl[\widehat q-z_{1-\alpha/2}\widehat{\mathrm{se}}(\widehat q),\,
         \widehat q+z_{1-\alpha/2}\widehat{\mathrm{se}}(\widehat q)\bigr], \quad
 I_{\mathrm{BC}} =\bigl[\widehat q_{\mathrm{BC}}-z_{1-\alpha/2}\widehat{\mathrm{se}}(\widehat q),\,
         \widehat q_{\mathrm{BC}}+z_{1-\alpha/2}\widehat{\mathrm{se}}(\widehat q)\bigr], \notag\\
 I_{\mathrm S}
 &=\left[\widehat q-z_{1-\alpha/2}\tfrac{1}{2\sqrt n\,\widehat F'(\widehat q)},\,
         \widehat q+z_{1-\alpha/2}\tfrac{1}{2\sqrt n\,\widehat F'(\widehat q)}\right], \quad
 I_{\mathrm C}=[L_{\mathrm C},U_{\mathrm C}], \\
 I_{\mathrm P}
 &=[L_{(\lceil n/2\rceil)},U_{(\lceil n/2\rceil)}],
 \qquad
 \left\{
 \begin{aligned}
  L_i&=\widehat\theta_i-z_{1-\alpha/(2n)}\widehat s_i,\\
  U_i&=\widehat\theta_i+z_{1-\alpha/(2n)}\widehat s_i.
 \end{aligned}
 \right. \notag
\end{align}
The lower and upper endpoints in $I_{\mathrm P}$ are ordered
separately and need not correspond to the same unit.

\subsection{Size and the role of bias correction}
Table~\ref{tab:simulation} reports rejection frequencies under
the null. Panel A considers a regular grid with
$n\in\{101,401,1001\}$ and $T\in\{250,500,1500\}$.
For $n=401$ and $1001$, both normal intervals reject between
0.050 and 0.059, close to the nominal level of 0.05.
Bias correction makes little difference in these designs.

The results for $n=101$ illustrate why increasing $T$ alone
need not improve coverage. More precise coefficient estimates
also mean a narrower $T^{-1/2}$ neighborhood of the target. The normal approximation
developed here relies on averaging over sufficiently many such
coefficients, whose number is governed by $n/\sqrt T$. For example,
at $(n,T)=(101,1500)$, the ratio $n/\sqrt T$ is relatively small at 2.61,
and both normal intervals reject with frequency 0.071.
Thus, its accuracy depends on the relative growth of $n$ and $T$,
rather than either sample dimension alone. 

\begin{table}[!t]\centering
\setlength{\tabcolsep}{5pt}
\caption{Rejection frequencies under the null}\label{tab:simulation}
\begin{tabular}{rrrccccc}\toprule
\multicolumn{2}{c}{Sample size}
& \multicolumn{1}{c}{Rate}
& \multicolumn{5}{c}{Rejection frequency}\\
\cmidrule(lr){1-2}\cmidrule(lr){3-3}\cmidrule(lr){4-8}
$n$ & $T$ & $n/T^{3/2}$
& $I_{\mathrm U}$ & $I_{\mathrm{BC}}$ & $I_{\mathrm C}$
& $I_{\mathrm S}$ & $I_{\mathrm P}$\\
\midrule
\multicolumn{8}{l}{\textit{A. Balanced array}}\\
101 & 250 & 0.026 & 0.069 & 0.070 & 0.001 & 0.000 & 0.000\\
401 & 250 & 0.101 & 0.054 & 0.054 & 0.000 & 0.000 & 0.000\\
1,001 & 250 & 0.253 & 0.056 & 0.054 & 0.000 & 0.000 & 0.000\\
101 & 500 & 0.009 & 0.062 & 0.063 & 0.003 & 0.000 & 0.000\\
401 & 500 & 0.036 & 0.053 & 0.053 & 0.000 & 0.000 & 0.000\\
1,001 & 500 & 0.090 & 0.050 & 0.051 & 0.000 & 0.000 & 0.000\\
101 & 1,500 & 0.002 & 0.071 & 0.071 & 0.030 & 0.000 & 0.000\\
401 & 1,500 & 0.007 & 0.059 & 0.059 & 0.000 & 0.000 & 0.000\\
1,001 & 1,500 & 0.017 & 0.053 & 0.054 & 0.000 & 0.000 & 0.000\\
\addlinespace[3pt]
\multicolumn{8}{l}{\textit{B. Balanced array with local concentrations}}\\
89 & 50 & 0.252 & 0.068 & 0.056 & 0.000 & 0.000 & 0.000\\
353 & 50 & 0.998 & 0.080 & 0.056 & 0.000 & 0.000 & 0.000\\
1,415 & 50 & 4.002 & 0.178 & 0.057 & 0.000 & 0.000 & 0.000\\
251 & 100 & 0.251 & 0.061 & 0.053 & 0.000 & 0.000 & 0.000\\
1,001 & 100 & 1.001 & 0.077 & 0.053 & 0.000 & 0.000 & 0.000\\
4,001 & 100 & 4.001 & 0.162 & 0.050 & 0.000 & 0.000 & 0.000\\
707 & 200 & 0.250 & 0.056 & 0.050 & 0.000 & 0.000 & 0.000\\
2,829 & 200 & 1.000 & 0.078 & 0.050 & 0.000 & 0.000 & 0.000\\
11,313 & 200 & 4.000 & 0.151 & 0.049 & 0.000 & 0.000 & 0.000\\
\addlinespace[3pt]
\bottomrule\end{tabular}
\par\smallskip\begin{minipage}{.98\textwidth}\footnotesize
Notes: 5,000 repetitions per design; $q_n=0$ throughout.
$I_{\mathrm U}$, $I_{\mathrm{BC}}$, $I_{\mathrm S}$, and $I_{\mathrm P}$ use
5\% critical values. $I_{\mathrm C}$ has no nominal-level parameter.
The approximate Monte Carlo standard error of a rejection frequency of 0.05 is 0.0031.
\end{minipage}\end{table}

Panel B examines the enlarged growth range in
Theorem~\ref{thm:bc}, which allows $n/T^{3/2}$ to approach
a positive finite constant. As this ratio increases, first-step
bias becomes more important relative to the standard error.
The uncorrected rejection frequencies rise with this ratio,
whereas the corrected frequencies remain close to 0.05.
For example, when $(n,T)=(4001,100)$ and $n/T^{3/2}\simeq4$, correction
reduces the rejection frequency from 0.162 to 0.050.
At $(n,T)=(11313,200)$ with the same ratio, the corresponding reduction
is from 0.151 to 0.049. Increasing both sample dimensions
therefore does not eliminate the distortion when their relative
growth leaves bias comparable to the standard error.
Estimating and removing the leading bias permits accurate
inference in precisely this case.

The conservative and projection intervals have substantially lower
rejection frequencies in these experiments.
The stochastic-reference interval also rejects infrequently
because its variance includes uncertainty from sampling units,
which is absent for the fixed observed-population target.
The power comparisons below examine how this additional
conservatism affects the ability to detect departures from
the null.

Section SC of the Supplementary Materials reports the same five
procedures at $\tau=0.25$. The bias correction also improves inference
for the lower quartile when $n/T^{3/2}$ is appreciable; for example,
at $(n,T)=(4001,100)$ it reduces rejection from 0.199 to 0.053 in
the locally balanced design. The additional experiments show greater size distortion at the lower quantile, while the overall pattern remains similar.

\subsection{Power and the role of local balance}
We examine whether the procedures can detect a nonzero population
median, both when local balance holds and when it fails.
Starting from each array with median zero, we add $\delta$ to
every coefficient, so the target becomes $q_n=\delta$.
We then test $H_0:q_n=0$, rejecting when the corresponding
interval excludes zero.

Figure~\ref{fig:power} plots rejection frequencies against
$\delta\sqrt T\in[-6,6]$ for $(n,T)=(401,250)$ and $(1001,500)$.
The horizontal axis expresses the change in the median relative
to the first-step standard error $T^{-1/2}$. At $\delta=0$,
the curves measure rejection under the true null; away from zero,
they measure power. Both configurations have relatively small
$n/T^{3/2}$, allowing us to focus on the effect of local balance.

\begin{figure}[!t]
\centering
\begingroup
\setlength{\fboxsep}{1pt}
\includegraphics[width=\textwidth]{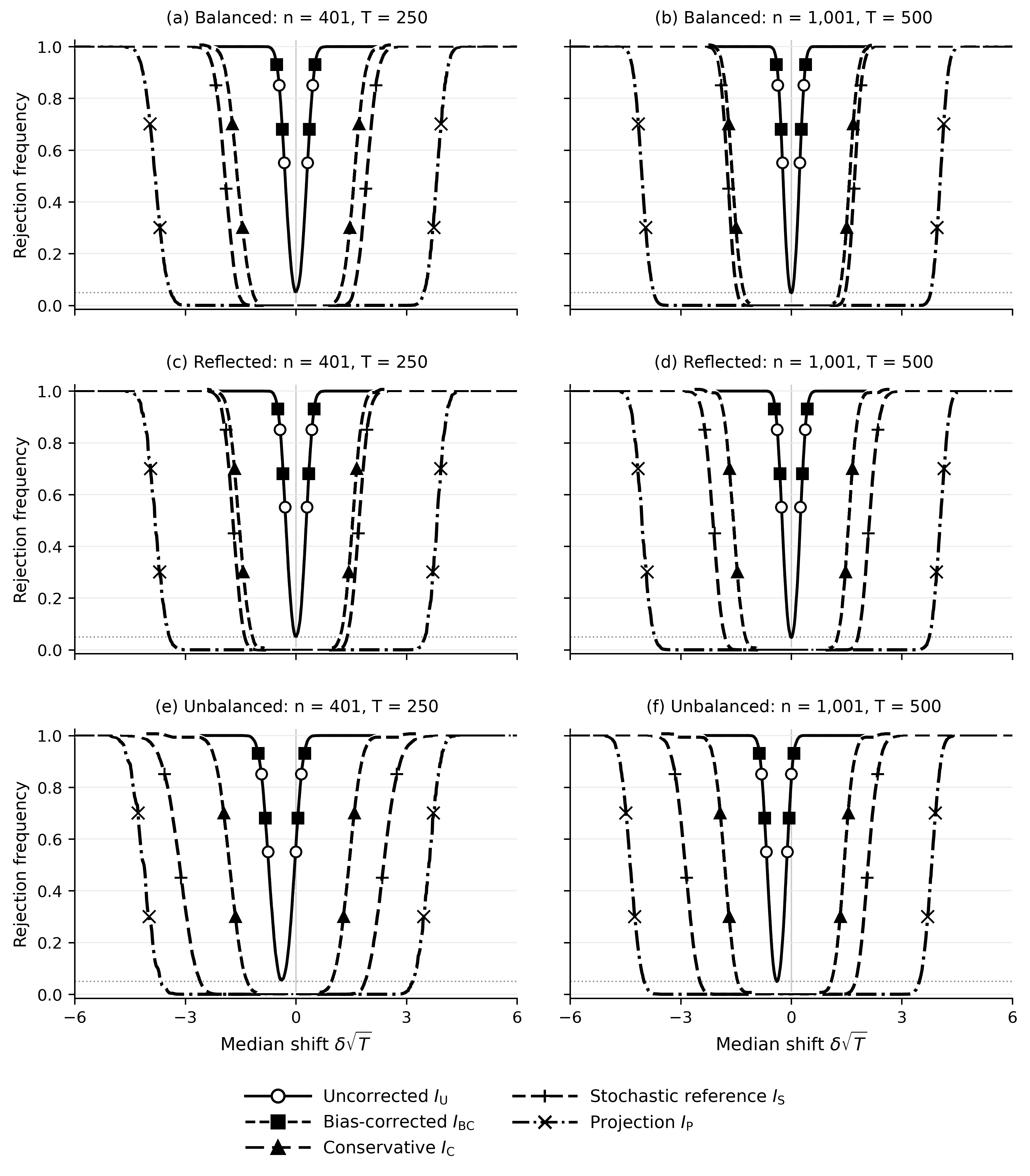}
\endgroup
\caption{Two-sided rejection frequencies for the zero-median null.
Rows show balanced, reflected, and unbalanced arrays; columns show
$(n,T)=(401,250)$ and $(1001,500)$. Each curve uses 5,000 repetitions.
The vertical line marks the null and the horizontal dotted line marks
0.05. The uncorrected and corrected curves nearly coincide in the
first two rows. }\label{fig:power}
\end{figure}

Under balance in panels (a) and (b) of Figure~\ref{fig:power}, the normal intervals combine accurate size with
the ability to detect small changes in the median. Their null
rejection frequencies range from 0.048 to 0.054, and their power
rises at substantially smaller shifts than that of
$I_{\mathrm C}$, $I_{\mathrm S}$, and $I_{\mathrm P}$.
The uncorrected and corrected curves nearly coincide because
first-step bias is small in these configurations.
The reflected design also shows that irregular spacing alone
does not undermine normal inference: symmetry in panels (c) and (d) preserves local
balance even though the coefficients do not follow a smooth grid.

When local balance fails, as shown in panels (e) and (f), the normal intervals can reject
frequently even though the null is true. At $(n,T)=(1001,500)$, for example,
the null rejection frequencies are 0.824 for $I_{\mathrm U}$
and 0.820 for $I_{\mathrm{BC}}$. Their curves reach their minima
at a negative shift rather than at zero. These high rejection frequencies therefore cannot be
interpreted as evidence of good power. The similar performance
of the two intervals illustrates why correcting first-step bias
does not resolve the centering error $\tilde q_n-q_n$ caused
by local imbalance.

The conservative interval maintains coverage under the null while
retaining useful power when local balance fails.  At $(n,T)=(1001,500)$, its rejection frequencies reach
0.808 and 0.995 at $\delta\sqrt T=-2$ and $2$, respectively,
whereas the projection interval has zero recorded rejections
at either shift. Thus, $I_{\mathrm C}$ detects changes that
the projection interval misses while avoiding the severe
overrejection of the normal intervals at the null.
 The trade-off is reduced sensitivity to small changes in the
median. When local balance holds, $I_{\mathrm C}$ generally
requires larger shifts than $I_{\mathrm U}$ and $I_{\mathrm{BC}}$
to attain comparable power, consistent with the longer intervals
needed to maintain coverage without local balance. Across the three designs, $I_{\mathrm S}$ generally requires
larger shifts than $I_{\mathrm C}$ to attain comparable rejection
frequencies, while $I_{\mathrm P}$ requires the largest shifts.\footnote{The widths of $I_{\mathrm S}$ and $I_{\mathrm C}$ are
of orders $n^{-1/2}$ and $T^{-1/2}$, respectively. Thus, neither
interval is necessarily narrower; their relative widths depend
on $n$, $T$, and the underlying population.}
These power comparisons alone, however, do not establish that
$I_{\mathrm S}$ and $I_{\mathrm P}$ are conservative.

\section{Empirical Illustration: Mortgage Performance}
\label{sec:empirical}

When mortgage delinquency declines, does it fall across credit-score levels, or does the overall decline mainly reflect a larger share of loans going to borrowers with higher credit scores? This distinction matters for interpreting changes in credit
quality. \citet{mayer2009} emphasize both economic conditions and
underwriting in explaining mortgage defaults, while \citet{davis2023}
show why mortgage risk cannot be understood from low-score borrowers
alone. We compare repayment performance in the 2007 and 2009 origination
cohorts at each credit score. Quantiles of these differences reveal
whether the improvement extends across much of the score range. 

\subsection{Data and target}

We use fixed-rate mortgages in Freddie Mac's public Single-Family
Loan-Level Dataset, obtained from the archived release of
\citet{mushava2023data}. Origination records provide credit scores,
and monthly performance records track delinquency and loan
termination. We compare mortgages originated in 2007 and 2009,
measuring serious delinquency within 12 and 24 months as a reported
delinquency of at least three months or an adverse termination.
Loans repaid before the horizon remain in the sample; missing
histories are not treated as successful repayment. Section SD of
the Supplementary Materials gives the eligibility rules and outcome
construction.

The units are the $n=181$ integer credit scores
$x_i=620+i-1$, $i=1,\ldots,n$. These scores are used in mortgage underwriting and
are observed directly. For each horizon, let $p_{2007}(x_i)$ and
$p_{2009}(x_i)$ denote the serious-delinquency probabilities for
loans originated in the two years. Define
\[
 \theta_i=p_{2007}(x_i)-p_{2009}(x_i),\qquad
 q_n=\theta_{(\lceil n\tau\rceil)}.
\]
A positive $\theta_i$ means that loans originated in 2009 were
less likely to become seriously delinquent than those originated
in 2007 at the same score. Equal weighting measures how
broadly risk declined across credit-score values, regardless of their loan shares. Thus, a positive lower quartile indicates improvement at
more than three quarters of these values. The quantiles rank delinquency reductions, rather than the credit scores themselves.

Figure~\ref{fig:creditscoreprofiles} provides an initial picture
of these changes using all eligible loans.  The 2009 profile lies
below the 2007 profile across the displayed range, with larger
absolute differences at lower scores. This suggests a broad but
uneven reduction in delinquency risk. We use quantile inference
to assess how widely the reductions extend across individual
score values, accounting for uncertainty in their estimated
delinquency probabilities.

\begin{figure}[!t]\centering
\includegraphics[width=0.7\textwidth]{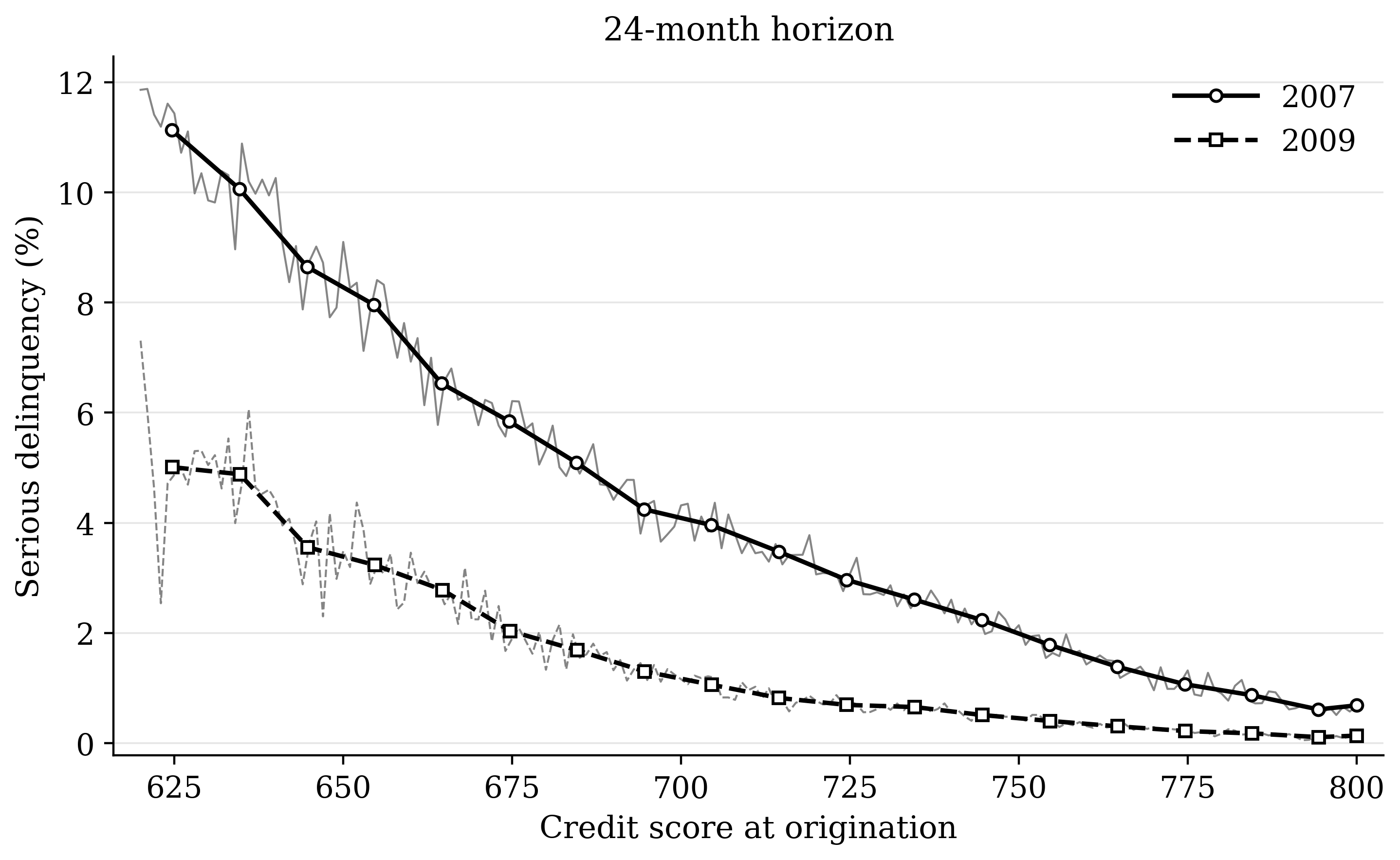}
\caption{Serious delinquency within 24 months by credit score and
origination year. Thin grey lines show individual-score frequencies.
Black solid and dashed lines connect frequencies pooled within
successive ten-score groups for 2007 and 2009, respectively; the
last group contains score 800. All eligible loans are used in this
descriptive figure.}
\label{fig:creditscoreprofiles}
\end{figure}

For estimation, we form a balanced design with $T=449$ loans from
each origination year at every score, the largest common count
among eligible loans. Here $T$ counts loans per score and year,
rather than calendar months. Selection uses a fixed random ordering
of loan identifiers, independent of outcomes. The same loans enter
both horizons and all five procedures. We estimate $\theta_i$ by
the difference in delinquency frequencies. The sum of the two
sample variances estimates $\sigma_i^2$, giving standard error
$\widehat\sigma_i/\sqrt T$; the difference in third central moments
enters the sample-mean bias correction.\footnote{Section \ref{subsec:sa five interval} reports the coefficient and standard-error
formulas.}

The score grid connects the application to the sufficient model
in Section SE.1. Smooth risk curves give smooth coefficient and
standard-deviation functions across scores. When the difference
crosses the target level with a nonzero slope away from the credit-score range endpoints, that model yields
\[
 |H_n(q_n)-\tau|=O(n^{-1}+T^{-1}),
\]
which implies Assumption~\ref{ass:balance} in the growth range of
Theorem~\ref{thm:wald}.\footnote{At the boundary $n/T^{3/2}\to c\in(0,\infty)$
allowed by Theorem~\ref{thm:bc}, the $T^{-1}$ term in this
bound is of the same order as the standard error
$n^{-1/2}T^{-1/4}$. Thus, the smooth-grid conditions alone
do not establish Assumption~\ref{ass:balance} in that regime.
 Normal inference for $q_n$ continues
to require local balance to make the separate deterministic
centering error $\tilde q_n-q_n$ negligible.} The risk curves may be nonlinear,
allowing unequal spacing between the coefficients. We interpret
normal inference under these conditions and report $I_{\mathrm C}$
to assess which conclusions remain supported without local balance.
The stochastic reference $I_{\mathrm S}$ uses standard error
$\sqrt{\tau(1-\tau)/n}/\widehat F'(\widehat q)$, as in
\eqref{eq:simulationreferences}. For $I_{\mathrm P}$, we project
simultaneous exact binomial intervals.

\subsection{Results}

The results point to a broad reduction in delinquency risk among
mortgages originated in 2009 relative to those originated in 2007.
Figure~\ref{fig:empirical} reports quantiles of the reduction in
24-month serious delinquency across the 181 fixed credit scores.
The estimated median reduction is 2.35 percentage points, with
$I_{\mathrm U}=[2.17,2.54]$ and
$I_{\mathrm{BC}}=[2.18,2.55]$. The improvement also extends below
the median: the estimated lower quartile is 1.01 percentage points,
with $I_{\mathrm U}=[0.89,1.14]$ and
$I_{\mathrm{BC}}=[0.91,1.16]$. Under the conditions for normal
inference, these results support lower delinquency risk at more
than three quarters of the credit-score values. Because the
comparison holds credit scores fixed, this pattern cannot be
explained solely by a larger share of loans going to borrowers
with higher scores.

\begin{figure}[!t]
\centering
\includegraphics[width=0.75\textwidth]{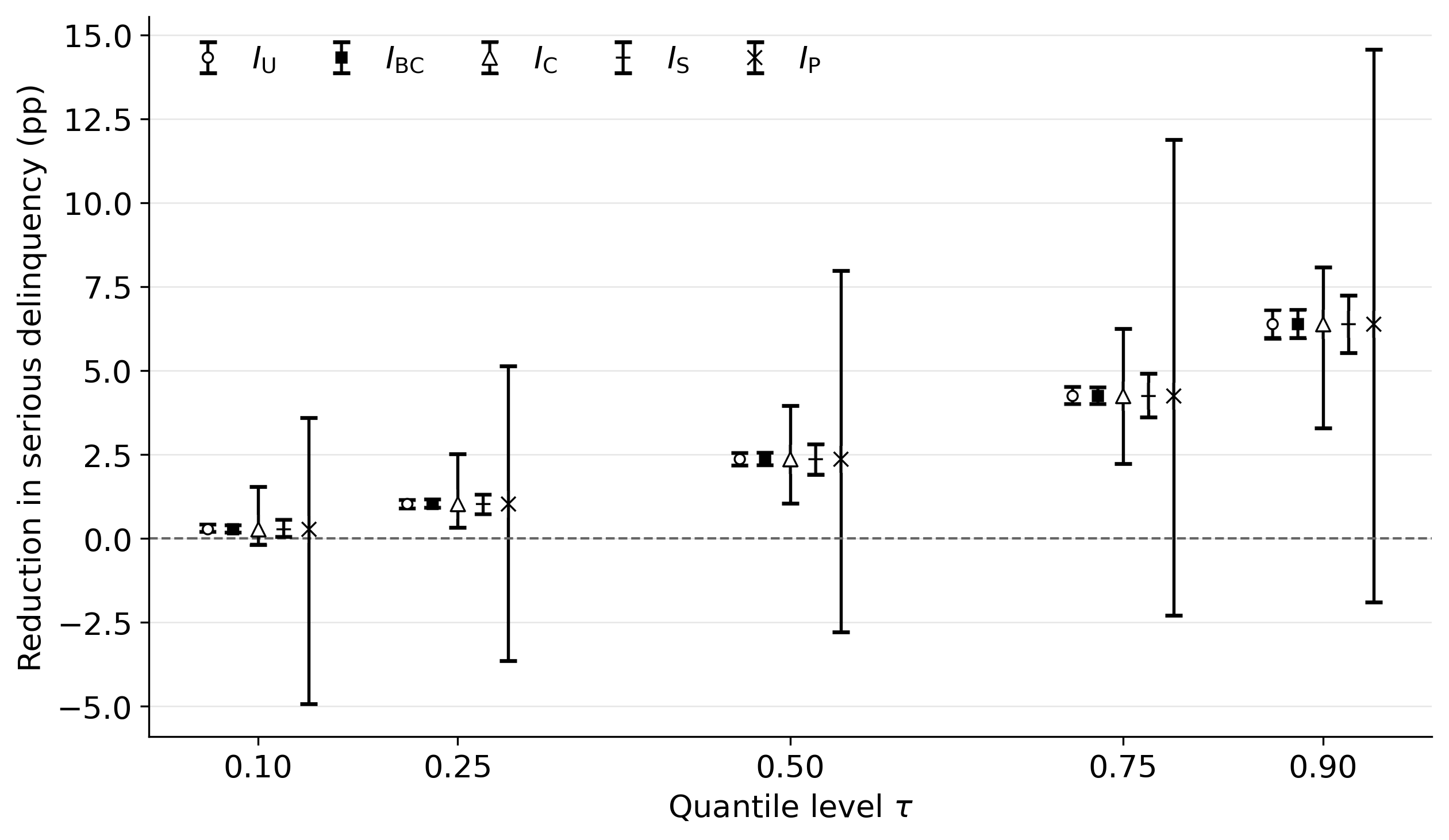}
\caption{Quantiles of the reduction in 24-month serious delinquency
between mortgages originated in 2007 and 2009. The five procedures
use $n=181$ credit scores and $T=449$ loans per score and origination
year. Effects are in percentage points. Markers show $\widehat q$,
or $\widehat q_{\mathrm{BC}}$ for $I_{\mathrm{BC}}$; horizontal
offsets separate the intervals. The dashed line marks zero.
Intervals are pointwise in $\tau$.}
\label{fig:empirical}
\end{figure}

The improvement is widespread but uneven. The estimated reduction
ranges from 0.30 percentage points at the lower decile to 6.38
percentage points at the upper decile. Even at the lower decile,
$I_{\mathrm U}=[0.19,0.41]$ and
$I_{\mathrm{BC}}=[0.17,0.39]$ exclude zero, suggesting improvement
at more than 90\% of the score values under normal inference.
The horizon also matters. Using the same loans, the 12-month analysis
in Section SD gives an estimated median reduction of 0.17 percentage
points, with $I_{\mathrm U}=[0.10,0.24]$ and
$I_{\mathrm{BC}}=[0.16,0.29]$. The normal procedure indicates a positive median reduction at
each horizon, although the estimated reduction is smaller
at 12 months.
Bias correction has little effect on the 24-month results and a
more noticeable effect at 12 months, where delinquency events
are less frequent.

The conservative intervals assess which findings remain supported
without local balance. At 24 months,
$I_{\mathrm C}=[1.04,3.96]$ for the median and
$[0.32,2.51]$ for the lower quartile both exclude zero, supporting
the central finding of improvement across more than three quarters
of the score values. Its lower-decile interval,
$[-0.19,1.54]$, includes zero, so it does not support the stronger
conclusion concerning more than 90\% of scores. Thus, the
evidence for widespread improvement is stronger over two years
and remains informative without relying on negligible centering
error. The stochastic reference also excludes zero at the lower
quartile, $I_{\mathrm S}=[0.72,1.30]$, although it includes variation
from sampling units that is absent here. The simultaneous projection
is less informative, including zero at every reported quantile.

Section~\ref{subsec:mortgagerobustness} reports
estimates and intervals using all eligible loans, alternative
selections of the balanced sample, and different lower bounds on
estimated standard deviations. These checks support the broad
24-month pattern reported here.\footnote{Section~\ref{subsec:mortgagedecomposition} further
decomposes the aggregate delinquency decline into changes in
delinquency at given credit scores and changes in the credit-score
distribution of loans.}
Overall, the normal and conservative quantile intervals support
lower delinquency risk across more than three quarters of the
credit-score values.

\section{Conclusion}\label{sec:conclusion}

Quantiles of deterministic individual-specific coefficients describe
how effects are distributed across the population being studied.
We develop a smoothed estimator, feasible variance estimation, and
bias-corrected normal inference, with primitive
conditions covering GMM. Without local balance, the sampling distribution of the
studentized estimation error relative to the target quantile
cannot be estimated uniformly consistently. A conservative interval has coverage tending to one uniformly and attains the optimal worst-case length rate in probability. Thus, inference remains possible without local balance, with an explicit cost in precision. 
The mortgage application reveals widespread but uneven reductions in delinquency risk across credit-score levels.

\appendix
\numberwithin{equation}{section}
\section{First-Step Conditions for GMM}\label{sec:primitive}
Propositions~\ref{prop:gmm} and \ref{prop:gmmbias} give sufficient
conditions for Assumptions~\ref{ass:first} and \ref{ass:second},
respectively. Proofs are in Section SA of the
Supplementary Materials.

\subsection{First-order estimation}\label{subsec:primitivefirst}
Consider GMM with fixed parameter and moment dimensions:
\begin{equation}\label{eq:gmm}
 \bar{\bm g}_i(\bm b)=\frac1T\sum_t\bm g_i(\bm Z_{it},\bm b),
 \qquad
 \widehat{\bm\beta}_i\in\arg\min_{\bm b\in\mathcal B}
 \bar{\bm g}_i(\bm b)^{\top}\bm W_i\bar{\bm g}_i(\bm b),
\end{equation}
where $\mathcal B\subset\mathbb R^p$ is compact and $\bm W_i$
is deterministic and symmetric. Write
$\bm m_i(\bm b)=\E\left[\bm g_i(\bm Z_{it},\bm b)\right]$,
$\mathcal B^\delta=\{\bm b:\operatorname{dist}(\bm b,\mathcal B)<\delta\}$,
\(
 \bm G_i=\E\left[\partial_{\bm b^{\top}}\bm g_i(\bm Z_{it},\bm\beta_i)\right]
 \), $ \bm\Omega_i
 =\E\!\left[
 \bm g_i(\bm Z_{it},\bm\beta_i)
 \bm g_i(\bm Z_{it},\bm\beta_i)^\top
 \right]$, and \(
 \bm A_i=(\bm G_i^{\top}\bm W_i\bm G_i)^{-1}\bm G_i^{\top}\bm W_i.
\)
The contrast influence and its variance are
\begin{equation}\label{eq:gmminfluence}
 \psi_{it}=-\bm a^{\top}\bm A_i\bm g_i(\bm Z_{it},\bm\beta_i),
 \qquad \sigma_i^2=\E[\psi_{it}^2]
=\bm a^\top\bm A_i\bm\Omega_i\bm A_i^\top\bm a.
\end{equation}

\begin{proposition}[First-order expansion for smooth GMM]\label{prop:gmm}
Suppose (i) $\{\bm Z_{it}\}_{t=1}^T$ are i.i.d. within $i$ and
 independent across $i$. For constants $\delta,c,C>0$,
(ii) $\bm m_i(\bm\beta_i)=\bm0$ and
$\{\bm b:\|\bm b-\bm\beta_i\|\le\delta\}\subset\mathcal B$;
(iii) for every $\varepsilon>0$,
\(
 \inf_i\inf_{\substack{\bm b\in\mathcal B\\
                       \|\bm b-\bm\beta_i\|\ge\varepsilon}}
 \bm m_i(\bm b)^{\top}\bm W_i\bm m_i(\bm b)
 \ge\kappa_\varepsilon>0;
\)
(iv) $\bm g_i(\bm Z_{it},\bm b)$ is three times continuously
differentiable in $\bm b$ on $\mathcal B^\delta$, with
$\max_{0\le j\le3}\sup_{\bm b\in\mathcal B^\delta}
\|\partial_{\bm b}^{j}\bm g_i(\bm Z_{it},\bm b)\|\le C,\ \text{a.s.};$
(v) $c\le\lambda_{\min}(\bm W_i)\le\lambda_{\max}(\bm W_i)\le C$, $c\le\lambda_{\min}(\bm G_i^{\top}\bm W_i\bm G_i)\le\lambda_{\max}(\bm G_i^{\top}\bm W_i\bm G_i)\le C$,
and $c\le\sigma_i^2\le C$.
All bounds are uniform over units and sample sizes.
Then Assumption~\ref{ass:first} holds with
\eqref{eq:gmminfluence} and the scale estimator \eqref{eq:gmmscale} below.
\end{proposition}

Let
\(
 \widehat{\bm G}_i
 =\frac1T\sum_{t=1}^T
 \partial_{\bm b^\top}\bm g_i(\bm Z_{it},\widehat{\bm\beta}_i)\) and \(
 \widehat{\bm\Omega}_i
 =\frac1T\sum_{t=1}^T
 \bm g_i(\bm Z_{it},\widehat{\bm\beta}_i)
 \bm g_i(\bm Z_{it},\widehat{\bm\beta}_i)^\top
 -\bar{\bm g}_i(\widehat{\bm\beta}_i)
  \bar{\bm g}_i(\widehat{\bm\beta}_i)^\top\)  be the sample
versions of $\bm G_i$ and $\bm \Omega_i$, respectively. When
$\lambda_{\min}(\widehat{\bm G}_i^\top\bm W_i
\widehat{\bm G}_i)\ge c/2$, define
\(
 \widehat{\bm A}_i
 =(\widehat{\bm G}_i^\top\bm W_i\widehat{\bm G}_i)^{-1}
   \widehat{\bm G}_i^\top\bm W_i
\)
and set
\begin{equation}\label{eq:gmmscale}
 \widehat\sigma_i^2
 =\min\left\{2C,\max\left\{\frac c2,
 \bm a^{\top}\widehat{\bm A}_i\widehat{\bm\Omega}_i
                       \widehat{\bm A}_i^{\top}\bm a\right\}\right\};
\end{equation}
otherwise set $\widehat\sigma_i^2=c/2$. This case occurs with exponentially small probability, uniformly in \(i\).

Smooth M-estimators are covered by taking $\bm g_i$ to be the
gradient of the loss under the same conditions. A binary-response
likelihood with bounded covariates and uniformly nonsingular
information is one example. These primitive conditions cover smooth
estimating equations, including
smooth M-estimation, rather than nonsmooth quantile-regression scores.

\subsection{Conditions for bias correction}\label{subsec:primitivebias}
For smooth GMM, $\bm v_{it}$ and $\bm D_i$ in
Assumption~\ref{ass:second} follow from expanding the estimating
equations and the sandwich variance. Section SA.4 gives their
formulas. The estimates $\widehat b_i$ and $\widehat\gamma_i$ are
the sample analogues of \eqref{eq:bcconstants}, using estimated
parameters, derivatives, and centered empirical moments.
Use bounded versions of sample inverses and bias coefficients
that agree with the original formulas in a fixed neighborhood
of their population values.

\begin{proposition}[Second-order expansion for smooth GMM]\label{prop:gmmbias}
Suppose the conditions of Proposition~\ref{prop:gmm} hold,
$\bm g_i(\bm Z_{it},\bm b)$ is four times continuously
differentiable in $\bm b$ on $\mathcal B^\delta$, and
\[
 \sup_{\bm b\in\mathcal B^\delta}
 \|\partial_{\bm b}^{4}\bm g_i(\bm Z_{it},\bm b)\|
 \le C\quad\text{almost surely, uniformly in }i,t.
\]
Then Assumption~\ref{ass:second} holds for
\eqref{eq:gmmscale} and the sample analogues of
\eqref{eq:bcconstants} described above.
\end{proposition}

Both propositions use deterministic GMM weights. Estimated
weights require additional terms for their estimation error. The same conclusion applies to smooth M-estimators satisfying
these conditions on their score equations.

\section{Proofs for the Main Theorems}\label{sec:mainthmproofs}
\subsection{Proof of Theorem~\ref{thm:wald}}\label{sec:sharp}
\begin{proof}
We prove a normal limit for the centered CDF, transfer it to
$\widehat q-\tilde q_n$, and estimate its variance and slope.
Local balance is used only in the final step to replace
$\tilde q_n$ by $q_n$.

\medskip\noindent\textit{Step 1: prove
$(G_n(\tilde q_n)-\E G_n(\tilde q_n))/\sqrt{V_n}\Rightarrow N(0,1)$.}
Let $v(u)=\Var\{\Phi(u-Z)\}$ for $Z\sim N(0,1)$, and define
\begin{equation}\label{eq:V}
 V_n=\frac1{n^2}\sum_i v\{u_i(\tilde q_n)\}.
\end{equation}
By Lemma~\ref{lem:envelope} and \eqref{eq:rootbracket},
$V_n\asymp(n\sqrt T)^{-1}$ and $H_n'(\tilde q_n)\asymp1$.
Applying Lemma~\ref{lem:comparison} to $\Phi$ and its square
bounds the absolute variance error by $C/(nT)$.
Since $V_n\asymp1/(n\sqrt T)$, the relative error is
$O(T^{-1/2})=o(1)$. Hence
\begin{equation}\label{eq:scorevariance}
 \Var\{G_n(\tilde q_n)\}=V_n\{1+o(1)\}.
\end{equation}
The summands are independent and bounded, and their total variance
has order $n/\sqrt T\to\infty$. Lindeberg's condition therefore holds, so
\begin{equation}\label{eq:clt}
 \frac{G_n(\tilde q_n)-\E G_n(\tilde q_n)}{\sqrt{V_n}}
 \Rightarrow N(0,1).
\end{equation}

\medskip\noindent\textit{Step 2: derive the expansion for $\widehat q-\tilde q_n$.}
Proposition~\ref{prop:root} and its slope calculation give
$\widehat F'(x)=H_n'(\tilde q_n)+o_p(1)$ between
$\widehat q$ and $\tilde q_n$. The mean comparison and feasible
first-step error are negligible relative to the standard error because
\(
 \frac{T^{-1}}{n^{-1/2}T^{-1/4}}=\sqrt{n/T^{3/2}}\to0.
\)
The mean value theorem and $H_n(\tilde q_n)=\tau$ give
\begin{align*}
 0&=\widehat F(\widehat q)-\tau=G_n(\tilde q_n)-\E G_n(\tilde q_n)
    +\{H_n'(\tilde q_n)+o_p(1)\}(\widehat q-\tilde q_n)
    +o_p(n^{-1/2}T^{-1/4}).
\end{align*}
Solving for the quantile error yields
\begin{equation}\label{eq:expansion}
 \widehat q-\tilde q_n
 =-\frac{G_n(\tilde q_n)-\E G_n(\tilde q_n)}{H_n'(\tilde q_n)}
       +o_p\left(n^{-1/2}T^{-1/4}\right).
\end{equation}

\medskip\noindent\textit{Step 3: prove consistency of the variance and slope estimates.}
Lemma~\ref{lem:normalidentities} gives the exact identity
$\E\Phi(-|u-Z|)^2=v(u)$. With the linear first step and true scales,
Lemmas~\ref{lem:envelope}-\ref{lem:comparison} consequently give
\[
 \E\sum_i\Phi(-|u_i(\tilde q_n)-W_i|)^2=n^2V_n+O(n/T).
\]
Independence and boundedness make the variance of this sum $O(n/\sqrt T)$.
Its mean is of order $n/\sqrt T$, so the relative random error is
$O_p(n^{-1/2}T^{1/4})=o_p(1)$. Lemma~\ref{lem:perturb} bounds the
effect of the coefficient remainder and estimated scale by $O_p(n/T)$.
The average absolute derivative of the squared contributions is
$O_p(1)$ on the fixed neighborhood, by the derivative-envelope
argument in that lemma. Evaluation at $\widehat q$ instead of
$\tilde q_n$ therefore changes the sum by
$O_p(n|\widehat q-\tilde q_n|)$. Relative to $n/\sqrt T$, this is
\[
 O_p\left(\sqrt T|\widehat q-\tilde q_n|\right)
 =O_p\left(n^{-1/2}T^{1/4}+T^{-1/2}\right)=o_p(1).
\]
The slope calculation in the proof of Proposition~\ref{prop:root}
likewise remains valid at the estimated root. Thus
\begin{equation}\label{eq:varianceconsistent}
 \frac{n^{-2}\sum_i p_i(\widehat q)^2}{V_n}\to_p1,
 \qquad
 \frac{\widehat F'(\widehat q)}{H_n'(\tilde q_n)}\to_p1.
\end{equation}
These consistency statements use only
Assumptions~\ref{ass:first}-\ref{ass:counts}. Combining
\eqref{eq:clt}, \eqref{eq:expansion}, and
\eqref{eq:varianceconsistent} gives
\[
 \frac{\widehat q-\tilde q_n}{\widehat{\mathrm{se}}(\widehat q)}
 =-\frac{G_n(\tilde q_n)-\E G_n(\tilde q_n)}{\sqrt{V_n}}+o_p(1)
 \Rightarrow N(0,1).
\]
This proves the first conclusion without local balance.

\medskip\noindent\textit{Step 4: center the limit at $q_n$ under local balance.}
Both $q_n$ and $\tilde q_n$ lie in the fixed neighborhood in
\eqref{eq:rootbracket}, where the slope of $H_n$ is bounded away
from zero. Hence Assumption~\ref{ass:balance} implies
\[
 |\tilde q_n-q_n|
 \le C|H_n(q_n)-\tau|
 =o\left(n^{-1/2}T^{-1/4}\right).
\]
Since $\widehat{\mathrm{se}}(\widehat q)\asymp_p n^{-1/2}T^{-1/4}$,
\(
 \frac{\widehat q-q_n}{\widehat{\mathrm{se}}(\widehat q)}
 =\frac{\widehat q-\tilde q_n}{\widehat{\mathrm{se}}(\widehat q)}
  +\frac{\tilde q_n-q_n}{\widehat{\mathrm{se}}(\widehat q)}
 \Rightarrow N(0,1).
\)
This proves the second conclusion.
\end{proof}

\subsection{Proof of Theorem~\ref{thm:bc}}\label{sec:proofbias}
\begin{proof}
We retain the order-$T^{-1}$ mean error at $\tilde q_n$, divide
by the local slope, and estimate the resulting bias coefficient.
Local balance is used only to obtain the conclusion for $q_n$.
Write $r_n=n^{-1/2}T^{-1/4}$ and
\[
 C_n(x)=\frac1n\sum_i\{b_i h_i(x)+\gamma_i h_i''(x)/(6T)\},
 \qquad B_n=\frac{C_n(\tilde q_n)}{H_n'(\tilde q_n)}.
\]
Assumption~\ref{ass:second} implies Assumption~\ref{ass:first}:
$\E Q_i^2\le C$ gives $\E(Q_i/T+r_i)^2\le C/T^2$, and the
bounded positive scales turn the variance expansion into the required
mean-square scale bound. The root and variance calculations above
therefore apply without local balance.

\medskip\noindent\textit{Step 1: derive the CDF limit and mean error at $\tilde q_n$.}
Use $V_n$ defined in \eqref{eq:V}, for which $V_n\asymp r_n^2$.
The centered linear CDF satisfies \eqref{eq:clt}. By
Lemma~\ref{lem:perturb}, the difference between the centered
feasible and linear CDFs has variance
$O((nT)^{-1})=o(r_n^2)$. Consequently,
\begin{equation}\label{eq:bcCDFclt}
 \frac{\widehat F(\tilde q_n)-\E\widehat F(\tilde q_n)}{\sqrt{V_n}}
 \Rightarrow N(0,1).
\end{equation}
By \eqref{eq:rootbracket}, $\tilde q_n$ lies within a fixed
multiple of $T^{-1/2}$ of $q_n$. Lemma~\ref{lem:biasmean}
therefore applies at $\tilde q_n$. Since $H_n(\tilde q_n)=\tau$,
\[
 \E\widehat F(\tilde q_n)-\tau
 =-\frac{C_n(\tilde q_n)}{T}+o(T^{-1})
 =-\frac{C_n(\tilde q_n)}{T}+o(r_n),
\]
where $T^{-1}/r_n=\sqrt{n/T^{3/2}}=O(1)$.

\medskip\noindent\textit{Step 2: derive the leading bias of $\widehat q-\tilde q_n$.}
Proposition~\ref{prop:root} gives
$\widehat q-\tilde q_n=O_p(r_n+T^{-1})=O_p(r_n)$.
The slope calculation in its proof gives
$\widehat F'(\tilde q_n)=H_n'(\tilde q_n)+o_p(1)$,
and the local second derivative is $O_p(\sqrt T)$.
Expanding $\widehat F(\widehat q)=\tau$ at $\tilde q_n$ yields
\[
 \begin{aligned}
 0={}&\widehat F(\tilde q_n)-\E\widehat F(\tilde q_n)
      -C_n(\tilde q_n)/T+H_n'(\tilde q_n)(\widehat q-\tilde q_n)+o_p(r_n).
 \end{aligned}
\]
The quadratic remainder is $O_p(\sqrt T\,r_n^2)=o_p(r_n)$,
because $\sqrt T\,r_n=n^{-1/2}T^{1/4}\to0$.
Dividing by $H_n'(\tilde q_n)$ gives
\begin{equation}\label{eq:bcrootproof}
 \widehat q-\tilde q_n=\frac{B_n}{T}
 -\frac{\widehat F(\tilde q_n)-\E\widehat F(\tilde q_n)}
        {H_n'(\tilde q_n)}+o_p(r_n).
\end{equation}

\medskip\noindent\textit{Step 3: establish the bias-corrected limit around $\tilde q_n$.}
Lemma~\ref{lem:Bestimate} gives $\widehat B_n-B_n=o_p(1)$ and hence
\(
 \frac{\widehat B_n-B_n}{Tr_n}
 =(\widehat B_n-B_n)\sqrt{n/T^{3/2}}=o_p(1).
\)
The variance consistency in \eqref{eq:varianceconsistent} uses
only Assumptions~\ref{ass:first}-\ref{ass:counts}, so
\(
 \widehat{\mathrm{se}}(\widehat q)\big/
 \{\sqrt{V_n}/H_n'(\tilde q_n)\}\to_p1.
\)
Subtracting $\widehat B_n/T$ in \eqref{eq:bcrootproof} and using
\eqref{eq:bcCDFclt} now gives
\[
 \frac{\widehat q_{\mathrm{BC}}-\tilde q_n}
      {\widehat{\mathrm{se}}(\widehat q)}
 =-\frac{\widehat F(\tilde q_n)-\E\widehat F(\tilde q_n)}{\sqrt{V_n}}
   +o_p(1)\Rightarrow N(0,1).
\]
The variance and slope orders also give
$|I_{\mathrm{BC}}|\asymp_p r_n$.

\medskip\noindent\textit{Step 4: obtain inference for $q_n$ under local balance.}
The positive local slope and Assumption~\ref{ass:balance} imply
$|\tilde q_n-q_n|\le C|H_n(q_n)-\tau|=o(r_n)$.
Therefore
\[
 \frac{\widehat q_{\mathrm{BC}}-q_n}
      {\widehat{\mathrm{se}}(\widehat q)}
 =\frac{\widehat q_{\mathrm{BC}}-\tilde q_n}
       {\widehat{\mathrm{se}}(\widehat q)}+o_p(1)
 \Rightarrow N(0,1).
\]
The definition of $I_{\mathrm{BC}}$ then gives
$P(q_n\in I_{\mathrm{BC}})\to1-\alpha$.
\end{proof}

\subsection{Proof of Theorem~\ref{thm:bound}}\label{sec:proofconservative}
\begin{proof}
All bounds below use only the common constants defining $\mathcal{F}$.
Consequently, every convergence in probability is uniform over that
class. Abbreviate the observable lower and upper bounds by
\begin{equation}\label{eq:observablebounds}
 \bar\ell(x)=\widehat F(x)-2\bar p(x),\qquad
 \bar u(x)=\widehat F(x)+2\bar p(x).
\end{equation}
We first show that $\widehat q$ lies within $K_0T^{-1/2}$ of
$q_n$. Uniform control on this neighborhood then shows that every
point between them satisfies $\bar\ell(x)\le\tau\le\bar u(x)$.
This proves coverage; nearby violations and a small covered
neighborhood provide the upper and lower length bounds.

\medskip\noindent\textit{Step 1: prove $P(|\widehat q-q_n|\le K_0T^{-1/2})\to1$ uniformly.}
The deterministic bracket \eqref{eq:rootbracket} has a gap of order
$T^{-1/2}$ at both endpoints. Lemmas~\ref{lem:comparison},
\ref{lem:perturb}, and \ref{lem:conservativebounds} give an $O(T^{-1})$
mean approximation error and an $O\{1/(n\sqrt T)\}$ variance for the feasible
CDF at these endpoints. For the right endpoint, monotonicity gives
\begin{align*}
 P\{\widehat q>q_n+K_0T^{-1/2}\}
 &\le P\{\widehat F(q_n+K_0T^{-1/2})<\tau\}\le\frac{C/(n\sqrt T)}{(c_0/(2\sqrt T))^2}
 \le C\sqrt T/n.
\end{align*}
The second inequality uses Chebyshev's inequality after absorbing
the $O(T^{-1})$ mean error into half the deterministic margin.
The left endpoint is identical with the inequalities reversed.
Together they give
\begin{equation}\label{eq:uniformrootwindow}
 \sup_{F\in\mathcal{F}}
 P_{nT}\{|\widehat q-q_n|>K_0T^{-1/2}\}\le C\sqrt T/n\to0.
\end{equation}
This uses the strict deterministic bracket, rather than only the
$O_p(T^{-1/2})$ conclusion of Proposition~\ref{prop:root}.

\medskip\noindent\textit{Step 2: prove $\bar\ell(x)\le\tau\le\bar u(x)$ throughout the segment.}
Lemma~\ref{lem:conservativebounds} gives, for every fixed $K$,
\[
 \sup_{|x-q_n|\le KT^{-1/2}}
 \bigl(|\bar\ell(x)-\E\bar\ell(x)|+
       |\bar u(x)-\E\bar u(x)|\bigr)=o_p(T^{-1/2}),
\]
and strict expectation bounds in the needed directions:
\[
 \begin{aligned}
 \inf_{q_n\le x\le q_n+KT^{-1/2}}\E\bar u(x)&\ge\tau+c_KT^{-1/2}-C_KT^{-1},\quad \sup_{q_n-KT^{-1/2}\le x\le q_n}\E\bar\ell(x)&\le\tau-c_KT^{-1/2}+C_KT^{-1}.
 \end{aligned}
\]
Take $K=K_0$. The mean margins have order $T^{-1/2}$,
whereas $T^{-1}=o(T^{-1/2})$ and the uniform estimation error
is $o_p(T^{-1/2})$. Combining these bounds with Step 1 gives,
on an event whose probability tends to one uniformly,
$|\widehat q-q_n|\le K_0T^{-1/2}$ and
\[
 \inf_{q_n\le x\le q_n+K_0T^{-1/2}}\bar u(x)>\tau,
 \qquad
 \sup_{q_n-K_0T^{-1/2}\le x\le q_n}\bar\ell(x)<\tau.
\]
If $\widehat q\ge q_n$, then
for every $x\in[q_n,\widehat q]$,
\(
 \bar\ell(x)\le\widehat F(x)\le\tau<\bar u(x).
\)
If $\widehat q\le q_n$, then for every $x\in[\widehat q,q_n]$,
\(
 \bar\ell(x)<\tau\le\widehat F(x)\le\bar u(x).
\)
The middle inequalities use monotonicity of $\widehat F$ and
$\widehat F(\widehat q)=\tau$. Therefore the defining inequality
for $I_C$ cannot fail between $\widehat q$ and $q_n$. Its endpoints
must lie on opposite sides of $q_n$, proving
\[
 \inf_{F\in\mathcal{F}}P_{nT}\{L_C\le q_n\le U_C\}\to1.
\]

\medskip\noindent\textit{Step 3: prove $|I_C|\le2KT^{-1/2}$ with probability tending to one.}
Lemma~\ref{lem:nearbyrejection} supplies fixed $K,c_1>0$ and
deterministic points $x_-\in[q_n-KT^{-1/2},q_n]$ and
$x_+\in[q_n,q_n+KT^{-1/2}]$ such that
\(
 \E\bar u(x_-)\le\tau-c_1T^{-1/2}\) and \(
 \E\bar\ell(x_+)\ge\tau+c_1T^{-1/2}.
\)
Uniform concentration gives
$\bar u(x_-)<\tau<\bar\ell(x_+)$ with probability tending to one.
Thus the defining inequality fails at both points, and
\[
 \widehat F(x_-)\le\bar u(x_-)<\tau
 =\widehat F(\widehat q)<\bar\ell(x_+)\le\widehat F(x_+).
\]
Consequently $x_-<\widehat q<x_+$, and the endpoint definitions give
\begin{equation}\label{eq:lengthupperproof}
 q_n-KT^{-1/2}\le x_-\le L_C<U_C\le x_+\le q_n+KT^{-1/2}.
\end{equation}
In particular, $|I_C|\le2KT^{-1/2}$ with probability tending to one
uniformly over $\mathcal{F}$. This step does not presume that the
adjusted bounds are monotone.

\medskip\noindent\textit{Step 4: prove $|I_C|\ge2\delta T^{-1/2}$ with probability tending to one.}
At $q_n$, Lemma~\ref{lem:conservativebounds} gives mean gaps of at
least $cT^{-1/2}$ for all sufficiently large samples. The same lemma gives
$\E\sup_{|x-q_n|\le T^{-1/2}}\widehat F'(x)\le C$. The absolute
one-sided slopes of $\bar\ell$ and $\bar u$ are bounded by
$3\widehat F'$, so their expectations are uniformly Lipschitz on
this neighborhood. In particular, for $|x-q_n|\le\delta T^{-1/2}$, we have 
\(
 |\E\bar\ell(x)-\E\bar\ell(q_n)|+
 |\E\bar u(x)-\E\bar u(q_n)|\le C\delta T^{-1/2}.
\)
Choose a fixed $\delta>0$ with $C\delta<c/2$. Then
\[
 \E\bar\ell(x)\le\tau-c/(2\sqrt T),\qquad
 \E\bar u(x)\ge\tau+c/(2\sqrt T)
 \quad\text{for every }|x-q_n|\le\delta T^{-1/2}.
\]
Uniform concentration makes the defining inequality hold throughout
this neighborhood with probability tending to one. Step 2 connects it to $\widehat q$,
so it belongs to the reported interval. Hence $|I_C|\ge2\delta T^{-1/2}$.
Combining this with Step 3 proves the length conclusion with
$c_0=2\delta$ and $C_0=2K$. 
\end{proof}

\begingroup\setlength{\bibsep}{2pt}
\setstretch{1.5}

\endgroup

\clearpage
\newgeometry{margin=1in}
\onehalfspacing
\renewcommand{\footnotelayout}{\singlespacing}
\setcounter{section}{0}
\setcounter{subsection}{0}
\setcounter{equation}{0}
\setcounter{table}{0}
\setcounter{figure}{0}
\setcounter{theorem}{0}
\setcounter{proposition}{0}
\setcounter{footnote}{0}
\renewcommand{\thesection}{S\Alph{section}}
\numberwithin{table}{section}
\numberwithin{figure}{section}
\numberwithin{theorem}{section}
\numberwithin{proposition}{section}
\renewcommand{\theHsection}{supp.\arabic{section}}
\renewcommand{\theHsubsection}{\theHsection.\arabic{subsection}}
\renewcommand{\theHequation}{\theHsection.\arabic{equation}}
\renewcommand{\theHtable}{\theHsection.\arabic{table}}
\renewcommand{\theHfigure}{\theHsection.\arabic{figure}}
\renewcommand{\theHtheorem}{\theHsection.\arabic{theorem}}
\renewcommand{\theHproposition}{\theHsection.\arabic{proposition}}
\renewcommand{\theHlemma}{\theHsection.\arabic{lemma}}
\renewcommand{\theHextensionprop}{\theHsection.\arabic{extensionprop}}
\emptythanks
\title{Supplementary Materials\\ Quantiles of Heterogeneous Coefficients in Observed Populations}
\author{Antonio F. Galv\~ao \and Ulrich Hounyo \and Jiahao Lin}
\date{\today}

\maketitle
These supplementary materials contain the proofs of propositions
in Section SA and the supporting lemmas in Section SB. Sections SC and SD
report additional simulation and empirical results. Section SE gives sufficient conditions for local balance and local counts,
and an extension to time-series observations.  

\section{Proofs of Propositions}\label{sec:mainproofs}

To distinguish the linear first-step approximation from the feasible
estimator, write
\[
 W_i=\frac1{\sigma_i\sqrt T}\sum_t\psi_{it},\qquad
 u_i(x)=\frac{\sqrt T(x-\theta_i)}{\sigma_i},\qquad
 G_n(x)=\frac1n\sum_i\Phi\{u_i(x)-W_i\}.
\]
Thus $G_n$ uses the influence average and true standard errors, whereas
$\widehat F$ also includes the remainder and estimated standard errors.
Expectations hold the coefficient array fixed.

\subsection{Proof of Proposition~\ref{prop:root}}\label{sec:counts}
\begin{proof}
We first locate the root of $H_n(x)=\tau$. We then bound
$\widehat F(\tilde q_n)-\tau$ and use the positive local slope
to obtain the bound for $\widehat q-\tilde q_n$.

\medskip\noindent\textit{Step 1: prove $|\tilde q_n-q_n|=O(T^{-1/2})$.}
Recall that $H_n(x)=n^{-1}\sum_i\Phi\{u_i(x)/\sqrt2\}$.
Adding and subtracting $F_n(q_n)$ gives the exact decomposition
\begin{align*}
 H_n(x)-\tau
 ={}&F_n(q_n)-\tau
   +\frac1n\sum_{\theta_i>q_n}\Phi\{u_i(x)/\sqrt2\}-\frac1n\sum_{\theta_i\le q_n}
           [1-\Phi\{u_i(x)/\sqrt2\}].
\end{align*}
Since $F_n(q_n)-\tau\ge0$ and every term in the positive sum is
nonnegative, retaining only $q_n<\theta_i\le q_n+T^{-1/2}$ yields
\begin{align*}
 H_n(x)-\tau\ge{}&
 \frac1n\sum_{q_n<\theta_i\le q_n+T^{-1/2}}
       \Phi\{u_i(x)/\sqrt2\}-\frac1n\sum_{\theta_i\le q_n}
       \overline\Phi\{u_i(x)/\sqrt2\},
 \qquad\overline\Phi=1-\Phi.
\end{align*}
Set $x=q_n+KT^{-1/2}$, where $K>2$. For each coefficient
in the first sum, $x-\theta_i>0$, so its contribution is at
least $1/2$. The right-side count in Assumption~\ref{ass:counts}
therefore gives
\[
 \frac1n\sum_{q_n<\theta_i\le q_n+T^{-1/2}}
       \Phi\{u_i(x)/\sqrt2\}\ge\frac{c}{2\sqrt T}.
\]
For the negative sum, group the distances $q_n-\theta_i$ into
$[0,T^{-1/2}]$ and
$(2^jT^{-1/2},2^{j+1}T^{-1/2}]$. The count upper bound and
the normal tail bound give
\begin{align*}
 &\frac1n\sum_{\theta_i\le q_n}
 \overline\Phi\left(\frac{\sqrt T(q_n-\theta_i)+K}
                         {\sqrt2\sigma_i}\right)\le\frac C{\sqrt T}
       \left\{e^{-c'K^2}+\sum_{j\ge0}2^j e^{-c'(K+2^j)^2}\right\}
       +Ce^{-c'T}.
\end{align*}
The last term covers coefficients outside the fixed neighborhood
in Assumption~\ref{ass:counts}. Choose a fixed $K$ large enough
that this sum is less than $c/(4\sqrt T)$ for large $T$.
Thus $H_n(q_n+KT^{-1/2})>\tau$ by a margin of order $T^{-1/2}$.

For the left endpoint, use $F_n(q_n-)<\tau$ instead:
\begin{align*}
 H_n(x)-\tau
 ={}&F_n(q_n-)-\tau
 -\frac1n\sum_{\theta_i<q_n}\overline\Phi\{u_i(x)/\sqrt2\}
 +\frac1n\sum_{\theta_i\ge q_n}\Phi\{u_i(x)/\sqrt2\}\\
 \le{}&-\frac1n\sum_{q_n-T^{-1/2}\le\theta_i<q_n}
                   \overline\Phi\{u_i(x)/\sqrt2\}
 +\frac1n\sum_{\theta_i\ge q_n}\Phi\{u_i(x)/\sqrt2\}.
\end{align*}
At $x=q_n-KT^{-1/2}$, the magnitude of the negative sum is
at least $c/(2\sqrt T)$; the positive sum has the same tail
bound as above. Hence fixed $K_0,c_0>0$ satisfy
\begin{equation}\label{eq:rootbracket}
 \begin{aligned}
 H_n(q_n-K_0T^{-1/2})&\le\tau-c_0T^{-1/2},\qquad H_n(q_n+K_0T^{-1/2})\ge\tau+c_0T^{-1/2}.
 \end{aligned}
\end{equation}
Because $H_n(\tilde q_n)=\tau$ and $H_n$ is increasing,
\[
 q_n-K_0T^{-1/2}\le\tilde q_n\le q_n+K_0T^{-1/2}.
\]
This proves the deterministic bound, including when other
coefficients tie at $q_n$.

\medskip\noindent\textit{Step 2: bound $\widehat F(\tilde q_n)-\tau$ and its slope.}
On each fixed $KT^{-1/2}$ neighborhood of $q_n$,
Lemma~\ref{lem:comparison} bounds each CDF approximation error
by $CT^{-1/2}(1+|u_i(x)|)^{-3}$. Averaging and applying
Lemma~\ref{lem:envelope} gives
\begin{equation}\label{eq:meanerror}
 |\E G_n(x)-H_n(x)|
 \le\frac C{\sqrt T}\frac1n\sum_i(1+|u_i(x)|)^{-3}
 \le CT^{-1}.
\end{equation}
For the derivative, each term has the additional factor
$\sqrt T/\sigma_i=O(\sqrt T)$, so the same calculation gives
\begin{equation}\label{eq:deriverror}
 |\E G_n'(x)-H_n'(x)|\le CT^{-1/2}.
\end{equation}
Independence across units likewise gives
\begin{equation}\label{eq:derivvariance}
 \Var\{G_n(\tilde q_n)\}\le\frac C{n\sqrt T},\qquad
 \Var\{G_n'(\tilde q_n)\}\le\frac{C\sqrt T}{n}.
\end{equation}
Since $H_n(\tilde q_n)=\tau$, decompose
\begin{align*}
 \widehat F(\tilde q_n)-\tau
 ={}&[\widehat F(\tilde q_n)-G_n(\tilde q_n)]
   +[G_n(\tilde q_n)-\E G_n(\tilde q_n)]+[\E G_n(\tilde q_n)-H_n(\tilde q_n)].
\end{align*}
Lemma~\ref{lem:perturb}, \eqref{eq:derivvariance}, and
\eqref{eq:meanerror} bound these terms by $O_p(T^{-1})$,
$O_p(n^{-1/2}T^{-1/4})$, and $O(T^{-1})$, respectively.
The derivative decomposition gives
\begin{align*}
 \widehat F(\tilde q_n)-\tau
 &=O_p(n^{-1/2}T^{-1/4}+T^{-1}),\\
 \widehat F'(\tilde q_n)-H_n'(\tilde q_n)
 &=O_p(T^{-1/2}+n^{-1/2}T^{1/4})=o_p(1).
\end{align*}

\medskip\noindent\textit{Step 3: prove
$\widehat q-\tilde q_n=O_p(n^{-1/2}T^{-1/4}+T^{-1})$.}
Lemma~\ref{lem:envelope} gives $H_n'(\tilde q_n)\ge c>0$.
Lemma~\ref{lem:perturb} bounds the local second derivative by
$O_p(\sqrt T)$, and
\[
 \sqrt T(n^{-1/2}T^{-1/4}+T^{-1})
 =n^{-1/2}T^{1/4}+T^{-1/2}\to0.
\]
For any fixed $M$, Taylor's theorem therefore gives
\begin{align*}
 &\widehat F\{\tilde q_n\pm M(n^{-1/2}T^{-1/4}+T^{-1})\}-\tau=\widehat F(\tilde q_n)-\tau
 \pm M(n^{-1/2}T^{-1/4}+T^{-1})
           \{H_n'(\tilde q_n)+o_p(1)\}.
\end{align*}
For any $\varepsilon>0$, take $M$ large enough that the minus
and plus expressions have opposite signs with probability at least
$1-\varepsilon$ for all sufficiently large samples. Monotonicity
of $\widehat F$ places its root between these endpoints. Thus
\begin{equation}\label{eq:rootrate}
 \widehat q-\tilde q_n
 =O_p(n^{-1/2}T^{-1/4}+T^{-1}).
\end{equation}
Together with Step 1, this proves the proposition.
\end{proof}

\subsection{Proof of Proposition~\ref{prop:impossibility}}
\label{sec:proofimpossibility}
\begin{proof}
We construct two sequences in $\mathcal{F}$ whose observation laws
become arbitrarily close, although the sampling distributions in
\eqref{eq:samplinglaw} have different continuous limits. The
construction works for every fixed interior quantile and every
growth sequence specified in the proposition.

\medskip\noindent\textit{Step 1: construct $F_0\in\mathcal{F}$ satisfying Assumption~\ref{ass:balance}.}
Write $k=\lceil n\tau\rceil$ and
$m=\lfloor\zeta n\rfloor$, where
$0<\zeta<\min(\tau,1-\tau)/2$ is fixed. The baseline array is
\[
 \begin{gathered}
 \underbrace{-2,\ldots,-2}_{k-1-m},\quad
 \left\{-1/(2\sqrt T)-j/n:j=1,\ldots,m\right\},\quad
 0,\\
 \left\{1/(2\sqrt T)+j/n:j=1,\ldots,m\right\},\quad
 \underbrace{2,\ldots,2}_{n-k-m}.
 \end{gathered}
\]
The multiplicities are nonnegative for all sufficiently large $n$.
Label the zero coefficient as unit $k$ and keep all unit labels
fixed. Observe independent $Y_{it}\sim N(\theta_i,1)$ and use
$\widehat\theta_i=\bar Y_i$ and known
$\widehat\sigma_i=\sigma_i=1$. Denote this model by $F_0$ and its observation law by $P_0$. Its target is $q_n(F_0)=0$.

Assumption~\ref{ass:first} holds with zero remainder. Each strict
side of zero contains $n/(2\sqrt T)+O(1)$ coefficients within distance
$T^{-1/2}$. For $T^{-1/2}\le x\le\epsilon_0$, with a sufficiently small fixed
$\epsilon_0>0$, the total count within distance $x$ is at most
$2nx+3$. Thus Assumption~\ref{ass:counts} also holds.
The reflected local pairs cancel at zero, giving
\begin{equation}\label{eq:baselinebalance}
 H_n^{(0)}(0)=\frac{k-1/2}{n}+O(e^{-T}),\qquad
 |H_n^{(0)}(0)-\tau|\le\frac1{2n}+O(e^{-T}).
\end{equation}
Here the superscript identifies the model. Since $n/\sqrt T\to\infty$
and $n/T^{3/2}\to0$, the last bound is
$o(n^{-1/2}T^{-1/4})$. The baseline satisfies
Assumption~\ref{ass:balance}.

Define its deterministic first-order standard error by
\begin{equation}\label{eq:oracleshift}
 a_{nT}
 =\frac{\sqrt{\Var_0\{\widehat F(0)\}}}{(H_n^{(0)})'(0)}
 \asymp n^{-1/2}T^{-1/4}.
\end{equation}
The order follows from Lemma~\ref{lem:envelope}. In particular,
$\sqrt T\,a_{nT}\asymp n^{-1/2}T^{1/4}\to0$.
Theorem~\ref{thm:wald} and the variance consistency argument in
\eqref{eq:varianceconsistent} yield, under $P_0$,
\begin{equation}\label{eq:baselinelaw}
 \frac{\widehat q}{\widehat{\mathrm{se}}(\widehat q)}
 \Rightarrow N(0,1),\qquad
 \frac{\widehat{\mathrm{se}}(\widehat q)}{a_{nT}}\to_p1.
\end{equation}
For the second conclusion, local balance gives
$\tilde q_n=o(a_{nT})$. Moving the variance and slope evaluations
from $\tilde q_n$ to zero has a negligible relative effect,
because $\sqrt T\,a_{nT}\to0$.

\medskip\noindent\textit{Step 2: construct $F_1$ with $q_n(F_1)-q_n(F_0)=a_{nT}$ and $\sup_A|P_1(A)-P_0(A)|\to0$.}
Let $F_1$ change only the central coefficient from zero to
$a_{nT}$. For large samples, $a_{nT}<1/(4\sqrt T)$, so it remains between
the same two neighbors and
\begin{equation}\label{eq:two_targets}
 q_n(F_0)=0,\qquad q_n(F_1)=a_{nT}.
\end{equation}
The strict left and right counts within $T^{-1/2}$ of the new target are
$n(1/(2\sqrt T)-a_{nT})+O(1)$ and $n(1/(2\sqrt T)+a_{nT})+O(1)$.
The upper count bound remains $2nx+3$. Hence $F_0,F_1\in\mathcal{F}$ with the same constants.
Write $P_1$ for the observation law under $F_1$.

Only the observations on unit $k$ differ. Their likelihood ratio
is
\[
 \frac{dP_1}{dP_0}
 =\exp\{\sqrt T a_{nT}Z-Ta_{nT}^2/2\},\qquad
 Z=\sqrt T\bar Y_k\sim N(0,1)\quad\text{under }P_0.
\]
Consequently, the largest difference in probability assigned to
any data event is
\begin{equation}\label{eq:distributionTV}
 \sup_A|P_1(A)-P_0(A)|
 =2\Phi(\sqrt T a_{nT}/2)-1\to0.
\end{equation}
This compares the laws of all the observations, so the conclusion
applies to any estimator computed from them.

The alternative does not satisfy local balance. Changing one
coefficient changes $H_n(x)$ by at most $1/n$ at every $x$.
Using \eqref{eq:baselinebalance}, the derivative bounds in
Lemma~\ref{lem:envelope}, and $\sqrt T\,a_{nT}\to0$ gives
\[
 H_n^{(1)}(a_{nT})-\tau
 =(H_n^{(0)})'(0)a_{nT}+o(a_{nT})
 \asymp n^{-1/2}T^{-1/4}.
\]
The imbalance is therefore of the same order as the standard
error, rather than negligible relative to it.

\medskip\noindent\textit{Step 3: prove $L_0(0)\to1/2$ and $L_1(0)\to\Phi(1)$.}
By \eqref{eq:distributionTV}, both conclusions in
\eqref{eq:baselinelaw}, viewed as statements about functions of
the data, also hold under $P_1$. But the target has changed, so
under $P_1$,
\[
 \frac{\widehat q-q_n(F_1)}{\widehat{\mathrm{se}}(\widehat q)}
 =\frac{\widehat q}{\widehat{\mathrm{se}}(\widehat q)}
  -\frac{a_{nT}}{\widehat{\mathrm{se}}(\widehat q)}
 \Rightarrow N(-1,1).
\]
Thus, writing $L_j=L_{nT,F_j}$ for this proof,
\begin{equation}\label{eq:lawseparation}
 L_0(0)\to\tfrac12,\qquad L_1(0)\to\Phi(1).
\end{equation}
The limits differ by more than $1/4$. Both limiting distributions
are continuous; the obstruction therefore also concerns their
limiting laws, not merely a discontinuity in a finite-sample CDF.

\medskip\noindent\textit{Step 4: prove the probability bound in \eqref{eq:distributionimpossibility}.}
Take $\varepsilon=1/8$ and let
$A_j=\{\|\widehat L_{nT}-L_j\|_\infty\le\varepsilon\}$,
$j=0,1$. By \eqref{eq:lawseparation}, these two events are disjoint
for all sufficiently large samples, because
\[
 A_0\cap A_1\subseteq
 \{\|L_0-L_1\|_\infty\le2\varepsilon\},\qquad
 \liminf\|L_0-L_1\|_\infty\ge\Phi(1)-\tfrac12>2\varepsilon.
\]
Therefore
\[
 P_0(A_0)+P_1(A_1)
 \le P_0(A_0)+P_0(A_1)+\sup_A|P_1(A)-P_0(A)|
 \le1+\sup_A|P_1(A)-P_0(A)|.
\]
It follows that
\[
 \max_{j=0,1}
 P_j\{\|\widehat L_{nT}-L_j\|_\infty>\varepsilon\}
 \ge\frac{1-\sup_A|P_1(A)-P_0(A)|}{2}\to\frac12.
\]
Taking the supremum over $\mathcal{F}$ and then the lower limit
proves \eqref{eq:distributionimpossibility}.
\end{proof}

\subsection{Proof of Proposition~\ref{prop:gmm}}\label{sec:proofgmm}
\begin{proof}
We first bound the parameter estimation error, then its linear-expansion
remainder, and finally the error in the estimated standard deviation.
Suppress $i$ within this proof; all constants are uniform across units.
Write $\bm\delta=\widehat{\bm\beta}-\bm\beta$ and use
$\bar{\bm g}'(\bm b)$ for the sample Jacobian at $\bm b$.

\medskip\noindent\textit{Step 1: prove $\E\|\bm\delta\|^k\le C_kT^{-k/2}$.}
Bounded independent observations give, for every fixed even $k\ge2$,
\[
 \E\|\bar{\bm g}(\bm\beta)\|^k
 +\E\|\bar{\bm g}'(\bm\beta)-\bm G\|^k\le C_kT^{-k/2}.
\]
In particular,
\begin{equation}\label{eq:gmmmoments}
 \begin{aligned}
 \E\|\bar{\bm g}(\bm\beta)\|^2
 +\E\|\bar{\bm g}'(\bm\beta)-\bm G\|^2&\le C/T,\\
 \E\|\bar{\bm g}(\bm\beta)\|^4
 +\E\|\bar{\bm g}'(\bm\beta)-\bm G\|^4&\le C/T^2.
 \end{aligned}
\end{equation}
To justify using the first-order condition locally, bounded derivatives
and Hoeffding's inequality on a fixed finite covering of $\mathcal{B}$
give, for each fixed $a>0$ and $j=0,1,2$,
\[
 P\left\{\sup_{\bm b\in\mathcal{B}}
  \|\partial^j\bar{\bm g}(\bm b)
       -\E\partial^j\bm g(\bm Z_t,\bm b)\|>a\right\}
 \le C_a e^{-c_aT}.
\]
Identification places the minimizer near $\bm\beta$, and the
population criterion has Hessian $2\bm G^{\top}\bm W\bm G$ at
$\bm\beta$. Thus, for fixed $\delta_0,c_0>0$, there is an event
$\mathcal{E}$ satisfying
\[
 \begin{gathered}
 P(\mathcal{E}^c)\le Ce^{-cT},\qquad
 \|\bm\delta\|\le\delta_0\quad\text{on }\mathcal{E},\\
 \inf_{\|\bm b-\bm\beta\|\le\delta_0}
 \lambda_{\min}\!\left(
 \partial^2_{\bm b}\{\bar{\bm g}(\bm b)^{\top}
                         \bm W\bar{\bm g}(\bm b)\}\right)
 \ge c_0\quad\text{on }\mathcal{E}.
 \end{gathered}
\]
The ball lies inside $\mathcal{B}$. The sample gradient at the truth
is $2\bar{\bm g}'(\bm\beta)^{\top}\bm W\bar{\bm g}(\bm\beta)$.
Integrating the Hessian between the truth and the minimizer gives
$\|\bm\delta\|\le C\|\bar{\bm g}(\bm\beta)\|$ on $\mathcal{E}$.
Compactness bounds $\bm\delta$ on its complement. Consequently,
\[
 \E\|\bm\delta\|^k
 \le C_k\E\|\bar{\bm g}(\bm\beta)\|^k+C_kP(\mathcal{E}^c)
 \le C_kT^{-k/2}.
\]
For later use, this includes
\begin{equation}\label{eq:gmmdelta}
 \E\|\bm\delta\|^2\le C/T,\qquad
 \E\|\bm\delta\|^4\le C/T^2.
\end{equation}

\medskip\noindent\textit{Step 2: prove $\E\|\bm R\|^2\le C/T^2$.}
On $\mathcal{E}$, expand
$\bar{\bm g}'(\widehat{\bm\beta})^{\top}\bm W
 \bar{\bm g}(\widehat{\bm\beta})=\bm0$ at $\bm\beta$.
Its linear terms are
$\bm G^{\top}\bm W\bar{\bm g}(\bm\beta)
 +\bm G^{\top}\bm W\bm G\bm\delta$.
Bounded derivatives and the bound on $\bm\delta$ in Step 1 imply
\[
 \|\bm\delta+\bm A\bar{\bm g}(\bm\beta)\|
 \le C\left\{
 \|\bar{\bm g}'(\bm\beta)-\bm G\|
                 \|\bar{\bm g}(\bm\beta)\|
 +\|\bar{\bm g}(\bm\beta)\|^2\right\},
 \qquad \bm A=(\bm G^{\top}\bm W\bm G)^{-1}\bm G^{\top}\bm W.
\]
Define $\bm R=\bm\delta+\bm A\bar{\bm g}(\bm\beta)$.
Cauchy-Schwarz and \eqref{eq:gmmmoments} give
\begin{align*}
 \E\|\bm R\|^2
 &\le C\{\E\|\bar{\bm g}'(\bm\beta)-\bm G\|^4\}^{1/2}
         \{\E\|\bar{\bm g}(\bm\beta)\|^4\}^{1/2}+C \E\|\bar{\bm g}(\bm\beta)\|^4+CP(\mathcal{E}^c)
 \le C/T^2.
\end{align*}
Thus
\begin{equation}\label{eq:gmmexpansion}
 \widehat{\bm\beta}-\bm\beta
 =-\bm A\bar{\bm g}(\bm\beta)+\bm R,
 \qquad \E\|\bm R\|^2\le C/T^2.
\end{equation}
Multiplying by $\bm a^{\top}$ gives the influence in
\eqref{eq:gmminfluence}, with $R=\bm a^{\top}\bm R$.
The influence $\psi$ is centered and has a bounded fourth moment
because $\E\bm g(\bm Z_t,\bm\beta)=\bm0$ and $\bm g,\bm A$ are bounded.

\medskip\noindent\textit{Step 3: prove $\E(\widehat\sigma-\sigma)^2\le C/T$.}
Use $\widehat{\bm G}$ and $\widehat{\bm\Omega}$ as defined in the
main appendix, evaluated at $\widehat{\bm\beta}$.
Bounded derivatives and sample covariance centering give
\begin{align*}
 \|\widehat{\bm G}-\bm G\|+\|\widehat{\bm\Omega}-\bm\Omega\|
 \le{}& C\|\bm\delta\|+\|\bar{\bm g}'(\bm\beta)-\bm G\|
 +C\|\bar{\bm g}(\bm\beta)\|^2\\
 &+\left\|\frac1T\sum_t\bm g(\bm Z_t,\bm\beta)
                  \bm g(\bm Z_t,\bm\beta)^{\top}-\bm\Omega\right\|.
\end{align*}
The squared expectation of the right side is $O(T^{-1})$.
On $\mathcal{E}$, the sandwich formula and square-root map have
bounded derivatives, so
\[
 |\widehat\sigma-\sigma|
 \le C\{\|\widehat{\bm G}-\bm G\|
                +\|\widehat{\bm\Omega}-\bm\Omega\|\}.
\]
The bounded fallback and truncation in \eqref{eq:gmmscale}
control the complement. Hence
\[
 \E(\widehat\sigma-\sigma)^2\le C/T+CP(\mathcal{E}^c)\le C/T.
\]
The variance bounds in the proposition supply the remaining parts
of Assumption~\ref{ass:first}. Smooth M-estimators follow by using
their score equations under the stated identification conditions.
\end{proof}

\subsection{Proof of Proposition~\ref{prop:gmmbias}}\label{sec:proofgmmbias}
\begin{proof}
We derive the quadratic term in the coefficient estimator, then
the linear term in its estimated variance. The resulting observation
functions determine $\bm v_{it}$ and $\bm D_i$; sample substitution
then gives the required bounds for $\widehat b_i$ and
$\widehat\gamma_i$.

\medskip\noindent\textit{Step 1: derive the quadratic expansion with
$\E\|\bm r_i\|^2\le C/T^3$.}
For a just-identified estimating equation, write
\[
 \bar{\bm m}_i(\bm\xi)=\frac1T\sum_t\bm m_i(\bm Z_{it},\bm\xi),
 \qquad \bar{\bm m}_i(\widehat{\bm\xi}_i)=\bm0.
\]
All derivatives below are evaluated at the true $\bm\xi_i$.
Let $\bm J_i=\E\bm m_i'$,
$\bm\delta_i=\widehat{\bm\xi}_i-\bm\xi_i$, and
$\bm\delta_{i1}=-\bm J_i^{-1}\bar{\bm m}_i$.
Taylor's theorem first gives
\[
 \bm0=\bar{\bm m}_i+\bar{\bm m}'_i\bm\delta_i
       +\tfrac12\bar{\bm m}''_i[\bm\delta_i,\bm\delta_i]
       +O(\|\bm\delta_i\|^3).
\]
Solving for $\bm\delta_i$, and replacing it by $\bm\delta_{i1}$
only in the quadratic terms, yields
\begin{equation}\label{eq:Msecond}
 \begin{split}
 \widehat{\bm\xi}_i-\bm\xi_i
 ={}&\bm\delta_{i1}
 -\bm J_i^{-1}(\bar{\bm m}'_i-\bm J_i)\bm\delta_{i1}-\tfrac12\bm J_i^{-1}\E\bm m_i''
       [\bm\delta_{i1},\bm\delta_{i1}]+\bm r_i.
 \end{split}
\end{equation}
Here $\bm m_i''[\bm u,\bm v]$ denotes the second derivative
applied to the two vectors.

The concentration and moment argument in Step 1 of the preceding
proof gives events $\mathcal{E}_i$ with
$\sup_iP(\mathcal{E}_i^c)\le Ce^{-cT}$ and, for $k=2,4,6$,
\[
 \sup_i \E\left\{
 \|\bm\delta_i\|^k+\|\bar{\bm m}'_i-\bm J_i\|^k
 +\|\bar{\bm m}''_i-\E\bm m_i''\|^k\right\}
 \le C_kT^{-k/2}.
\]
On $\mathcal{E}_i$, bounded inverse Jacobians and derivatives imply
\begin{align*}
 \|\bm\delta_i-\bm\delta_{i1}\|
 &\le C\{\|\bar{\bm m}'_i-\bm J_i\|\|\bm\delta_i\|
                          +\|\bm\delta_i\|^2\},\\
 \|\bm r_i\|
 &\le C\bigl\{\|\bm\delta_i\|^3
 +\|\bar{\bm m}'_i-\bm J_i\|
                       \|\bm\delta_i-\bm\delta_{i1}\|+\|\bm\delta_i\|\|\bm\delta_i-\bm\delta_{i1}\|
 +\|\bar{\bm m}''_i-\E\bm m_i''\|\|\bm\delta_i\|^2\bigr\}.
\end{align*}
Substitution makes each term a product of three estimation errors.
For example, H\"older's inequality gives
\[
 \E\{\|\bar{\bm m}'_i-\bm J_i\|^4\|\bm\delta_i\|^2\}
 \le\{\E\|\bar{\bm m}'_i-\bm J_i\|^6\}^{2/3}
     \{\E\|\bm\delta_i\|^6\}^{1/3}
 \le C/T^3.
\]
The other products have the same bound. Boundedness on
$\mathcal{E}_i^c$ then gives
\[
 \sup_i \E\|\bm r_i\|^2
 \le C/T^3+C\sup_iP(\mathcal{E}_i^c)\le C/T^3.
\]

For overidentified GMM, apply this calculation to
$\bm\xi=(\bm\beta^{\top},\bm\lambda^{\top})^{\top}$ and
\[
 \bm m_i(\bm Z,\bm\beta,\bm\lambda)=
 \begin{pmatrix}
 [\partial_{\bm\beta^{\top}}\bm g_i(\bm Z,\bm\beta)]^{\top}
                      \bm W_i\bm\lambda\\
 \bm g_i(\bm Z,\bm\beta)-\bm\lambda
 \end{pmatrix}.
\]
At the truth $\bm\lambda=\bm0$, and
\[
 \bm J_i=
 \begin{pmatrix}\bm0&\bm G_i^{\top}\bm W_i\\
                 \bm G_i&-\bm I\end{pmatrix}.
\]
The eigenvalue conditions in Proposition~\ref{prop:gmm} bound
$\bm J_i^{-1}$. The selected root has
$\widehat{\bm\lambda}_i=\bar{\bm g}_i(\widehat{\bm\beta}_i)$;
its first block is the GMM first-order condition. Its localization
and moment bounds follow from the preceding proof, so no new
identification assumption is needed for the augmented equation.
Taking the contrast $(\bm a^{\top},\bm0^{\top})^{\top}$ in
\eqref{eq:Msecond} gives \eqref{eq:secondtheta}. The two quadratic
terms determine the symmetric matrix $\bm D_i$.

\medskip\noindent\textit{Step 2: derive the variance expansion with
$\E w_i^2\le C/T^2$.}
On $\mathcal{E}_i$, the sandwich variance is a smooth function
$S(\bm\xi,\bm M)$, where $\bm M$ collects the required moments
of bounded observation functions $\bm h(\bm Z,\bm\xi)$.
At the population values, its first derivative gives
\[
 \begin{aligned}
 \bm u_{it}&=-\bm J_i^{-1}\bm m_i(\bm Z_{it},\bm\xi_i),\\
 \eta_{it}&=S_{\bm\xi}\bm u_{it}
 +S_{\bm M}\{\bm h(\bm Z_{it},\bm\xi_i)-\E\bm h
                         +(\E\bm h_{\bm\xi})\bm u_{it}\}.
 \end{aligned}
\]
Thus $\E\eta_{it}=0$, $|\eta_{it}|\le C$, and
\[
 \widehat\sigma_i^2-\sigma_i^2
 =\frac1T\sum_t\eta_{it}+w_i.
\]
The Taylor remainder is quadratic in the parameter and sample-moment
errors. Their fourth moments are $O(T^{-2})$, giving
\[
 \sup_i\E w_i^2\le C/T^2+C\sup_iP(\mathcal{E}_i^c)\le C/T^2.
\]
The truncation in \eqref{eq:gmmscale} affects only the exponentially
unlikely complement of a sufficiently small fixed neighborhood.

\medskip\noindent\textit{Step 3: verify the conditions on
$\bm v_{it}$, $\bm D_i$, $\widehat b_i$, and $\widehat\gamma_i$.}
Choose the first two coordinates of $\bm v_{it}$ as
$\psi_{it}$ and $\eta_{it}$, and append the centered observations
in $\bm m_i$ and $\bm m_i'-\bm J_i$ needed for
\eqref{eq:Msecond}. Then
\[
 \E\bm v_{it}=\bm0,\qquad
 \|\bm v_{it}\|+\|\bm D_i\|\le C,
 \qquad
 Q_i=\tfrac12\bm S_i^{\top}\bm D_i\bm S_i,
\]
with fixed dimension and $\bm S_i=T^{-1/2}\sum_t\bm v_{it}$,
as in Assumption~\ref{ass:second}.

Use centered sample observation functions and the sample analogues
of \eqref{eq:bcconstants}. On $\mathcal{E}_i$, these formulas
have bounded derivatives in the estimated parameters and moments.
Each argument has mean-square error $O(T^{-1})$, so the mean
value theorem gives
\[
 \sup_i\E\{(\widehat b_i-b_i)^2+
                  (\widehat\gamma_i-\gamma_i)^2\}
 \le C/T+C\sup_iP(\mathcal{E}_i^c)\le C/T.
\]
The bounded sample versions control the complement. Steps 1-3
verify Assumption~\ref{ass:second} for the stated M/GMM class.
\end{proof}

\subsection{An interval-length lower bound}\label{sec:lengthlower}
The following result complements Proposition~\ref{prop:impossibility}.
It concerns the length required for uniform coverage. 

\begin{extensionprop}[Lower bound on interval length]
\label{prop:lengthlower}
Fix $\tau\in(0,1)$ and $0<\alpha<1/2$. Let $I_{nT}$ be any
confidence interval satisfying
\[
 \liminf_{n,T\to\infty}\inf_{F\in\mathcal{F}}
 P_{nT}\{q_n(F)\in I_{nT}\}\ge1-\alpha.
\]
Then there are constants $d,\eta>0$ such that
\begin{equation}\label{eq:lowerbound}
 \liminf_{n,T\to\infty}\sup_{F\in\mathcal{F}}
 P_{nT}\{\sqrt T\,|I_{nT}|\ge d\}\ge\eta.
\end{equation}
Within the growth range of Theorem~\ref{thm:wald}, this lower bound
holds along a sequence satisfying Assumption~\ref{ass:balance}.
\end{extensionprop}

\begin{proof}
We use two models whose targets differ by $dT^{-1/2}$ but whose
data laws remain close. Coverage of both targets then forces the
interval to have that length with positive probability.
Use the baseline model $F_0$, with observation law $P_0$, from
the proof of Proposition~\ref{prop:impossibility}. For a fixed $0<d\le1/4$,
let $F_d$ change only its central coefficient from zero to $dT^{-1/2}$,
and write $P_d$ for its observation law. The target remains the central coefficient,
so $q_n(F_0)=0$ and $q_n(F_d)=dT^{-1/2}$.
Assumption~\ref{ass:first} holds exactly, and the strict left and
right counts within $T^{-1/2}$ of $dT^{-1/2}$ are
\[
 (1/2-d)n/\sqrt T+O(1),\qquad (1/2+d)n/\sqrt T+O(1).
\]
The upper count bound is again $2nx+3$. Both models therefore
belong to $\mathcal{F}$ with common constants.

The likelihood ratio is $\exp(dZ-d^2/2)$, where
$Z=\sqrt T\bar Y_k\sim N(0,1)$ under $P_0$. Hence
\begin{equation}\label{eq:TV}
 \sup_A|P_d(A)-P_0(A)|=2\Phi(d/2)-1.
\end{equation}
Choose $d$ small enough that
$\eta:=1-2\alpha-\{2\Phi(d/2)-1\}>0$.
If an interval contains both zero and $dT^{-1/2}$, its length is at
least $dT^{-1/2}$. Transferring coverage of $dT^{-1/2}$ from $P_d$ to $P_0$
using \eqref{eq:TV} gives
\begin{align*}
 P_0\{|I_{nT}|\ge dT^{-1/2}\}
 &\ge P_0\{0\in I_{nT},\ dT^{-1/2}\in I_{nT}\}\\
 &\ge P_0\{0\in I_{nT}\}+P_d\{dT^{-1/2}\in I_{nT}\}
       -1-\{2\Phi(d/2)-1\}.
\end{align*}
The uniform coverage requirement makes the lower limit of the
right side at least $\eta$. Thus
\begin{equation}\label{eq:honest_length_proof}
 \liminf_{n,T\to\infty}
 P_0\{\sqrt T\,|I_{nT}|\ge d\}\ge\eta>0,
\end{equation}
which implies \eqref{eq:lowerbound}. The count arguments require
only $n/\sqrt T\to\infty$. If also $n/T^{3/2}\to0$,
\eqref{eq:baselinebalance} shows that $P_0$ satisfies local balance,
proving the final assertion.
\end{proof}

\clearpage
\section{Supporting Lemmas}\label{sec:lemmas}
The notation $W_i,u_i,G_n$ is defined at the start of
Section~\ref{sec:mainproofs}. Constants may change between displays,
but depend only on the common bounds in the assumptions and any fixed
neighborhood size.

\begin{lemma}[Bounds for weighted sums near the quantile]\label{lem:envelope}
Under Assumption \ref{ass:counts} and bounded positive $\sigma_i$, for every fixed $K<\infty$ and
$a>1$,
\begin{equation}\label{eq:envelope}
 \sup_{|x-q_n|\le K/\sqrt T}\frac1n\sum_i(1+|u_i(x)|)^{-a}
       =O(T^{-1/2}).
\end{equation}
The same conclusion holds for Gaussian envelopes. Throughout this window,
\begin{equation}\label{eq:slope}
 0<c_K\le H_n'(x)\le C_K,\qquad |H_n''(x)|\le C_K\sqrt T.
\end{equation}
If $v(u)=\Var\{\Phi(u-Z)\}$ for standard normal $Z$, then it holds uniformly that
\begin{equation}\label{eq:varianceorder}
 \frac1{n^2}\sum_i v\{u_i(x)\}\asymp\frac1{n\sqrt T}
\end{equation}
\end{lemma}
\begin{proof}
We first prove the weighted-sum bound, then apply it to the
derivatives and variance. For $|x-q_n|\le KT^{-1/2}$,
\[
 |u_i(x)-u_i(q_n)|\le K/\sigma_i\le C_K,
 \qquad
 (1+|u_i(x)|)^{-a}\le C_K(1+|u_i(q_n)|)^{-a}.
\]
It therefore suffices to bound the sum at $q_n$.
There are at most $Cn/\sqrt T$ coefficients within $T^{-1/2}$. In the
annulus $2^j/\sqrt T<|\theta_i-q_n|\le2^{j+1}/\sqrt T\le\epsilon$,
the count is at most $Cn2^{j+1}/\sqrt T$ and the weight at most $C2^{-ja}$.
After division by $n$, these groups and the coefficients outside
the fixed neighborhood give
\[
 \frac1n\sum_i(1+|u_i(q_n)|)^{-a}
 \le CT^{-1/2}\left\{1+\sum_{j\ge0}2^{-j(a-1)}\right\}
       +CT^{-a/2}=O(T^{-1/2}),
\]
because $a>1$. A polynomial times a Gaussian density is bounded
by a constant multiple of the same polynomial bound.

Differentiation of $H_n$ once gives
$\frac{\sqrt T}{\sqrt2\sigma_i}\phi\{u_i(x)/\sqrt2\}$ per term. The upper bound follows
from \eqref{eq:envelope}. At least $cn/\sqrt T$ coefficients lie within
$T^{-1/2}$ of $q_n$. For every $x$ in the displayed window their standardized
distances lie in a fixed bounded interval. Each density contribution
is therefore at least $c_K\sqrt T$, so
\[
 H_n'(x)\ge\frac1n\frac{cn}{\sqrt T}\,c_K\sqrt T\ge c_K'>0.
\]
For the second derivative, the same weighted-sum bound gives
$|H_n''(x)|\le CT\cdot O(T^{-1/2})=O(\sqrt T)$.

The continuous function $v$ is positive on compact sets.
On $|Z|\le |u|/2$, the difference between $\Phi(u-Z)$ and
$\ind\{u\ge0\}$ is at most $\overline\Phi(|u|/2)$. Thus
\[
 v(u)\le \E[\Phi(u-Z)-\ind\{u\ge0\}]^2
 \le P(|Z|>|u|/2)+\overline\Phi(|u|/2)^2\le Ce^{-cu^2}.
\]
The local lower count and the upper weighted-sum bound now give
\[
 \frac{c_K}{n^2}\frac n{\sqrt T}
 \le\frac1{n^2}\sum_i v\{u_i(x)\}
 \le\frac{C_K}{n^2}\frac n{\sqrt T},
\]
which is \eqref{eq:varianceorder}.
\end{proof}

\begin{lemma}[Normal approximation for smooth functions]\label{lem:comparison}
Under Assumption~\ref{ass:first}, for any bounded absolutely continuous $g$ whose derivative is bounded by
$C(1+|v|^m)e^{-cv^2}$,
\begin{equation}\label{eq:comparison}
 |\E g(u-W_i)-\E g(u-Z)|\le\frac{C}{\sqrt T(1+|u|^3)}.
\end{equation}
Continuous piecewise smooth functions with the same derivative envelope
also qualify. For a polynomial times a Gaussian envelope $K_0$,
\begin{equation}\label{eq:momentenvelope}
 \E K_0(u-W_i)\le Ce^{-c'u^2}+
                  \frac{C}{\sqrt T(1+|u|^3)}.
\end{equation}
\end{lemma}
\begin{proof}
The independent centered summands defining $W_i$ satisfy
\[
 \Var(W_i)=1,\qquad
 \sum_t \E\left|\frac{\psi_{it}}{\sigma_i\sqrt T}\right|^3
 \le\frac C{\sqrt T}.
\]
The nonuniform Berry-Esseen inequality (Nagaev-Bikelis;
see \citealp{bobkov2023}) therefore gives
\begin{equation}\label{eq:BE}
 |P(W_i\le z)-\Phi(z)|\le\frac{C}{\sqrt T(1+|z|^3)}.
\end{equation}
Integration by parts and \eqref{eq:BE} give
\[
 |\E g(u-W_i)-\E g(u-Z)|
 \le\frac C{\sqrt T}\int_{\mathbb R}
             \frac{|g'(u-z)|}{1+|z|^3}\,dz.
\]
For $|u|>2$, split the integral into $|u-z|\le|u|/2$ and
its complement. On the first set $|z|\ge|u|/2$; on the
second set the Gaussian derivative bound applies. Hence
\[
 \int\frac{|g'(u-z)|}{1+|z|^3}\,dz
 \le\frac C{1+|u|^3}\int|g'(v)|\,dv
       +\int_{|v|>|u|/2}|g'(v)|\,dv
 \le\frac C{1+|u|^3}.
\]
The same bound for $|u|\le2$ follows from integrability of $g'$.
Bounded $g$ and vanishing CDF differences remove the boundary terms. Split at nondifferentiability points for continuous piecewise
smooth functions. Applying the result to $K_0$, whose normal expectation
is bounded by another Gaussian envelope, gives the second assertion.
\end{proof}

\begin{lemma}[Effects of remainders and estimated standard errors]\label{lem:perturb}
Let
\(
 \rho_i=\frac{\sqrt T\,R_i}{\sigma_i},\ 
 e_i=\frac{\widehat\sigma_i}{\sigma_i}-1.
\) Let $g$ be any of $\Phi(v)$, $\Phi(-|v|)$, $\Phi(-|v|)^2$,
$\Phi(v)\pm2\Phi(-|v|)$, or their squares. Under Assumption \ref{ass:first}-\ref{ass:counts}, for each
fixed $K$,
\begin{equation}\label{eq:perturb}
 \E\sup_{|x-q_n|\le K/\sqrt T}\frac1n\sum_i
 \left|g\left(\frac{u_i(x)-W_i-\rho_i}{1+e_i}\right)
                   -g\{u_i(x)-W_i\}\right|=O(T^{-1}).
\end{equation}
The corresponding mean-square difference for each summand is $O(T^{-1})$.
For $g=\Phi$, differentiating with respect to $x$ gives an aggregate
expected supremum difference $O(T^{-1/2})$. Also
\begin{equation}\label{eq:secondderiv}
 \sup_{|x-q_n|\le K/\sqrt T}
       (|G_n''(x)|+|\widehat F''(x)|)=O_p(\sqrt T),\qquad
 \sup_{|x-q_n|\le K/\sqrt T}\widehat F'(x)=O_p(1).
\end{equation}
\end{lemma}
\begin{proof}
We bound the coefficient remainder and the estimated-standard-deviation
error separately, then treat the derivatives by the same argument.
Assumption~\ref{ass:first} gives
\[
 \E\rho_i^2+\E e_i^2\le C/T,\qquad
 0<c\le1+e_i\le C.
\]

\medskip\noindent\textit{Step 1: prove the location-remainder contribution is $O(T^{-1})$.}
The derivative of every listed function, away from its finitely many corners,
is bounded by a Gaussian envelope. On $|\rho_i|\le1$, shifting its argument
changes it by at most
$C|\rho_i|K_1(u-W_i)$, where $K_1$ is a Gaussian envelope enlarged over a
unit neighborhood. The complementary event has probability at most $C/T$
and contributes $C/T$ because the functions are bounded.

Cauchy-Schwarz separates the remainder from the leading term without
assuming their independence:
\[
 \E\{|\rho_i|K_1(u-W_i)\}
 \le (\E\rho_i^2)^{1/2}\{\E K_1(u-W_i)^2\}^{1/2}
 \le CT^{-1/2}\{e^{-c'u^2}+T^{-1/4}(1+|u|)^{-3/2}\}.
\]
The last inequality follows from \eqref{eq:momentenvelope}. The Gaussian
envelope and the polynomial envelope with exponent $3/2>1$ each average
to $O(T^{-1/2})$ by Lemma~\ref{lem:envelope}. Thus their contributions
are $O(T^{-1})$ and $O(T^{-5/4})$, respectively. Taking the supremum
over the fixed standardized window only enlarges the envelopes by fixed
constants. The bad-event contribution is already $O(T^{-1})$ and
requires no restriction to nearby coefficients.

\medskip\noindent\textit{Step 2: prove the estimated-standard-deviation contribution is $O(T^{-1})$.}
For the change in standard deviation, the mean value theorem gives
$C|e_i|K_2(v)$ with $v=u-W_i-\rho_i$ and a Gaussian envelope $K_2$;
the factor $|v|$ arising from scale differentiation is absorbed in that
envelope. On $|\rho_i|\le1$ use the same argument with $e_i$ and an
enlarged envelope of $u-W_i$. On its complement use boundedness and
$P(|\rho_i|>1)\le C/T$. This proves \eqref{eq:perturb}. The functions are
globally Lipschitz in location, and their scale derivatives are bounded
uniformly on the allowed ratios. Thus the squared difference is bounded
by $C(\rho_i^2+e_i^2)$, proving the per-unit mean-square assertion.

\medskip\noindent\textit{Step 3: prove the derivative bounds in \eqref{eq:secondderiv}.}
For the derivative of $\Phi$, the same argument has an additional
$\sqrt T$ factor. On the bad event the derivative is bounded by $C\sqrt T$,
so its expected contribution is $O(T^{-1/2})$. Finally, each first or
second derivative is bounded by $C\sqrt T$ or $CT$ times a Gaussian
envelope. On $|\rho_i|\le1$, \eqref{eq:momentenvelope} and local counts
make its average expected envelope $O(T^{-1/2})$; the bad event adds
$O(T^{-1})$. The same bounds hold with estimated scales. Markov's
inequality proves \eqref{eq:secondderiv}.
\end{proof}

\begin{lemma}[Normal-smoothing identities]\label{lem:normalidentities}
For $Z\sim N(0,1)$ and $\mu=\Phi(u/\sqrt2)$,
\[
 \E\Phi(-|u-Z|)=\mu(1-\mu),\qquad
 \E\Phi(-|u-Z|)^2=\Var\{\Phi(u-Z)\}.
\]
With normal first-step errors and known scales,
$\E\widehat F(x)=H_n(x)$ and $\E\widehat F'(x)=H_n'(x)$.
\end{lemma}
\begin{proof}
Let $Z,U$ be independent standard normals. For any real $v$,
\[
 \Phi(v)-\Phi(-|v|)
 =\Phi(v)^2-\Phi(-|v|)^2
 =\{2\Phi(v)-1\}\ind\{v\ge0\}.
\]
To evaluate its expectation at $v=u-Z$, condition on $Z$. The expression
$\{2\Phi(u-Z)-1\}\ind\{Z\le u\}$ equals the conditional probability
of $|U|\le u-Z$. Hence
\[
 \E[\{2\Phi(u-Z)-1\}\ind\{Z\le u\}]
 =P(Z+|U|\le u)=P(Z+U\le u,Z-U\le u).
\]
The pair $(Z+U,Z-U)$ is jointly normal, each component has variance two,
and its covariance is $\Var(Z)-\Var(U)=0$. Joint normality and zero
covariance imply independence, so the last probability is
$\Phi(u/\sqrt2)^2$. Also,
$\E\Phi(u-Z)=P(U+Z\le u)=\Phi(u/\sqrt2)$. Subtracting the expected
identities above proves
\[
 \E\Phi(-|u-Z|)=\mu(1-\mu),\qquad
 \E\Phi(-|u-Z|)^2=\E\Phi(u-Z)^2-\mu^2,
 \quad \mu=\Phi(u/\sqrt2).
\]
In particular,
\begin{equation}\label{eq:identity}
 \E\Phi(-|u-Z|)^2=v(u).
\end{equation}
The analogous slope identity uses the convolution
\[
 \int\phi(u-z)\phi(z)\,dz
 =\frac1{\sqrt2}\phi(u/\sqrt2).
\]
This is the density of $U+Z\sim N(0,2)$. Consequently, in the normal
reference with $\hat\theta_i=\theta_i+\sigma_iT^{-1/2}Z$ and known scale,
\[
 \E\left[\frac{\sqrt T}{\sigma_i}\phi\left(\frac{\sqrt T(x-\hat\theta_i)}{\sigma_i}\right)\right]
 =\frac{\sqrt T}{\sigma_i}\int\phi\{\sqrt T(x-\theta_i)/\sigma_i-z\}\phi(z)\,dz
 =\frac{\sqrt T}{\sqrt2\sigma_i}\phi\left(\frac{\sqrt T(x-\theta_i)}{\sqrt2\sigma_i}\right).
\]
Averaging gives $\E\widehat F'(x)=H_n'(x)$ in this reference experiment.
This is an equality of expectations, not a samplewise equality of slopes.

\end{proof}

\par\addvspace{\baselineskip}
\Needspace{13\baselineskip}
\noindent Define
\[
 \mu(u)=\Phi(u/\sqrt2),\qquad r_n=n^{-1/2}T^{-1/4},\qquad
 C_n(x)=\frac1n\sum_i\{b_i h_i(x)+\gamma_i h_i''(x)/(6T)\}.
\]
A prime on $\mu$ differentiates its
argument; a prime on $h_i$ differentiates with respect to $x$.

\begin{lemma}[The mean of the feasible smoothed CDF]\label{lem:biasmean}
Under Assumptions~\ref{ass:second} and \ref{ass:counts}, for every fixed $K$,
\[
 \sup_{|x-q_n|\le KT^{-1/2}}
 |\E\widehat F(x)-H_n(x)+C_n(x)/T|=o(T^{-1}).
\]
\end{lemma}
\begin{proof}
We expand the expectation of the linear first step, add the
quadratic coefficient and estimated-variance terms, and combine
their Gaussian expectations into $C_n(x)/T$.

\medskip\noindent\textit{Step 1: prove \eqref{eq:smooththird} and an averaged remainder $o(T^{-1})$.}
Let $u=\sqrt T(x-\theta_i)/\sigma_i$. For bounded iid
summands, expansion of the log characteristic function gives
\[
 T\log \E\exp\left(\frac{iv\psi_{it}}{\sigma_i\sqrt T}\right)
 =-\frac{v^2}{2}
  +\frac{(iv)^3\E\psi_{it}^3}{6\sigma_i^3\sqrt T}
  +O(v^4/T),\qquad |v|\le T^{1/12}.
\]
The left side is the log characteristic function of $W_i$.
Multiplication by $e^{-v^2/2}$ adds the independent smoothing
normal. After exponentiation, the remainder is bounded by
$CT^{-1}(v^4+v^6)e^{-cv^2}$ on this range; outside it,
the smoothing factor makes the integral exponentially small.
CDF inversion divides by $v$, and
\[
 \frac C T\int_{\mathbb R}(|v|^3+|v|^5)e^{-cv^2}\,dv=O(T^{-1}).
\]
Inverting the leading normal and third-cumulant terms therefore gives
\begin{equation}\label{eq:smooththird}
 \E\Phi(u-W_i)=\mu(u)
 -\frac{T^{-1/2} \E\psi_{it}^3}{6\sigma_i^3}\mu'''(u)
 +a_i(u),\qquad \sup_{i,u}|a_i(u)|=O(T^{-1}).
\end{equation}
All bounds are uniform over units; no Cram\'er or nonlattice
condition is used.

A uniform $O(T^{-1})$ bound alone is not enough after averaging over
all units. The nonuniform Berry-Esseen comparison in
Lemma~\ref{lem:comparison} also gives
$|a_i(u)|\le CT^{-1/2}(1+|u|)^{-3}$, after absorbing the Gaussian
third-derivative term. Split the units at $|u|\le M$. Local counts
bound their fraction by $C_MT^{-1/2}$, so the first part of the average
remainder is $O_M(T^{-3/2})$. On $|u|>M$, the local-count summation
argument gives a bound of $CT^{-1}(1+M)^{-2}+o(T^{-1})$.
The units outside a fixed coefficient neighborhood have $|u|\gtrsim
\sqrt T$ and contribute only $O(T^{-2})$. First letting
$T\to\infty$ and then $M\to\infty$ proves that the averaged remainder
is $o(T^{-1})$. Enlarging the standardized arguments by a fixed $K$
gives the asserted uniformity in $x$.

\medskip\noindent\textit{Step 2: obtain \eqref{eq:scoresecond} with averaged remainder $o(T^{-1})$.}
Put $V_i=T^{-1/2}\sum_t\eta_{it}$. The expansions in
Assumption~\ref{ass:second} imply
\begin{equation}\label{eq:argumentsecond}
 \frac{\sqrt T(x-\widehat\theta_i)}{\widehat\sigma_i}
 =u-W_i-T^{-1/2}\left\{\frac{Q_i}{\sigma_i}
             +\frac{(u-W_i)V_i}{2\sigma_i^2}\right\}
 +\text{a second-order remainder}.
\end{equation}
Taylor expansion of $\Phi$ therefore contributes
\begin{equation}\label{eq:scoresecond}
 -T^{-1/2} \E\left\{\frac{Q_i}{\sigma_i}\phi(u-W_i)
       +\frac{V_i}{2\sigma_i^2}(u-W_i)\phi(u-W_i)\right\}.
\end{equation}
To bound the remainder after averaging, use boundedness and
centering of the observation functions to obtain
$\E\|\bm S_i\|^{2j}\le C_j$ and
$P(|W_i|>|u|/2)\le Ce^{-c'u^2}$ for each fixed integer $j\ge0$.
Split the expectation at $|W_i|=|u|/2$. On the first event the
kernel is at most $e^{-cu^2/4}$; Cauchy-Schwarz bounds the second:
\begin{equation}\label{eq:weightedlocal}
 \begin{aligned}
 \E\bigl[(1+\|\bm S_i\|^j)e^{-c(u-W_i)^2}\bigr]
 &\le e^{-cu^2/4}\E(1+\|\bm S_i\|^j)+\{\E(1+\|\bm S_i\|^j)^2\}^{1/2}
                  P(|W_i|>|u|/2)^{1/2}\\
 &\le C_j e^{-c_j u^2}.
 \end{aligned}
\end{equation}
Enlarging the kernel over a fixed neighborhood of its argument
preserves the bound.

To apply this inequality, first omit $r_i,w_i$ and restrict to
$|T^{-1/2} Q_i/\sigma_i|\le1$ and
$|T^{-1/2} V_i|\le\sigma_i^2/2$.
Taylor expansion of the location and the positive scale shows that
the remainder in the CDF is bounded by $CT^{-1}$ times a
polynomial in $\bm S_i$ and a Gaussian envelope of $u-W_i$.
By \eqref{eq:weightedlocal} and local counts, its aggregate expectation
is bounded by
\[
 \frac{C}{nT}\sum_i e^{-c u_i(x)^2}
 =O(T^{-3/2}).
\]
Because $Q_i$ is a bounded quadratic form in sub-Gaussian sums,
the excluded events have probability $O(T^{-j/2})$ for every fixed
$j$. Their CDF contributions are therefore negligible; the same is
true of the first-order Taylor terms on those events by their moment
bounds.

Now restore $r_i,w_i$. Their relevant root-mean-square magnitudes are
\[
 \{\E(\sqrt T\,r_i/\sigma_i)^2\}^{1/2}=O(T^{-1}),\qquad
 (\E w_i^2)^{1/2}=O(T^{-1}).
\]
On $|\sqrt T\,r_i/\sigma_i|\le1$ and $|w_i|\le c_0$, for a sufficiently small
fixed $c_0>0$, a mean-value bound and Cauchy-Schwarz give expected
changes at most $CT^{-1} e^{-c u^2}$, using
\eqref{eq:weightedlocal}. The complementary events have probability
$O(T^{-2})$; boundedness of the CDF controls their contribution.
Consequently their averaged effect is at most
\[
 \frac{C}{nT}\sum_i e^{-c u_i(x)^2}
       +CT^{-2}=O(T^{-3/2}).
\]
All constants can be chosen uniformly over the fixed window in $x$.
These bounds justify \eqref{eq:scoresecond} with aggregate remainder
$o(T^{-1})$, allowing arbitrary within-unit dependence among the
location, scale, and remainder terms.

\medskip\noindent\textit{Step 3: identify the correction as $-C_n(x)/T$.}
Let $\bm S_i^G$ be a centered Gaussian vector with covariance
$\bm\Sigma_i$, and let $W_i^G,V_i^G,Q_i^G$ be the corresponding linear
and quadratic functions. The multivariate central limit comparison for
bounded iid vectors, followed by truncation of their uniformly bounded
moments, replaces each expectation in \eqref{eq:scoresecond} by its
Gaussian counterpart with error $o(1)$ uniformly in $i$ and in compact
$u$-sets. 

 For $|u|\le M$, the fraction of contributing units is $O_M(T^{-1/2})$,
so the averaged replacement error is $o_M(T^{-1/2})$.
For $|u|>M$, \eqref{eq:weightedlocal} and its Gaussian analogue bound
that average by $CT^{-1/2} e^{-cM^2}+o(T^{-1/2})$.
Taking $T\to\infty$ and then $M\to\infty$ gives $o(T^{-1/2})$.
Multiplication by the leading $T^{-1/2}$ in \eqref{eq:scoresecond}
therefore gives $o(T^{-1})$, as required.

Since $\Cov(\bm S_i^G,W_i^G)=\bm d_i/\sigma_i$,
two Gaussian integrations by parts give
\[
 \E\{\bm S_i^G\bm S_i^{G\top}g(W_i^G)\}
 =\bm\Sigma_i \E g(W_i^G)
  +\frac{\bm d_i\bm d_i^{\top}}{\sigma_i^2}\E g''(W_i^G).
\]
Set $g(w)=\phi(u-w)$ and take one half of the trace after
multiplication by $\bm D_i$. Since $\E g(W_i^G)=\mu'(u)$ and
$\E g''(W_i^G)=\mu'''(u)$, this yields
\begin{equation}\label{eq:gaussQ}
 \E\{Q_i^G\phi(u-W_i^G)\}
 =b_i\mu'(u)+\frac{\bm d_i^{\top}\bm D_i\bm d_i}{2\sigma_i^2}\mu'''(u).
\end{equation}
For the scale term,
$\Cov(V_i^G,W_i^G)=\E(\psi_{it}\eta_{it})/\sigma_i$ and
$\partial_w\{(u-w)\phi(u-w)\}=\{(u-w)^2-1\}\phi(u-w)$.
One Gaussian integration by parts therefore gives
\begin{equation}\label{eq:gaussV}
 \E\{V_i^G(u-W_i^G)\phi(u-W_i^G)\}
 =\frac{\E(\psi_{it}\eta_{it})}{\sigma_i}\mu'''(u).
\end{equation}
Combining \eqref{eq:smooththird}, \eqref{eq:scoresecond}, and
\eqref{eq:gaussQ}-\eqref{eq:gaussV}, the first correction is
\[
 -T^{-1/2}\left\{\frac{b_i}{\sigma_i}\mu'(u)
       +\frac{\E\psi_{it}^3+3\bm d_i^{\top}\bm D_i\bm d_i+3\E(\psi_{it}\eta_{it})}{6\sigma_i^3}\mu'''(u)\right\}.
\]
Since
\[
 h_i(x)=\frac{\sqrt T}{\sigma_i}\mu'(u),\qquad
 h_i''(x)=\frac{T^{3/2}}{\sigma_i^3}\mu'''(u),
\]
this is precisely $-\{b_i h_i(x)+\gamma_i h_i''(x)/(6T)\}/T$.
The aggregate remainder bounds in the preceding steps prove the lemma.
\end{proof}

\begin{lemma}[Estimating the bias coefficient]\label{lem:Bestimate}
Under Assumptions~\ref{ass:second} and \ref{ass:counts},
$\widehat B_n-B_n\to_p0$.
\end{lemma}
\begin{proof}
We show that the numerator and denominator of \eqref{eq:Bhat},
after division by $n$, differ from $C_n(\tilde q_n)$ and $H_n'(\tilde q_n)$
by $o_p(1)$. We first work at $\tilde q_n$ with the linear coefficient
approximation and true standard deviations, then introduce the
estimated quantities and replace $\tilde q_n$ by $\widehat q$.

\medskip\noindent
\textit{Step 1: estimate $H_n'(\tilde q_n)$ and $C_n(\tilde q_n)$ using the
linear coefficient approximation.}
Define the two weights
\[
 \begin{aligned}
 K_{i,0}(x)
 &=\frac{\sqrt T}{\sigma_i}\phi\{u_i(x)-W_i\},\qquad K_{i,2}(x)
 =\frac{\sqrt T}{\sigma_i^3}\phi''\{u_i(x)-W_i\},
 \end{aligned}
\]
where
\(
 u_i(x)-W_i
 =\frac{\sqrt T}{\sigma_i}
   \left(x-\theta_i-\frac1T\sum_t\psi_{it}\right).
\)
Thus these weights use the linear approximation
$\theta_i+T^{-1}\sum_t\psi_{it}$ to $\widehat\theta_i$.

Replacing $W_i$ by a standard normal $Z$, normal convolution
and differentiation give
 $\E\left[\frac{\sqrt T}{\sigma_i}\phi\{u_i(x)-Z\}\right]
 =h_i(x)$ and 
 $\E\left[\frac{\sqrt T}{\sigma_i^3}\phi''\{u_i(x)-Z\}\right]
 =\frac{h_i''(x)}{T}$.
Lemma~\ref{lem:comparison} bounds the error from replacing
$W_i$ by $Z$. Averaging this bound and applying
Lemma~\ref{lem:envelope} yields, at $x=\tilde q_n$,
\[
 \begin{aligned}
 \frac1n\sum_i \E K_{i,0}(\tilde q_n)
 &=H_n'(\tilde q_n)+O(T^{-1/2}),\\
 \frac1n\sum_i \E\left\{
 b_iK_{i,0}(\tilde q_n)+\frac{\gamma_i}{6}K_{i,2}(\tilde q_n)
 \right\}
 &=C_n(\tilde q_n)+O(T^{-1/2}).
 \end{aligned}
\]
The bounded weights $b_i$ and $\gamma_i$ do not change these
error orders.

For $j\in\{0,2\}$, \eqref{eq:weightedlocal} gives
\(
 \E K_{i,j}(\tilde q_n)^2
 \le CT e^{-c u_i(\tilde q_n)^2}.
\)
Independence across units and Lemma~\ref{lem:envelope} therefore
imply
\[
 \operatorname{Var}\left\{\frac1n\sum_i K_{i,j}(\tilde q_n)\right\}
 \le \frac{CT}{n^2}\sum_i e^{-c u_i(\tilde q_n)^2}
 =O\left(\frac{\sqrt T}{n}\right)=o(1).
\]
The same variance bound applies to the weighted sum above.
Consequently,
\[
 \begin{aligned}
 \frac1n\sum_i K_{i,0}(\tilde q_n)-H_n'(\tilde q_n)&=o_p(1),\\
 \frac1n\sum_i\left\{
 b_iK_{i,0}(\tilde q_n)+\frac{\gamma_i}{6}K_{i,2}(\tilde q_n)
 \right\}-C_n(\tilde q_n)&=o_p(1).
 \end{aligned}
\]

\medskip\noindent
\textit{Step 2: show that estimating the weights changes these
averages by $o_p(1)$.}
Let $\widehat K_{i,0}(x)$ and $\widehat K_{i,2}(x)$ be obtained
from $K_{i,0}(x)$ and $K_{i,2}(x)$ by replacing $\sigma_i$
with $\widehat\sigma_i$ and $u_i(x)-W_i$ with
$\sqrt T(x-\widehat\theta_i)/\widehat\sigma_i$.

The functions $\phi$, $\phi''$, and their derivatives are
bounded by Gaussian envelopes. The argument of
Lemma~\ref{lem:perturb}, including the prefactor $\sqrt T$,
therefore gives, for $j\in\{0,2\}$ and every fixed $K$,
\[
 \E\sup_{|x-q_n|\le K/\sqrt T}
 \frac1n\sum_i
 |\widehat K_{i,j}(x)-K_{i,j}(x)|
 =O(T^{-1/2}).
\]
Hence replacing the linear coefficients and true standard
deviations changes each average by $O_p(T^{-1/2})$. For the estimated bias coefficients, decompose
\[
 \widehat b_i\widehat K_{i,0}-b_iK_{i,0}
 =
 \widehat b_i(\widehat K_{i,0}-K_{i,0})
 +(\widehat b_i-b_i)K_{i,0}.
\]
Boundedness of $\widehat b_i$ controls the first term by the
preceding bound. At $x=\tilde q_n$, Cauchy-Schwarz controls the second:
\[
 \begin{aligned}
 \frac1n\sum_i
 \E|(\widehat b_i-b_i)K_{i,0}(\tilde q_n)|
 &\le
 \frac1n\sum_i
 \{\E(\widehat b_i-b_i)^2\}^{1/2}
 \{\E K_{i,0}(\tilde q_n)^2\}^{1/2}\\
 &\le \frac Cn\sum_i e^{-c u_i(\tilde q_n)^2}
 =O(T^{-1/2}).
 \end{aligned}
\]
The same calculation applies to
$(\widehat\gamma_i-\gamma_i)K_{i,2}(\tilde q_n)$.
No independence between these estimates and the kernel
weights is needed.

Writing the feasible numerator as
\(
 \widehat C_n(x)
 =\frac1n\sum_i\left\{
 \widehat b_i\widehat K_{i,0}(x)
 +\frac{\widehat\gamma_i}{6}\widehat K_{i,2}(x)
 \right\},
\)
and noting that
$\widehat F'(x)=n^{-1}\sum_i\widehat K_{i,0}(x)$,
we obtain
\[
 \widehat C_n(\tilde q_n)-C_n(\tilde q_n)=o_p(1),
 \qquad
 \widehat F'(\tilde q_n)-H_n'(\tilde q_n)=o_p(1).
\]

\medskip\noindent
\textit{Step 3: replace $\tilde q_n$ by $\widehat q$ and prove
$\widehat B_n-B_n=o_p(1)$.}
Proposition~\ref{prop:root} gives
\[
 \sqrt T\,|\widehat q-\tilde q_n|
 =O_p\left(n^{-1/2}T^{1/4}+T^{-1/2}\right)
 =o_p(1).
\]
Together with \eqref{eq:rootbracket}, this places the entire
segment between $\widehat q$ and $\tilde q_n$ in a fixed
$KT^{-1/2}$ neighborhood of $q_n$ with probability tending to one.

Differentiating either kernel weight introduces one more
factor $\sqrt T$. The envelope and perturbation bounds give
\[
 \sup_{|x-q_n|\le K/\sqrt T}
 \frac1n\sum_i
 \left\{|\widehat K_{i,0}'(x)|
       +|\widehat K_{i,2}'(x)|\right\}
 =O_p(\sqrt T).
\]
Since $\widehat b_i$ and $\widehat\gamma_i$ are bounded,
the mean value theorem now yields
\[
 \begin{aligned}
 &|\widehat C_n(\widehat q)-\widehat C_n(\tilde q_n)|
 +|\widehat F'(\widehat q)-\widehat F'(\tilde q_n)|\le
 |\widehat q-\tilde q_n|
 \sup_{|x-q_n|\le K/\sqrt T}
 \left\{|\widehat C_n'(x)|+|\widehat F''(x)|\right\}
 =o_p(1).
 \end{aligned}
\]
Combining this with Step 2 gives
\(
 \widehat C_n(\widehat q)-C_n(\tilde q_n)=o_p(1),
 \ 
 \widehat F'(\widehat q)-H_n'(\tilde q_n)=o_p(1).
\)

Finally, Lemma~\ref{lem:envelope} gives
$C_n(\tilde q_n)=O(1)$ and $H_n'(\tilde q_n)\ge c>0$.
Therefore $\widehat F'(\widehat q)$ is bounded away from zero
with probability tending to one, and
\[
 \begin{aligned}
 \widehat B_n-B_n
 &=
 \frac{\widehat C_n(\widehat q)-C_n(\tilde q_n)}
      {\widehat F'(\widehat q)}
 +
 \frac{C_n(\tilde q_n)\{H_n'(\tilde q_n)-\widehat F'(\widehat q)\}}
      {\widehat F'(\widehat q)H_n'(\tilde q_n)}
 =o_p(1).
 \end{aligned}
\]
\end{proof}

\par\addvspace{\baselineskip}
\Needspace{16\baselineskip}
\noindent Let $\bar\ell,\bar u$ denote the observable
bounds in \eqref{eq:observablebounds}, and define their normal-model means
\[
 \begin{aligned}
 \ell_n(x)&=\frac1n\sum_i\mu_i(x)\{2\mu_i(x)-1\},\\
 u_n(x)&=\frac1n\sum_i\mu_i(x)\{3-2\mu_i(x)\},
 \qquad \mu_i(x)=\Phi\left(\frac{\sqrt T(x-\theta_i)}{\sqrt2\sigma_i}\right).
 \end{aligned}
\]

\begin{lemma}[Uniform bounds for conservative inference]\label{lem:conservativebounds}
Under Assumptions~\ref{ass:first}-\ref{ass:counts}, for every fixed $K$,
\[
 \begin{aligned}
 \sup_{|x-q_n|\le KT^{-1/2}}
   \bigl(|\E\bar\ell(x)-\ell_n(x)|+|\E\bar u(x)-u_n(x)|\bigr)&\le C_KT^{-1},\\
 \sup_{|x-q_n|\le KT^{-1/2}}
   \bigl[\Var\bar\ell(x)+\Var\bar u(x)+\Var\widehat F(x)\bigr]&\le C_K/(n\sqrt T),\\
 \E\sup_{|x-q_n|\le KT^{-1/2}}\widehat F'(x)&\le C_K.
 \end{aligned}
\]
Furthermore,
\(
 \sup_{|x-q_n|\le KT^{-1/2}}
 \bigl(|\bar\ell(x)-\E\bar\ell(x)|+
       |\bar u(x)-\E\bar u(x)|\bigr)=o_p(T^{-1/2}),
\)
and there is $c_K>0$ such that
\[
 \inf_{q_n\le x\le q_n+KT^{-1/2}}u_n(x)\ge\tau+c_KT^{-1/2},
 \qquad
 \sup_{q_n-KT^{-1/2}\le x\le q_n}\ell_n(x)\le\tau-c_KT^{-1/2}.
\]
All statements are uniform over the class $\mathcal{F}$ defined in the main paper.
\end{lemma}
\begin{proof}
We prove the mean inequalities first. Variance bounds at finitely
many points and a derivative bound then give uniform control
between those points.

\medskip\noindent\textit{Step 1: prove \eqref{eq:boundmeanapprox} and the margins \eqref{eq:rightmargin}-\eqref{eq:leftmargin}.}
The formulas for $\ell_n$ and $u_n$ follow from the normal-smoothing
identities in Lemma~\ref{lem:normalidentities}. Lemmas~\ref{lem:comparison} and
\ref{lem:perturb} give, for every fixed $K$,
\begin{equation}\label{eq:boundmeanapprox}
 \sup_{|x-q_n|\le KT^{-1/2}}
 \bigl(|\E\bar\ell(x)-\ell_n(x)|+|\E\bar u(x)-u_n(x)|\bigr)
 \le C_KT^{-1}.
\end{equation}

For $0\le\mu\le1/2$, $\mu(2\mu-1)\le0$; for
$1/2\le\mu\le1$, $\mu(2\mu-1)\le1$.
Likewise, $\mu(3-2\mu)\ge1$ for $\mu\ge1/2$ and is
nonnegative for all $\mu\in[0,1]$. Hence
\[
 \ell_n(x)\le F_n(x-),\qquad u_n(x)\ge F_n(x).
\]
The inequalities are valid at ties: at $\theta_i=x$, the lower
and upper contributions are zero and one, respectively.

For each strict left-side coefficient with $q_n-T^{-1/2}\le\theta_i<q_n$,
bounded positive $\sigma_i$ give
$1/2\le\mu_i(q_n)\le1-\delta$ for a common $\delta>0$.
Its lower contribution falls below its latent indicator by
\[
 1-\mu_i(q_n)\{2\mu_i(q_n)-1\}
 =(1-\mu_i(q_n))(1+2\mu_i(q_n))\ge\delta.
\]
For each strict right-side coefficient with $q_n<\theta_i\le q_n+T^{-1/2}$,
$\delta\le\mu_i(q_n)\le1/2$, and its upper contribution is
at least $2\delta$. There are at least $cn/\sqrt T$ units in each group.
Using $F_n(q_n-)<\tau\le F_n(q_n)$ gives
\begin{equation}\label{eq:referencegaps}
 \ell_n(q_n)\le\tau-c_1T^{-1/2},\qquad u_n(q_n)\ge\tau+c_1T^{-1/2}
\end{equation}
for some common $c_1>0$. Combining this with
\eqref{eq:boundmeanapprox} proves
\begin{equation}\label{eq:strictmeans}
 \E\bar\ell(q_n)\le\tau-c_1T^{-1/2}+CT^{-1},\qquad
 \E\bar u(q_n)\ge\tau+c_1T^{-1/2}-CT^{-1}.
\end{equation}

For $q_n\le x\le q_n+KT^{-1/2}$, coefficients at or below $q_n$
have upper contributions at least one. The original right-side
group in $(q_n,q_n+T^{-1/2}]$ has upper contributions bounded below
by a positive constant depending only on $K$ and the scale bounds.
All other upper contributions are nonnegative. Hence
\begin{equation}\label{eq:rightmargin}
 \inf_{q_n\le x\le q_n+KT^{-1/2}}u_n(x)\ge\tau+c_KT^{-1/2}.
\end{equation}
For $q_n-KT^{-1/2}\le x\le q_n$, coefficients at or above $q_n$
have lower contributions at most zero. The original left-side
group in $[q_n-T^{-1/2},q_n)$ has lower contributions uniformly below
one, and every other lower contribution is at most its indicator
of lying strictly below $q_n$. Thus
\begin{equation}\label{eq:leftmargin}
 \sup_{q_n-KT^{-1/2}\le x\le q_n}\ell_n(x)\le\tau-c_KT^{-1/2}.
\end{equation}
The same groups of coefficients establish these bounds throughout
the respective neighborhoods. No lower density bound at every
intermediate point is needed. By \eqref{eq:boundmeanapprox},
the actual expectation bounds differ by at most $C_KT^{-1}$.

\medskip\noindent\textit{Step 2: prove \eqref{eq:boundvar} and $\E\sup_{|x-q_n|\le KT^{-1/2}}\widehat F^{\prime}(x)\le C_K$.}
We also need variance bounds throughout each fixed neighborhood.
As functions of $a\in[0,1]$, $a\pm2\min(a,1-a)$ are
$3$-Lipschitz. The independent-copy formula for variance therefore
bounds their normal-model variances by nine times the variance
of $\Phi\{\sqrt T(x-\widehat\theta_i)/\sigma_i\}$. Lemma~\ref{lem:envelope}
bounds the sum of these variances by $Cn/\sqrt T$.
Normal approximation changes the aggregate first and second
moments by $O(n/T)$, using Lemma~\ref{lem:comparison}.
The difference caused by the nonlinear remainder and estimated
scale has second moment $O(T^{-1})$ per unit by
Lemma~\ref{lem:perturb}. Applying
$\Var(X+Y)\le2\Var(X)+2\E Y^2$ and cross-unit independence yields
\begin{equation}\label{eq:boundvar}
 \sup_{|x-q_n|\le KT^{-1/2}}
 \bigl[\Var\{\bar\ell(x)\}+\Var\{\bar u(x)\}
       +\Var\{\widehat F(x)\}\bigr]\le C_K/(n\sqrt T).
\end{equation}
The proof of Lemma~\ref{lem:perturb} also gives
$\E\sup_{|x-q_n|\le KT^{-1/2}}\widehat F'(x)\le C_K$.
In particular, the expected observable bounds are Lipschitz on this
window, with constant at most $3C_K$.

\medskip\noindent\textit{Step 3: prove the uniform $o_p(T^{-1/2})$ bound in \eqref{eq:uniformscore}.}
Fix $K$ and partition $[q_n-KT^{-1/2},q_n+KT^{-1/2}]$ into
$M=\lfloor n^{1/2}T^{-1/4}\rfloor$ equal cells. For every fixed
$\varepsilon>0$, the probability of an error exceeding
$\varepsilon/(2\sqrt T)$ at any grid point is bounded by
\[
 (M+1)\frac{C_K/(n\sqrt T)}{\varepsilon^2/(4T)}
 =C_{K,\varepsilon}(M+1)\sqrt T/n,
\]
by Chebyshev's inequality and \eqref{eq:boundvar}.

Between grid points, the absolute one-sided derivatives of both
observable bounds are at most $3\widehat F'(x)$. The proof of
Lemma~\ref{lem:perturb} gives
$\E\sup_{|x-q_n|\le KT^{-1/2}}\widehat F'(x)\le C_K$.
The expected functions have the corresponding Lipschitz bound
$C_K$. The mesh is $2K/(M\sqrt T)$, so Markov's inequality bounds the
probability of an interpolation error exceeding
$\varepsilon/(2\sqrt T)$ by $C_{K,\varepsilon}/M$.
Combining the grid and interpolation bounds gives a probability
of exceeding $\varepsilon/\sqrt T$ bounded by
\[
 C_{K,\varepsilon}\left\{\frac{(M+1)\sqrt T}{n}+\frac1M\right\}\to0,
\]
with the same bound for every $F\in\mathcal{F}$. Consequently,
\begin{equation}\label{eq:uniformscore}
 \sup_{|x-q_n|\le KT^{-1/2}}
 \bigl(|\bar\ell(x)-\E\bar\ell(x)|+
       |\bar u(x)-\E\bar u(x)|\bigr)=o_p(T^{-1/2})
\end{equation}
uniformly over $\mathcal{F}$.

\end{proof}

\begin{lemma}[Rejected values near the target]\label{lem:nearbyrejection}
Under Assumptions~\ref{ass:first}-\ref{ass:counts}, there are fixed
$K,c_1>0$ and deterministic points $x_-\in[q_n-KT^{-1/2},q_n]$ and
$x_+\in[q_n,q_n+KT^{-1/2}]$ such that, for all sufficiently large samples,
\[
 \E\bar u(x_-)\le\tau-c_1T^{-1/2},\qquad
 \E\bar\ell(x_+)\ge\tau+c_1T^{-1/2}.
\]
The constants are common to the class $\mathcal{F}$; the points may
depend on the coefficient array.
\end{lemma}
\begin{proof}
We show that the weighted integral of $\ell_n(x)-\tau$ is
positive to the right of $q_n$, with most of this integral
inside a fixed $T^{-1/2}$ neighborhood. Some point in that
neighborhood must therefore have $\ell_n(x)>\tau$.
Reflection gives the left-side point; the mean approximation
then gives the inequalities for $\E\bar\ell$ and $\E\bar u$.

\medskip\noindent\textit{Step 1: prove the positive-integral bound in \eqref{eq:positivefullintegral}.}
For this calculation define
\[
 g_\sigma(v)=\Phi\left(\frac{v}{\sqrt2\sigma}\right)
       \left\{2\Phi\left(\frac{v}{\sqrt2\sigma}\right)-1\right\},
 \qquad a_i=\sqrt T(\theta_i-q_n).
\]
Then $\ell_n(q_n+T^{-1/2}t)=n^{-1}\sum_i g_{\sigma_i}(t-a_i)$.
Differentiating $g_\sigma$ gives
\[
 g_\sigma'(v)=\frac1{\sqrt2\sigma}
 \phi\left(\frac v{\sqrt2\sigma}\right)
 \left\{4\Phi\left(\frac v{\sqrt2\sigma}\right)-1\right\}.
\]
For $\lambda>0$, integration by parts has zero boundary terms
by Gaussian tail decay. With $X\sim N(0,2\sigma^2)$ it gives
\[
 \int_{\mathbb R}e^{-\lambda v}g_\sigma(v)\,dv
 =\frac1\lambda \E\left[e^{-\lambda X}
          \left\{4\Phi\left(\frac X{\sqrt2\sigma}\right)-1\right\}\right].
\]
Completing the square shifts the normal mean to $-2\lambda\sigma^2$
and introduces the factor $e^{\lambda^2\sigma^2}$; under that
shifted normal law, the CDF has mean $\Phi(-\lambda\sigma)$.
Choose $\lambda$ sufficiently small that
$4\Phi(-\lambda\sigma)-1\ge1/2$ for all allowed $\sigma$. Then
\begin{equation}\label{eq:laplacepositive}
 \int_{\mathbb R}e^{-\lambda v}g_\sigma(v)\,dv
 =\frac{e^{\lambda^2\sigma^2}}{\lambda}
       \{4\Phi(-\lambda\sigma)-1\}
 \ge\frac1{2\lambda}.
\end{equation}

If $a_i>0$, $g_{\sigma_i}(v)\le0$ for $v<0$, so removing
the part $v<-a_i$ from this integral increases it. Consequently,
\begin{equation}\label{eq:positiveunitintegral}
 \int_0^\infty e^{-\lambda t}g_{\sigma_i}(t-a_i)\,dt
 \ge\frac{e^{-\lambda a_i}}{2\lambda}>0.
\end{equation}
For $a_i\le0$, Gaussian tails instead give
\[
 \int_0^\infty e^{-\lambda t}
       \{1-g_{\sigma_i}(t-a_i)\}\,dt
 \le C e^{-c a_i^2}.
\]
Lemma~\ref{lem:envelope} bounds the average of the latter
quantities by $CT^{-1/2}$. Using $F_n(q_n)\ge\tau$ and retaining
the at least $cn/\sqrt T$ coefficients with $0<a_i\le1$ yields
\begin{align}
 \int_0^\infty e^{-\lambda t}
       \{\ell_n(q_n+T^{-1/2}t)-\tau\}\,dt
 &\ge \frac{cT^{-1/2}e^{-\lambda}}{2\lambda}-CT^{-1/2}
 \ge 2c_3T^{-1/2},\label{eq:positivefullintegral}
\end{align}
for a sufficiently small fixed $\lambda$ and a fixed $c_3>0$.
Here the constant in the negative term is independent of
$\lambda$, since $e^{-\lambda t}\le1$ for $t\ge0$.
Contributions from other coefficients above $q_n$ are nonnegative
after integration by \eqref{eq:positiveunitintegral}.

\medskip\noindent\textit{Step 2: find $x_+\in[q_n,q_n+KT^{-1/2}]$ satisfying \eqref{eq:rightrejectmean}.}
We next truncate the integral at a fixed $K$, uniformly over arrays.
For all real $v$,
\[
 |g_\sigma(v)-\ind\{v\ge0\}|
 \le3\overline\Phi\left(\frac{|v|}{\sqrt2\sigma}\right),
 \qquad\overline\Phi=1-\Phi.
\]
For $x=q_n+T^{-1/2}t$ and $0\le t\le\epsilon\sqrt T/8$,
the local-count upper bound gives
\[
 0\le F_n(x)-\tau
 \le\frac1n\#\{i:q_n\le\theta_i\le x\}
 \le CT^{-1/2}(1+t).
\]
To bound $|\ell_n(x)-F_n(x)|$, split the coefficients at
$|\theta_i-q_n|=2T^{-1/2}(1+t)$ and then into annuli whose radii
double. The central count is at most $CnT^{-1/2}(1+t)$.
Each outer count grows at most linearly in its radius, while
the displayed Gaussian bound decreases exponentially in its
squared standardized radius. Summing gives
$|\ell_n(x)-F_n(x)|\le CT^{-1/2}(1+t)+Ce^{-cT}$.
The triangle inequality now gives, for sufficiently large $T$,
\begin{equation}\label{eq:lengthtailenvelope}
 |\ell_n(q_n+T^{-1/2}t)-\tau|\le CT^{-1/2}(1+t),
 \qquad 0\le t\le\epsilon\sqrt T/8.
\end{equation}

Outside this range, $|\ell_n-\tau|\le2$ suffices. Hence
\[
 \int_K^\infty e^{-\lambda t}|\ell_n(q_n+T^{-1/2}t)-\tau|\,dt
 \le CT^{-1/2}\int_K^\infty e^{-\lambda t}(1+t)\,dt
       +Ce^{-\lambda\epsilon\sqrt T/8}.
\]
Choose a fixed $K$ making the first term at most $c_3/(2\sqrt T)$.
The second is $o(T^{-1/2})$. Equation~\eqref{eq:positivefullintegral}
then implies, for all sufficiently large $T$,
\begin{equation}\label{eq:positivetruncatedintegral}
 \int_0^K e^{-\lambda t}\{\ell_n(q_n+T^{-1/2}t)-\tau\}\,dt\ge c_3T^{-1/2}.
\end{equation}
Since $\int_0^K e^{-\lambda t}\,dt\le1/\lambda$,
\[
 \sup_{0\le t\le K}\{\ell_n(q_n+T^{-1/2}t)-\tau\}
 \ge\lambda\int_0^K e^{-\lambda t}
                  \{\ell_n(q_n+T^{-1/2}t)-\tau\}\,dt
 \ge\lambda c_3T^{-1/2}.
\]
Continuity therefore gives a deterministic point
$x_+\in[q_n,q_n+KT^{-1/2}]$ at which
\begin{equation}\label{eq:rightrejectmean}
 \ell_n(x_+)\ge\tau+\lambda c_3T^{-1/2}.
\end{equation}
The point may depend on the array; the constants do not.

\medskip\noindent\textit{Step 3: obtain the left-side point and replace the normal-model means.}
Reflect the coefficients around zero, keeping their scales,
and denote the resulting functions by the superscript
$\mathrm{ref}$. Then
\[
 u_n(x)=1-\ell_n^{\mathrm{ref}}(-x),\qquad
 F_n^{\mathrm{ref}}(-q_n)=1-F_n(q_n-)>1-\tau.
\]
The reflected array has the same local counts. Applying Steps 1-2
at $-q_n$ with level $1-\tau$ therefore gives the left-side point.
This argument uses the displayed CDF inequality, so it also covers
integer $n\tau$ and ties.
After enlarging $K$ and reducing a common constant $c_4>0$ if needed,
there is $x_-\in[q_n-KT^{-1/2},q_n]$ with
\begin{equation}\label{eq:leftrejectmean}
 u_n(x_-)\le\tau-c_4T^{-1/2},
 \qquad \ell_n(x_+)\ge\tau+c_4T^{-1/2}.
\end{equation}

Finally, Lemma~\ref{lem:conservativebounds} bounds the difference
between each actual expectation and its normal-model mean by $C_KT^{-1}$.
For sufficiently large samples this is at most $c_4/(2\sqrt T)$.
The displayed bounds therefore hold with $c_1=c_4/2$.
\end{proof}
\clearpage
\section{Additional Simulation Results}
\label{subsec:quarter_simulation}
We repeat the mean-estimation experiment at $\tau=0.25$, using the
same skewed errors, estimated standard errors, and feasible bias
correction as in the main paper. Each design uses 5,000 repetitions.
The target is $q_n=\theta_{(\lceil n/4\rceil)}$. For the reference
procedures, the stochastic standard error is
$\sqrt{\tau(1-\tau)/n}/\widehat F'(\widehat q)$, and the projection
uses the $\lceil n/4\rceil$th ordered lower and upper endpoints.

Panel A of Table~\ref{tab:quarter_size} uses the main-paper normal grid.
The lower quartile has fewer nearby coefficients than the median,
and the grid is not symmetric around this quantile. The normal
procedures consequently show somewhat greater size distortion:
at $(n,T)=(1001,1500)$, rejection frequencies are 0.063 and 0.062,
whereas at $(101,1500)$ they are 0.118 and 0.116. First-step bias
correction has little effect on this grid; it does not remove the
distinct displacement caused by smoothing the coefficient array.
The conservative interval also reflects the limited number of local
comparisons at $(101,1500)$, rejecting with frequency 0.080.
Its other grid rejection frequencies range from 0 to 0.014.

To examine local balance separately, let $n=4k+1$. Place one
coefficient at zero, $k$ pairs at $(-a_j,a_j)$, and the remaining
$2k$ coefficients at $3+j/(2k+1)$, $j=1,\ldots,2k$.
With $a_j=|Z_j|$ and fixed standard normal draws $Z_j$, this gives
an irregular array whose lower quartile is zero. Common true
standard deviations imply
$H_n(0)-1/4=1/(4n)+O(e^{-cT})$ for some $c>0$,
which is negligible at the normal-inference scale in the growth
range considered. The unbalanced version doubles every positive
coefficient while retaining the same target. Panels B and C use
the same fixed draws before this transformation.

\begin{table}[!htbp]\centering\small
\setlength{\tabcolsep}{5pt}
\caption{Lower-quartile rejection frequencies}\label{tab:quarter_size}
\begin{tabular}{rrccccc}\toprule
\multicolumn{2}{c}{Sample size} & \multicolumn{5}{c}{Rejection frequency}\\
\cmidrule(lr){1-2}\cmidrule(lr){3-7}
$n$ & $T$ & $I_{\mathrm{U}}$ & $I_{\mathrm{BC}}$ & $I_{\mathrm{C}}$ & $I_{\mathrm{S}}$ & $I_{\mathrm{P}}$\\\midrule
\multicolumn{7}{l}{\textit{A. Normal grid}}\\
101 & 250 & 0.089 & 0.089 & 0.003 & 0.000 & 0.000\\
401 & 250 & 0.080 & 0.080 & 0.000 & 0.000 & 0.000\\
1,001 & 250 & 0.086 & 0.087 & 0.000 & 0.000 & 0.000\\
101 & 500 & 0.104 & 0.106 & 0.014 & 0.000 & 0.000\\
401 & 500 & 0.075 & 0.075 & 0.000 & 0.000 & 0.000\\
1,001 & 500 & 0.072 & 0.071 & 0.000 & 0.000 & 0.000\\
101 & 1,500 & 0.118 & 0.116 & 0.080 & 0.000 & 0.000\\
401 & 1,500 & 0.071 & 0.072 & 0.000 & 0.000 & 0.000\\
1,001 & 1,500 & 0.063 & 0.062 & 0.000 & 0.000 & 0.000\\
\addlinespace[3pt]
\multicolumn{7}{l}{\textit{B. Locally paired array}}\\
401 & 250 & 0.064 & 0.067 & 0.000 & 0.000 & 0.000\\
1,001 & 500 & 0.064 & 0.065 & 0.000 & 0.000 & 0.000\\
\addlinespace[3pt]
\multicolumn{7}{l}{\textit{C. Unbalanced array}}\\
401 & 250 & 0.146 & 0.143 & 0.000 & 0.000 & 0.000\\
1,001 & 500 & 0.345 & 0.340 & 0.000 & 0.000 & 0.000\\
\addlinespace[3pt]
\bottomrule\end{tabular}
\par\smallskip\begin{minipage}{.97\textwidth}\footnotesize
Notes: 5,000 repetitions per design; $\tau=0.25$. The target is the latent order statistic $\theta_{(\lceil n/4\rceil)}$. $I_{\mathrm U}$, $I_{\mathrm{BC}}$, $I_{\mathrm S}$, and $I_{\mathrm P}$ use 5\% critical values; $I_{\mathrm C}$ has no nominal-level parameter. Entries are rounded to three decimal places.
\end{minipage}\end{table}

The paired arrays give normal rejection frequencies between 0.064
and 0.067. After local balance is removed, the frequencies rise
to 0.345 and 0.340 for $I_{\mathrm U}$ and $I_{\mathrm{BC}}$ at
$(n,T)=(1001,500)$. The conservative procedure remains conservative
in these designs. Figure~\ref{fig:quarter_power} shifts every coefficient
by $\delta$ and tests the zero-quartile null, using the same first-step
errors for every shift. Under imbalance, the normal procedures'
high rejection frequencies near zero reflect incorrect centering.
The conservative interval requires larger shifts, while preserving
coverage of the unshifted target.

\begin{figure}[!htbp]\centering
\includegraphics[width=\textwidth]{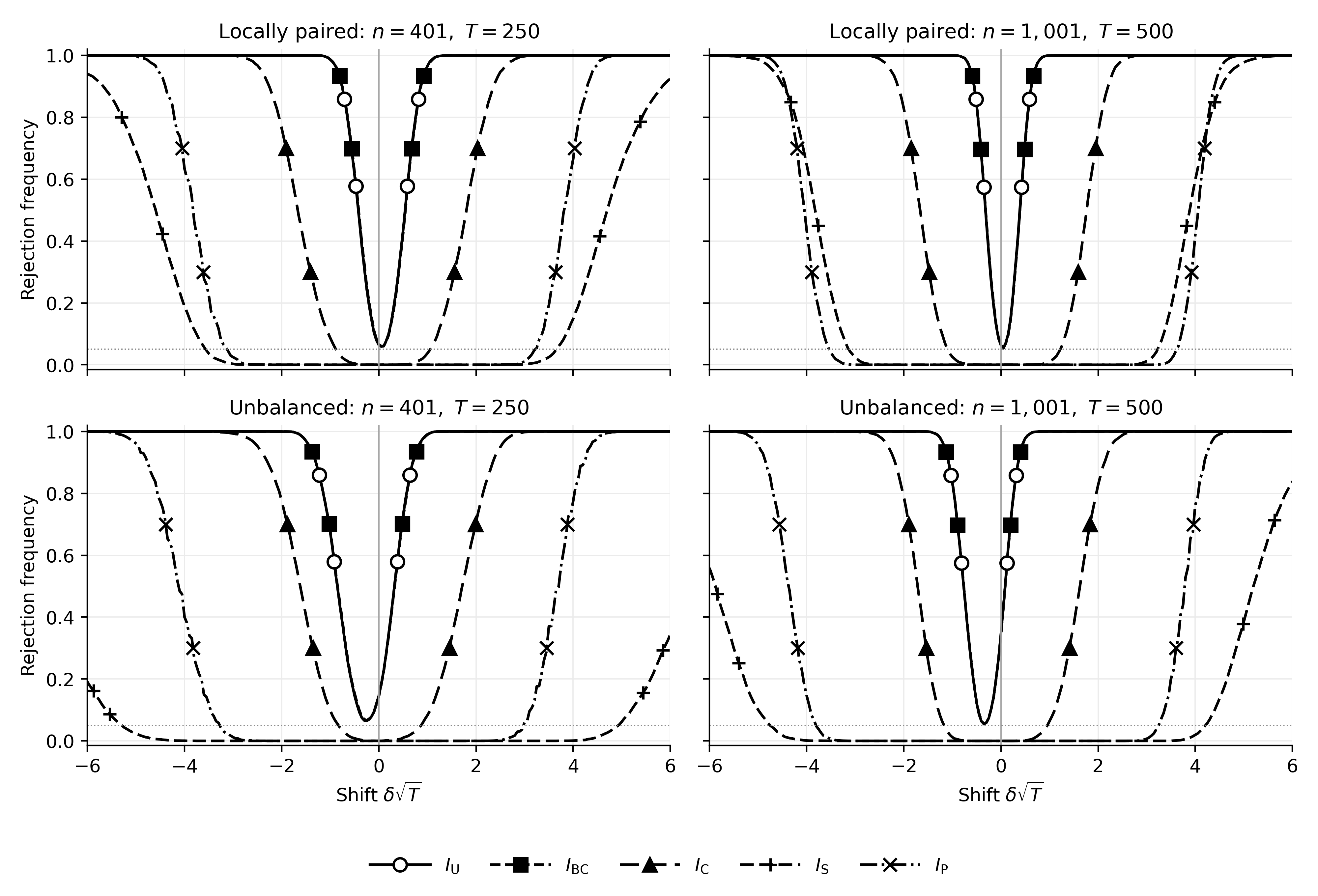}
\caption{Rejection frequencies for the zero lower-quartile null.
Rows show locally paired and unbalanced arrays; columns show
$(n,T)=(401,250)$ and $(1001,500)$. The horizontal axis is
$\delta\sqrt T$. Each curve uses 5,000 repetitions, with no size
adjustment. The dotted line marks 0.05.}
\label{fig:quarter_power}
\end{figure}

Finally, we examine the bias correction when $n/T^{3/2}$ is larger.
In the paired construction above, replace the positive distances
by $m=\lfloor n/\sqrt T+1/2\rfloor$ copies of $1/(2\sqrt T)$
and the $k-m$ values $j/(k-m+1)$, $j=1,\ldots,k-m$.
The distant positive coefficients are unchanged. This concentrates
coefficients near the lower quartile while preserving local balance.
We use $T\in\{50,100,200\}$ and
$n=4\lfloor cT^{3/2}/4\rfloor+1$ for $c\in\{1/4,1,4\}$.
Table~\ref{tab:quarter_bias} shows that increasing this ratio makes
the first-step correction important away from the median as well.
At $T=100$, raising $n/T^{3/2}$ from approximately 0.25 to 4
raises uncorrected rejection from 0.065 to 0.199; the corresponding
bias-corrected frequencies are 0.052 and 0.053. Thus the correction
addresses first-step bias without relying on the target being a median.

\begin{table}[!htbp]\centering\small
\setlength{\tabcolsep}{5pt}
\caption{Lower-quartile bias correction with local concentrations}\label{tab:quarter_bias}
\begin{tabular}{rrrccccc}\toprule
\multicolumn{2}{c}{Sample size} & \multicolumn{1}{c}{Rate} & \multicolumn{5}{c}{Rejection frequency}\\
\cmidrule(lr){1-2}\cmidrule(lr){3-3}\cmidrule(lr){4-8}
$n$ & $T$ & $n/T^{3/2}$ & $I_{\mathrm{U}}$ & $I_{\mathrm{BC}}$ & $I_{\mathrm{C}}$ & $I_{\mathrm{S}}$ & $I_{\mathrm{P}}$\\\midrule
89 & 50 & 0.252 & 0.092 & 0.068 & 0.000 & 0.000 & 0.000\\
353 & 50 & 0.998 & 0.102 & 0.057 & 0.000 & 0.000 & 0.000\\
1,413 & 50 & 3.997 & 0.232 & 0.052 & 0.000 & 0.000 & 0.000\\
249 & 100 & 0.249 & 0.065 & 0.052 & 0.000 & 0.000 & 0.000\\
1,001 & 100 & 1.001 & 0.104 & 0.054 & 0.000 & 0.000 & 0.000\\
4,001 & 100 & 4.001 & 0.199 & 0.053 & 0.000 & 0.000 & 0.000\\
705 & 200 & 0.249 & 0.058 & 0.046 & 0.000 & 0.000 & 0.000\\
2,829 & 200 & 1.000 & 0.092 & 0.050 & 0.000 & 0.000 & 0.000\\
11,313 & 200 & 4.000 & 0.196 & 0.052 & 0.000 & 0.000 & 0.000\\
\addlinespace[3pt]
\bottomrule\end{tabular}
\par\smallskip\begin{minipage}{.97\textwidth}\footnotesize
Notes: 5,000 repetitions per design; $\tau=0.25$. The target is the latent order statistic $\theta_{(\lceil n/4\rceil)}$. $I_{\mathrm U}$, $I_{\mathrm{BC}}$, $I_{\mathrm S}$, and $I_{\mathrm P}$ use 5\% critical values; $I_{\mathrm C}$ has no nominal-level parameter. Entries are rounded to three decimal places.
\end{minipage}\end{table}
\clearpage

\section{Additional Empirical Results}\label{sec:proofempirical}

\subsection{Data construction and balanced sampling}

We use the 2007 and 2009 quarterly origination and
performance files in the archived Freddie Mac release described
in Section~\ref{sec:empirical}. The archive contains monthly performance
through September 2022. Variable definitions follow
\citet{freddiemac2026}. Loan identifiers join the origination and
performance records; reported credit scores are used without grouping
or interpolation.

A loan is eligible if its credit score lies in $\{620,\ldots,800\}$ and
its first 24 monthly records, ages 0-23 measured from the original
first-payment date, are observed continuously
until the horizon or a documented termination. We retain voluntary
payoffs (termination code 01), third-party sales (02), short sales
or charge-offs (03), and REO dispositions (09). Repurchases, defect
flags, unexplained gaps, and other early exits are excluded. Reported loan age can reset following a modification; we instead use
calendar months relative to the original first-payment date throughout. The
same eligible population is used at both horizons. These rules
restrict the target to loans with usable performance histories;
no missing delinquency record is imputed as zero.

For a horizon of $h\in\{12,24\}$ months, $Y_{it,a}=1$ if a loan
in cohort $a$ has a reported delinquency status of at least three,
an REO status, or an adverse termination (02, 03, or 09) at a
loan age in $\{0,\ldots,h-1\}$. Otherwise $Y_{it,a}=0$.
A voluntary payoff ends exposure to the recorded event and remains
in the denominator. Thus the outcome is recorded serious delinquency
before payoff or the horizon, not a counterfactual default probability
had prepaid loans remained outstanding.

The eligible counts in the credit-score range are 883,581 and
1,498,933 in the two cohorts. The smallest number of loans at any credit score in either
cohort is $T=449$, giving a balanced design with $n=181$.
Within each credit score and cohort, we retain the first $T$ loans
in the fixed random ordering described in Section~\ref{sec:empirical}. This selects observations without using their
outcomes. This gives every credit score the same sample size. Each credit-score value receives
equal weight in the target regardless of its number of eligible loans. Loans are not matched across cohorts.

The credit-score range and horizons were fixed before interval calculation.
An initial first-quarter-only extraction gave $T=116$, with sparse
events at many credit-score values. Both cohorts were therefore expanded to
all four quarters while retaining the credit-score range and outcomes.

\subsection{First-step estimates and the five intervals}
\label{subsec:sa five interval}
Write $\bar Y_{ia}=T^{-1}\sum_{t=1}^T Y_{it,a}$, where
$a=0$ denotes 2007 and $a=1$ denotes 2009. Then
\[
 \widehat\theta_i=\bar Y_{i0}-\bar Y_{i1},\qquad
 \widehat\sigma_i^2=\frac{T}{T-1}
 \{\bar Y_{i0}(1-\bar Y_{i0})+\bar Y_{i1}(1-\bar Y_{i1})\}.
\]
The standard error of $\widehat\theta_i$ is
$\widehat\sigma_iT^{-1/2}$. For numerical evaluation we restrict
$\widehat\sigma_i$ to $[0.01,2]$. The lower bound is used for one credit-score value in the 24-month calculation
and 21 credit-score values in the 12-month calculation; these cells contain no
recorded events in either selected cohort.
The statistical approximation treats loans as independent draws
within and across credit-score/cohort cells, conditional on the cohort
environments. Common lender or geographic shocks can violate this
approximation; fixing credit scores does not remove such dependence.

For centered outcomes $\varepsilon_{it,a}=Y_{it,a}-\E Y_{it,a}$,
\[
 \psi_{it}=\varepsilon_{it,0}-\varepsilon_{it,1},\qquad
 \eta_{it}=\varepsilon_{it,0}^2+\varepsilon_{it,1}^2-\sigma_i^2.
\]
Independence between cohorts gives
$\E(\psi_{it}\eta_{it})=\E\psi_{it}^3
=\E\varepsilon_{it,0}^3-\E\varepsilon_{it,1}^3$.
Since the difference of means is unbiased, the bias coefficients are
\[
 \widehat b_i=0,\qquad
 \widehat\gamma_i=4\{\bar Y_{i0}(1-\bar Y_{i0})(1-2\bar Y_{i0})
              -\bar Y_{i1}(1-\bar Y_{i1})(1-2\bar Y_{i1})\}.
\]
Substituting these into \eqref{eq:Bhat} gives the corrected center.
The intervals $I_{\mathrm U}$ and $I_{\mathrm{BC}}$ use the same
standard error \eqref{eq:se}; $I_{\mathrm S}$ uses
$\sqrt{\tau(1-\tau)/n}/\widehat F'(\widehat q)$.
The conservative endpoints are the first outward crossings of
$\widehat F(x)+2\bar p(x)=\tau$ to the left of $\widehat q$ and
$\widehat F(x)-2\bar p(x)=\tau$ to its right. Doubling the search
resolution leaves the reported endpoints unchanged.

For $I_{\mathrm P}$, form a Clopper-Pearson interval
$[\ell_{ia},u_{ia}]$ for each of the $2n$ cohort probabilities,
using tail probability $\alpha/(4n)$ with $\alpha=0.05$
\citep{clopper1934}. Set
\[
 L_i=\ell_{i0}-u_{i1},\qquad U_i=u_{i0}-\ell_{i1},\qquad
 I_{\mathrm P}=[L_{(\lceil n\tau\rceil)},U_{(\lceil n\tau\rceil)}].
\]
The union bound gives simultaneous coefficient coverage of at least
$1-\alpha$, and ordering preserves these bounds for $q_n$.
This finite-sample guarantee requires the binomial model for each
cell; it does not require independence between the simultaneous
coverage events.

\subsection{Quantile estimates, intervals, and credit-score profiles}

\begin{table}[!htbp]\centering\small
\setlength{\tabcolsep}{3.2pt}
\caption{Mortgage-performance quantiles: $n=181$ and $T=449$}
\label{tab:mortgagecomplete}
\resizebox{\textwidth}{!}{\begin{tabular}{rrccccc}\toprule
$\tau$ & Estimate & $I_{\mathrm U}$ & $I_{\mathrm{BC}}$ & $I_{\mathrm C}$ & $I_{\mathrm S}$ & $I_{\mathrm P}$\\\midrule
\multicolumn{7}{l}{\textit{24-month serious delinquency}}\\
0.10 & 0.30 & [0.19, 0.41] & [0.17, 0.39] & [-0.19, 1.54] & [0.04, 0.55] & [-4.93, 3.59]\\
0.25 & 1.01 & [0.89, 1.14] & [0.91, 1.16] & [0.32, 2.51] & [0.72, 1.30] & [-3.65, 5.13]\\
0.50 & 2.35 & [2.17, 2.54] & [2.18, 2.55] & [1.04, 3.96] & [1.90, 2.80] & [-2.79, 7.98]\\
0.75 & 4.26 & [4.00, 4.51] & [4.00, 4.50] & [2.22, 6.24] & [3.60, 4.91] & [-2.29, 11.88]\\
0.90 & 6.38 & [5.97, 6.80] & [5.97, 6.81] & [3.29, 8.08] & [5.53, 7.24] & [-1.89, 14.57]\\
\addlinespace[3pt]
\multicolumn{7}{l}{\textit{12-month serious delinquency}}\\
0.10 & -0.54 & [-0.64, -0.44] & [-0.68, -0.48] & [-1.04, 0.25] & [-0.75, -0.33] & [-4.57, 2.11]\\
0.25 & -0.09 & [-0.14, -0.04] & [-0.18, -0.08] & [-0.56, 0.45] & [-0.19, 0.01] & [-3.71, 2.67]\\
0.50 & 0.17 & [0.10, 0.24] & [0.16, 0.29] & [-0.18, 0.79] & [0.05, 0.29] & [-2.67, 3.15]\\
0.75 & 0.65 & [0.56, 0.74] & [0.59, 0.77] & [-0.02, 1.27] & [0.49, 0.81] & [-2.11, 4.39]\\
0.90 & 1.22 & [1.08, 1.36] & [1.08, 1.35] & [0.15, 1.83] & [0.96, 1.48] & [-2.11, 5.49]\\
\bottomrule\end{tabular}}
\par\smallskip\begin{minipage}{.98\textwidth}\footnotesize
Notes: Estimates and endpoints are percentage-point reductions from
2007 to 2009. The same $T=449$ loans per credit score and cohort enter both
horizons. The normal and stochastic-reference intervals use 95\%
calibration; $I_{\mathrm C}$ has no nominal-level parameter.
$I_{\mathrm P}$ projects simultaneous exact binomial intervals with
Bonferroni adjustment. Normal inference for $q_n$ requires local
balance. Coverage statements are pointwise in $\tau$.
\end{minipage}\end{table}

Figure~\ref{fig:creditscoreprofiles2} displays cohort delinquency
frequencies using all eligible loans at the 12-month horizon. Thin lines
show frequencies at individual credit-score values; the heavier lines
connect averages over groups of ten credit-score values. These plotted averages are descriptive; the estimates and intervals
use the observations at each individual credit-score value.

\begin{figure}[!htbp]\centering
\includegraphics[width=\textwidth]{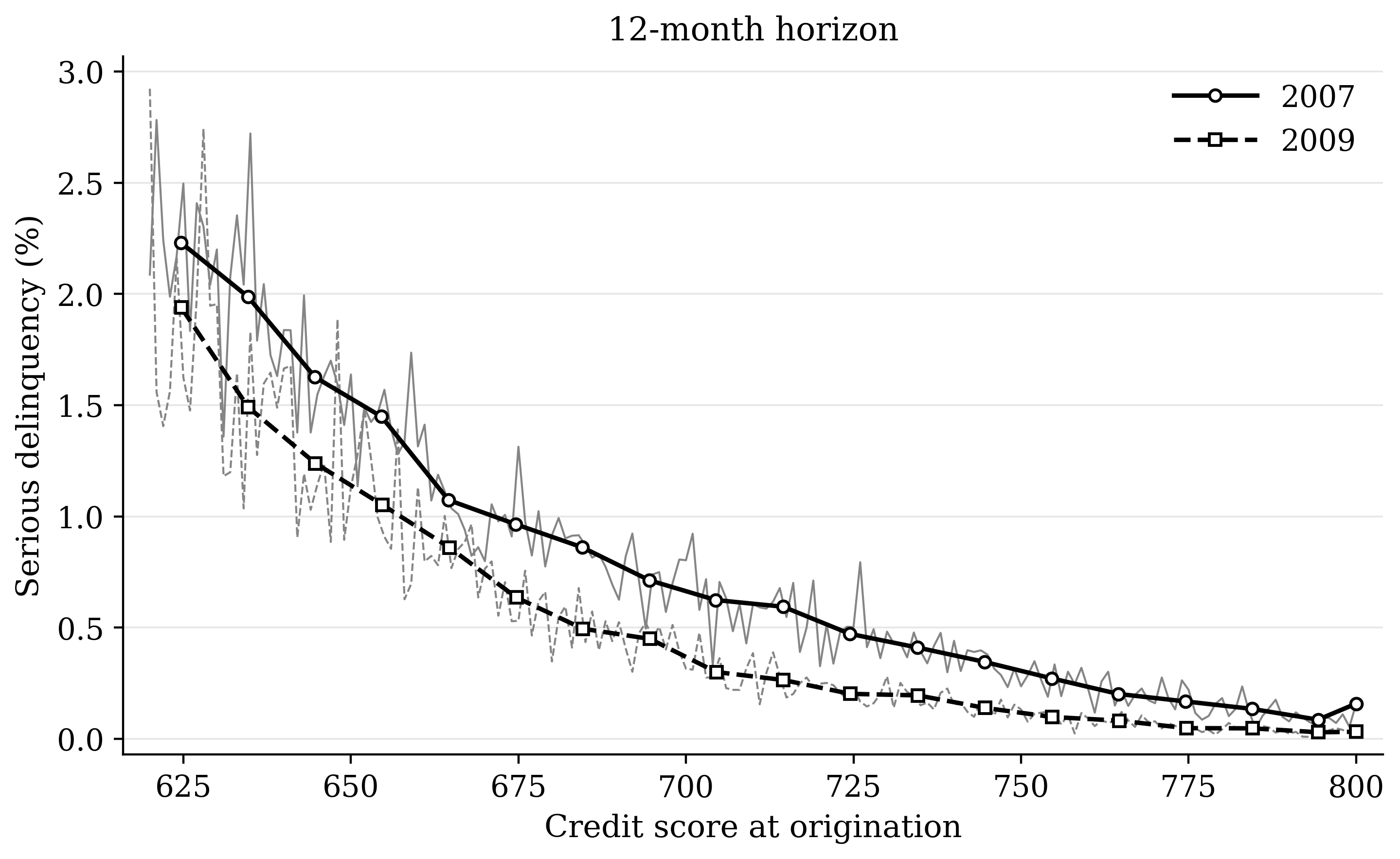}
\caption{Serious-delinquency frequencies by credit score and origination
cohort. Thin lines show frequencies at individual credit-score values; black
solid and dashed lines connect averages over groups of ten credit scores for 2007 and 2009, respectively.
All eligible loans are used in this descriptive figure.}
\label{fig:creditscoreprofiles2}
\end{figure}

\FloatBarrier

\subsection{Robustness checks}\label{subsec:mortgagerobustness}

We examine whether the results depend on restricting the sample to
$T=449$ loans per credit score and cohort, on which loans are selected,
or on the lower bound used for estimated standard deviations.
The credit-score range, outcome definitions, and equal weighting of
the 181 credit-score values remain unchanged.

\paragraph{All eligible loans.}
We first use all 883,581 eligible loans from 2007 and 1,498,933 from
2009 to calculate both quantile estimates and the five intervals.
Let $m_{ia}$ be the number of loans with credit score $x_i$ in cohort
$a$, and let $\widehat p_{ia}$ be their delinquency frequency, where
$a=0$ denotes 2007 and $a=1$ denotes 2009. The coefficient estimates
and standard errors are
\[
 \widehat\theta_i^{\mathrm{all}}=\widehat p_{i0}-\widehat p_{i1},
 \qquad
 \widehat s_i^{\mathrm{all}}=
 \left\{\frac{\widehat p_{i0}(1-\widehat p_{i0})}{m_{i0}-1}
       +\frac{\widehat p_{i1}(1-\widehat p_{i1})}{m_{i1}-1}\right\}^{1/2}.
\]
Substituting these into the smoothed CDF gives
$\widehat q^{\mathrm{all}}$. All estimated standard errors are
positive, so no lower bound is needed. The formulas for
$I_{\mathrm U}$, $I_{\mathrm C}$, and $I_{\mathrm S}$ use these
estimates and standard errors. For $I_{\mathrm P}$, the binomial
intervals use the actual number of loans in each cohort at each
credit score.

For bias correction, the sample sizes also enter the third-moment
term. Define
\[
 \widehat\kappa_i=
 \frac{\widehat p_{i0}(1-\widehat p_{i0})(1-2\widehat p_{i0})}{m_{i0}^2}
 -\frac{\widehat p_{i1}(1-\widehat p_{i1})(1-2\widehat p_{i1})}{m_{i1}^2}.
\]
For independent sample means, both the third moment of the estimation
error and its covariance with the estimated variance have this leading
sample analogue. Thus, skewness and variance estimation together give
$4\widehat\kappa_i$. With
$z_i=(\widehat q^{\mathrm{all}}-\widehat\theta_i^{\mathrm{all}})/
\widehat s_i^{\mathrm{all}}$, the corrected estimate is
\begin{equation}\label{eq:mortgageallbias}
 \widehat q_{\mathrm{BC}}^{\mathrm{all}}
 =\widehat q^{\mathrm{all}}-
 \frac{\sum_i\{4\widehat\kappa_i/[6(\widehat s_i^{\mathrm{all}})^3]\}
                 (z_i^2-1)\phi(z_i)}
      {\sum_i(\widehat s_i^{\mathrm{all}})^{-1}\phi(z_i)}.
\end{equation}
When $m_{i0}=m_{i1}=T$, this is exactly the correction used in the
balanced analysis. These calculations assess empirical sensitivity
to unequal sample sizes; the formal coverage results in the paper
retain the common-$T$ formulation.

For comparison, Table~\ref{tab:mortgageallsample} reports both smoothed
estimates and raw quantiles, $\widehat\theta_{(\lceil n\tau\rceil)}$,
for the balanced sample and all eligible loans.
Table~\ref{tab:mortgageallintervals} reports the all-loan intervals.

\begin{table}[!htbp]\centering\small
\setlength{\tabcolsep}{7pt}
\caption{Quantile estimates using the balanced sample and all eligible loans}
\label{tab:mortgageallsample}
\begin{tabular}{rrrrr}\toprule
 & \multicolumn{2}{c}{Balanced sample}
 & \multicolumn{2}{c}{All eligible loans}\\
\cmidrule(lr){2-3}\cmidrule(lr){4-5}
$\tau$ & Smoothed & Raw & Smoothed & Raw\\\midrule
\multicolumn{5}{l}{\textit{24-month serious delinquency}}\\
0.10 & 0.30 & 0.67 & 0.68 & 0.70\\
0.25 & 1.01 & 1.11 & 1.34 & 1.29\\
0.50 & 2.35 & 2.45 & 2.57 & 2.66\\
0.75 & 4.26 & 4.01 & 3.90 & 4.06\\
0.90 & 6.38 & 6.01 & 5.34 & 5.12\\
\addlinespace[3pt]
\multicolumn{5}{l}{\textit{12-month serious delinquency}}\\
0.10 & -0.54 & -0.22 & -0.02 & 0.06\\
0.25 & -0.09 & 0.00 & 0.08 & 0.10\\
0.50 & 0.17 & 0.22 & 0.20 & 0.22\\
0.75 & 0.65 & 0.67 & 0.39 & 0.39\\
0.90 & 1.22 & 0.89 & 0.67 & 0.54\\
\bottomrule\end{tabular}
\par\smallskip\begin{minipage}{.98\textwidth}\footnotesize
Notes: Estimates are percentage-point reductions in delinquency from
2007 to 2009. The balanced sample has $T=449$ loans per credit score and cohort.
The full sample uses all 883,581 eligible loans from 2007 and 1,498,933
from 2009. Each credit-score value receives equal weight in both calculations.
The raw quantile is the indicated order statistic of the estimated
credit-score-specific differences.
\end{minipage}\end{table}

\begin{table}[!htbp]\centering\small
\setlength{\tabcolsep}{3pt}
\caption{Quantile intervals using all eligible loans}
\label{tab:mortgageallintervals}
\begin{tabular*}{\textwidth}{@{\extracolsep{\fill}}rccccc@{}}\toprule
$\tau$ & $I_{\mathrm U}$ & $I_{\mathrm{BC}}$ & $I_{\mathrm C}$
& $I_{\mathrm S}$ & $I_{\mathrm P}$\\\midrule
\multicolumn{6}{l}{\textit{24-month serious delinquency}}\\
0.10 & [0.641, 0.728] & [0.642, 0.729] & [0.516, 0.849] & [0.543, 0.825] & [-0.418, 1.324]\\
0.25 & [1.278, 1.398] & [1.278, 1.398] & [1.081, 1.688] & [1.033, 1.642] & [0.149, 2.188]\\
0.50 & [2.498, 2.652] & [2.501, 2.655] & [2.156, 3.216] & [2.241, 2.908] & [0.591, 4.114]\\
0.75 & [3.764, 4.045] & [3.764, 4.045] & [3.102, 5.079] & [3.457, 4.352] & [1.092, 6.699]\\
0.90 & [5.062, 5.610] & [5.041, 5.589] & [4.018, 6.405] & [4.752, 5.920] & [1.473, 9.977]\\
\addlinespace[5pt]
\multicolumn{6}{l}{\textit{12-month serious delinquency}}\\
0.10 & [-0.044, 0.011] & $[-0.054,\,0.000]^{\dagger}$ & [-0.292, 0.208] & [-0.071, 0.038] & [-3.061, 0.378]\\
0.25 & [0.059, 0.091] & [0.060, 0.092] & [-0.029, 0.287] & [0.047, 0.104] & [-1.340, 0.606]\\
0.50 & [0.178, 0.222] & [0.177, 0.220] & [0.047, 0.435] & [0.157, 0.242] & [-0.383, 0.990]\\
0.75 & [0.357, 0.432] & [0.350, 0.425] & [0.125, 0.677] & [0.322, 0.468] & [-0.161, 1.566]\\
0.90 & [0.599, 0.738] & [0.589, 0.728] & [0.194, 1.049] & [0.531, 0.806] & [-0.102, 2.604]\\
\bottomrule\end{tabular*}
\par\smallskip\begin{minipage}{\textwidth}\footnotesize
Notes: Endpoints are in percentage points. All eligible loans are used,
with their actual sample sizes at each credit score and in each cohort.
The normal and stochastic-reference intervals use 95\% calibration;
$I_{\mathrm C}$ has no nominal-level parameter. The projection uses
simultaneous exact binomial intervals. Intervals are pointwise in $\tau$.
$\dagger$ The unrounded upper endpoint is $-0.000058$ percentage points.
\end{minipage}\end{table}
\FloatBarrier

The 24-month results continue to show a broad reduction in delinquency.
Using all loans raises the smoothed lower quartile from 1.01 to 1.34
percentage points and the median from 2.35 to 2.57, while reducing the
upper decile from 6.38 to 5.34. At the lower quartile, both normal
intervals are approximately $[1.28,1.40]$, and
$I_{\mathrm C}=[1.08,1.69]$. The conservative lower-decile interval,
$[0.52,0.85]$, also excludes zero. Thus, the all-loan calculation
supports the stronger conclusion of improvement at more than 90\%
of credit-score values. The projection intervals now exclude zero
from the lower quartile upward.

The 12-month reductions remain much smaller, and some conclusions
change. The estimated median is 0.20 percentage points, with both
normal intervals approximately $[0.18,0.22]$ and
$I_{\mathrm C}=[0.05,0.44]$. At the lower quartile, the smoothed
estimate changes from $-0.09$ in the balanced sample to $0.08$ using
all loans; the normal intervals become positive, while
$I_{\mathrm C}=[-0.03,0.29]$ still includes zero. At the lower decile,
$I_{\mathrm U}$ and $I_{\mathrm C}$ include zero, and the upper
endpoint of $I_{\mathrm{BC}}$ is almost zero. These comparisons show
that the main 24-month finding persists, while the smaller changes
at 12 months are more sensitive to the number of loans used and the
resulting estimation precision.

\paragraph{Alternative balanced selections.}
We repeat the five procedures for 100 independent random selections
of $T=449$ loans per credit score and cohort, without replacement.
Each selection uses the same loans at both horizons. Because the
calculations depend only on event counts, selection is implemented
by drawing from the observed counts of loans with an event by
12 months, a first event after 12 but within 24 months, or no event
by 24 months. Their multivariate hypergeometric distribution is
exactly that induced by selecting loans uniformly without replacement.
The original selection remains the main specification.

\begin{table}[!htbp]\centering\small
\setlength{\tabcolsep}{4pt}
\caption{Sensitivity to the selection of the balanced sample}
\label{tab:mortgageselections}
\begin{tabular*}{.96\textwidth}{@{\extracolsep{\fill}}rrrrrrrr@{}}\toprule
 & \multicolumn{2}{c}{Range of $\widehat q$}
 & \multicolumn{5}{c}{\shortstack{Intervals entirely above zero\\(out of 100 selections)}}\\
\cmidrule(lr){2-3}\cmidrule(lr){4-8}
$\tau$ & Minimum & Maximum & $I_{\mathrm U}$ & $I_{\mathrm{BC}}$
& $I_{\mathrm C}$ & $I_{\mathrm S}$ & $I_{\mathrm P}$\\\midrule
\multicolumn{8}{l}{\textit{A. 24-month serious delinquency}}\\
0.10 & 0.17 & 0.51 & 100 & 100 & 1 & 90 & 0\\
0.25 & 0.89 & 1.21 & 100 & 100 & 100 & 100 & 0\\
0.50 & 2.13 & 2.55 & 100 & 100 & 100 & 100 & 0\\
0.75 & 3.90 & 4.45 & 100 & 100 & 100 & 100 & 0\\
0.90 & 5.68 & 6.57 & 100 & 100 & 100 & 100 & 0\\
\addlinespace[5pt]
\multicolumn{8}{l}{\textit{B. 12-month serious delinquency}}\\
0.10 & -0.65 & -0.43 & 0 & 0 & 0 & 0 & 0\\
0.25 & -0.17 & -0.04 & 0 & 0 & 0 & 0 & 0\\
0.50 & 0.08 & 0.25 & 100 & 100 & 0 & 82 & 0\\
0.75 & 0.54 & 0.72 & 100 & 100 & 0 & 100 & 0\\
0.90 & 1.05 & 1.34 & 100 & 100 & 98 & 100 & 0\\
\bottomrule\end{tabular*}
\par\smallskip\begin{minipage}{.96\textwidth}\footnotesize
Notes: Each selection contains $T=449$ loans per credit score and cohort,
with $n=181$ credit scores. Estimates are in percentage points.
The last five columns count selections in which the interval's lower
endpoint is positive. These counts describe sensitivity to loan selection,
not confidence-interval coverage.
\end{minipage}\end{table}

Both normal intervals remain positive at every reported 24-month
quantile in all 100 selections. The conservative interval supports
a positive lower quartile in all 100, while its lower-decile interval
is positive in only one. The distinction between improvement across
more than three quarters and across more than 90\% of credit-score
values is therefore stable across balanced selections. At 12 months,
both normal median intervals remain positive in every selection,
whereas the conservative median interval includes zero throughout.
The normal lower-quartile intervals are negative in all 100 selections.
Thus, changing which loans enter the balanced sample has little effect
on its main sign conclusions; using all available loans changes both
the sample sizes and the precision of the coefficient estimates.

Across the 100 alternative balanced selections, mean interval lengths
at the 24-month median are 0.36, 0.36, 3.00, 0.86, and 10.95
percentage points for $I_{\mathrm U}$, $I_{\mathrm{BC}}$,
$I_{\mathrm C}$, $I_{\mathrm S}$, and $I_{\mathrm P}$, respectively.
The corresponding 12-month lengths are 0.13, 0.13, 0.96, 0.24,
and 6.04 percentage points.

\paragraph{Lower bounds on estimated standard deviations.}
Holding the original sample fixed, we replace the lower bound $0.01$
on $\widehat\sigma_i$ by $0.005$ and $0.02$, retaining the upper bound
of $2$. No sign conclusion changes for $I_{\mathrm U}$,
$I_{\mathrm{BC}}$, or $I_{\mathrm C}$ at either horizon. At 24 months,
their interval endpoints retain the reported two-decimal values; the
largest change is below 0.002 percentage points. At 12 months, the
largest endpoint changes across the five quantiles are 0.028, 0.028,
and 0.012 percentage points, respectively. The stochastic reference
has one sign change: its 12-month lower-quartile interval lies below
zero at the bound $0.02$, becoming $[-0.206,-0.021]$,
whereas it includes zero at the other two bounds. The projection
interval is unchanged because it does not use these standard deviations.

\FloatBarrier

\subsection{Decomposing the aggregate delinquency decline}
\label{subsec:mortgagedecomposition}

The quantile results show how broadly delinquency risk declined across
credit-score values. A complementary question is how much of the overall
decline reflects improvements at given credit scores, rather than a
larger share of loans going to borrowers with higher credit scores.
We separate these two contributions using all eligible loans.

Let $w_{ia}=m_{ia}/\sum_jm_{ja}$ be the share of cohort $a$'s loans
with credit score $x_i$, where $a=0$ denotes 2007 and $a=1$ denotes 2009.
Its aggregate delinquency frequency is
$\widehat R_a=\sum_iw_{ia}\widehat p_{ia}$. Both the delinquency
frequencies and the credit-score shares can change between cohorts:
\begin{equation}\label{eq:mortgagedecomposition}
 \begin{aligned}
 \widehat R_0-\widehat R_1
 ={}&\underbrace{\sum_i\frac{w_{i0}+w_{i1}}2
            (\widehat p_{i0}-\widehat p_{i1})}
            _{\text{delinquency changes at given credit scores}}+\underbrace{\sum_i\frac{\widehat p_{i0}+\widehat p_{i1}}2
            (w_{i0}-w_{i1})}
            _{\text{changes in the distribution of credit scores}}.
 \end{aligned}
\end{equation}
The first term measures the decline when credit-score shares are held
fixed at their average across the two cohorts. The second measures the
contribution of changing credit-score shares, holding delinquency
frequencies at their cohort averages. The two terms add exactly to the
aggregate decline.

\begin{table}[!htbp]\centering\small
\setlength{\tabcolsep}{6pt}
\caption{Contributions to the aggregate delinquency decline}
\label{tab:mortgagedecomposition}
\begin{tabular*}{.96\textwidth}{@{\extracolsep{\fill}}lrrr@{}}\toprule
Horizon & Total decline
& \shortstack{Delinquency changes\\at given credit scores}
& \shortstack{Changes in\\credit-score shares}\\\midrule
24 months & 2.985 & 1.927 (64.6\%) & 1.058 (35.4\%)\\
12 months & 0.440 & 0.199 (45.1\%) & 0.242 (54.9\%)\\
\bottomrule\end{tabular*}
\par\smallskip\begin{minipage}{.96\textwidth}\footnotesize
Notes: All eligible loans with credit scores 620--800 are included.
Values are in percentage points; parentheses give each component's share
of the total decline. Rounding may affect the displayed sums.
\end{minipage}\end{table}

At 24 months, the aggregate delinquency frequency fell by 2.99 percentage
points. Reductions at given credit scores account for 1.93 points,
or about two thirds of this decline; changes in credit-score shares
account for the remaining 1.06 points. Thus, a change in the mix of
borrowers does not explain the entire improvement. At 12 months,
the aggregate reduction is smaller, and changes in credit-score shares
account for a larger proportion.

This decomposition and the quantile analysis address different aspects
of the same improvement. The decomposition measures contributions to the
loan-weighted aggregate decline; the equally weighted quantiles show
how broadly the decline extends across credit-score values. Improvements
at a given credit score can reflect changes in economic conditions,
borrower characteristics, and loan terms, so this comparison remains
descriptive.

\clearpage

\section{Local Conditions and Time-Series Dependence}\label{sec:extension}

\subsection{Local balance}\label{sec:balance}
Smooth coefficient functions evaluated on a fixed grid provide a
sufficient condition for local balance. The grid may consist of
observed characteristic values, such as credit scores, or the
midpoints of fixed segments. Both coefficients and their first-step
standard deviations may vary across grid points.

\begin{extensionprop}[Local balance across fixed segments]
\label{prop:segmentbalance}
Let $x_i=a+(i-1/2)(b-a)/n$, and suppose there exist fixed
functions $g$ and $\sigma$ satisfying (i)
\(
 \max_{1\le i\le n}
 \{|\theta_i-g(x_i)|+|\sigma_i-\sigma(x_i)|\}\le Cn^{-2};
\)
 (ii) $g$ and $\sigma$ are twice continuously differentiable
on $[a,b]$, with $\inf_{x\in[a,b]}\sigma(x)>0$; (iii) for
\(
 G(z)=\frac1{b-a}\int_a^b\ind\{g(v)\le z\}\,dv,
 \  q=G^{-1}(\tau),
\)
the equation $g(x)=q$ has finitely many solutions, all in $(a,b)$,
with $g'(x)\ne0$ at each solution. Then, as $n,T\to\infty$,
\begin{equation}\label{eq:segmentbalance}
 |H_n(q_n)-\tau|=O(n^{-1}+T^{-1}).
\end{equation}
Consequently, Assumption~\ref{ass:balance} holds when
$T^{1/2}\ll n\ll T^{3/2}$.
\end{extensionprop}

The conditions require segment coefficients and standard deviations to be well approximated by smooth functions evaluated at segment midpoints. The coefficient function must cross the target level with nonzero slope, away from the endpoints.  

\begin{proof}
We compare the coefficient array with its smooth response function,
then bound the effect of normal smoothing. All bounds below are
uniform for $|z-q|\le\epsilon$, with $\epsilon>0$ sufficiently small.

\medskip\noindent
\textit{Step 1: prove $\sup_{|z-q|\le\epsilon}|F_n(z)-G(z)|=O(n^{-1})$.}
Write $v_1(z),\ldots,v_m(z)$ for the solutions of $g(v)=z$.
The assumptions give a fixed number of such solutions and
$|g'(v_j(z))|\ge c>0$. Hence
\[
 G'(z)=\frac1{b-a}\sum_{j=1}^m\frac1{|g'(v_j(z))|},
 \qquad 0<c\le G'(z)\le C.
\]
The set $\{v:g(v)\le z\}$ is a union of finitely many intervals.
For each interval, its grid-point count differs from $n$ times
its length divided by $b-a$ by at most two. Therefore
\[
 \frac1n\sum_i\ind\{g(x_i)\le z\}=G(z)+O(n^{-1}).
\]
Since $|\theta_i-g(x_i)|\le Cn^{-2}$, bounding $F_n(z)$ by
these grid counts at $z-Cn^{-2}$ and $z+Cn^{-2}$ proves the claim.
The positive density of $G$ then gives
\begin{equation}\label{eq:segmentquantile}
 q_n-q=O(n^{-1}),\qquad G(q_n)-\tau=O(n^{-1}).
\end{equation}

\medskip\noindent
\textit{Step 2: approximate $H_n$ by an integral with error $O(n^{-1})$.}
Define
\[
 H^*(z)=\frac1{b-a}\int_a^b
 \Phi\left(\frac{\sqrt T\{z-g(v)\}}{\sqrt2\sigma(v)}\right)dv.
\]
Near each crossing, change variables from $v$ to $g(v)$.
Away from these crossings, $|z-g(v)|$ is bounded away from zero.
These two bounds give
\[
 \int_a^b\left|
 \frac{d}{dv}\Phi\left(
 \frac{\sqrt T\{z-g(v)\}}{\sqrt2\sigma(v)}\right)\right|dv
 \le C\sqrt T\int_a^b
 (|g'(v)|+|z-g(v)|)e^{-cT\{z-g(v)\}^2}\,dv\le C.
\]
The midpoint Riemann-sum error is thus $O(n^{-1})$:
\[
 \frac1n\sum_i
 \Phi\left(\frac{\sqrt T\{z-g(x_i)\}}{\sqrt2\sigma(x_i)}\right)
 =H^*(z)+O(n^{-1}).
\]
The derivative of each summand with respect to its standard deviation
is uniformly bounded, so replacing $\sigma(x_i)$ by $\sigma_i$
adds $O(n^{-2})$. Bounding the coefficient perturbations by
$Cn^{-2}$ in either direction consequently yields
\[
 H^*(z-Cn^{-2})-C/n\le H_n(z)
 \le H^*(z+Cn^{-2})+C/n.
\]
The same change of variables gives $\sup_{|z-q|\le\epsilon}H^{*\prime}(z)\le C$.
It follows that
\begin{equation}\label{eq:segmentintegral}
 \sup_{|z-q|\le\epsilon}|H_n(z)-H^*(z)|=O(n^{-1}).
\end{equation}

\medskip\noindent
\textit{Step 3: prove $\sup_{|z-q|\le\epsilon}|H^*(z)-G(z)|=O(T^{-1})$.}
For a standard normal $Z$,
\[
 H^*(z)=\E\left[\frac1{b-a}\int_a^b
 \ind\{g(v)+\sqrt2\sigma(v)Z/\sqrt T\le z\}\,dv\right].
\]
On $|Z|\le c\sqrt T$, with a sufficiently small fixed $c>0$,
the perturbed equation has one root $v_{j,T}$ near each $v_j(z)$.
Taylor expansion at $v_j(z)$ gives
\[
 v_{j,T}-v_j(z)
 =-\frac{\sqrt2\sigma(v_j(z))}{g'(v_j(z))}\frac Z{\sqrt T}
   +O(Z^2/T).
\]
Adding the changes in the lengths of the intervals below $z$ yields
\[
 \frac1{b-a}\int_a^b
 \ind\{g(v)+\sqrt2\sigma(v)Z/\sqrt T\le z\}\,dv-G(z)
 =-\frac{\sqrt2 Z}{(b-a)\sqrt T}
   \sum_{j=1}^m\frac{\sigma(v_j(z))}{|g'(v_j(z))|}
   +O(Z^2/T).
\]
The event $|Z|\le c\sqrt T$ is symmetric, so the expectation of
the linear term is zero. The complementary event has exponentially small
probability. Thus
\[
 \sup_{|z-q|\le\epsilon}|H^*(z)-G(z)|
 \le \frac C T \E Z^2+C P(|Z|>c\sqrt T)=O(T^{-1}).
\]
Combining this with \eqref{eq:segmentquantile} and
\eqref{eq:segmentintegral} at $z=q_n$ proves
\eqref{eq:segmentbalance}. Finally,
\(
 \frac{n^{-1}+T^{-1}}{n^{-1/2}T^{-1/4}}
 =\sqrt{\frac{\sqrt T}{n}}+\sqrt{\frac n{T^{3/2}}}\longrightarrow0
\)
under the stated growth conditions.
\end{proof}

\bigskip
\ 
 
\bigskip

For exact grid values $\theta_i=g(x_i)$ and $\sigma_i=\sigma(x_i)$,
condition (i) of Proposition~\ref{prop:segmentbalance} holds with
zero error. For the mortgage comparison, write the two cohort
probabilities as $p_0(x)$ and $p_1(x)$. Then
\[
 g(x)=p_0(x)-p_1(x),\qquad
 \sigma(x)^2=p_0(x)\{1-p_0(x)\}+p_1(x)\{1-p_1(x)\}.
\]
Smooth probabilities bounded away from zero and one give the required
smooth positive standard deviation.

A reflected-pair construction covers the growth range of
Theorem~\ref{thm:bc}. Suppose $q_n$ is unique and every coefficient
other than $q_n$ within some fixed distance $\epsilon>0$ is paired
with its reflection about $q_n$, with equal $\sigma_i$ within each
pair and bounded positive scales. With $k=\lceil n\tau\rceil$,
paired normal probabilities sum to one, the central contribution is
$1/2$, and Gaussian tails give
\[
 H_n(q_n)=\frac{k-1/2}{n}+O(e^{-cT}),\qquad
 |H_n(q_n)-\tau|=O(n^{-1}+e^{-cT}).
\]
This is $o(n^{-1/2}T^{-1/4})$ when $n/\sqrt T\to\infty$ and
$n/T^{3/2}$ is bounded. The paired distances may be irregular,
subject to Assumption~\ref{ass:counts}.

\subsection{Local counts}\label{subsec:ordinaryiid}
The fixed-segment model above also implies
Assumption~\ref{ass:counts}. From Step 1 of the proof of
Proposition~\ref{prop:segmentbalance}, the empirical CDF differs
from a distribution with positive, bounded density near its quantile
by $O(n^{-1})$. Hence, for sufficiently small $x>0$,
\[
 \begin{aligned}
 \#\{i:q_n<\theta_i\le q_n+x\}
 &=n\{G(q_n+x)-G(q_n)\}+O(1),\\
 \#\{i:q_n-x\le\theta_i<q_n\}
 &=n\{G(q_n)-G(q_n-x)\}+O(1).
 \end{aligned}
\]
Both counts lie between $cnx-O(1)$ and $Cnx+O(1)$.
For $x\ge T^{-1/2}$, the $O(1)$ terms are negligible when
$n/\sqrt T\to\infty$. The same bound controls the number of
coefficients in the entire neighborhood, including ties.

Ordinary independent coefficient draws satisfy the local-count
condition as well. Suppose their common distribution $F$ has a
positive, continuous density near $q=F^{-1}(\tau)$. With probability
tending to one, $q_n$ lies in this neighborhood. Conditional on
$q_n=z$ and the labels on either side, the other coefficients are
independent draws from the two truncated distributions. Writing
$k=\lceil n\tau\rceil$, the conditional means of the side counts are
\[
 \begin{aligned}
 \E[\#\{i:z<\theta_i\le z+x\}\mid q_n=z]
 &=(n-k)\frac{F(z+x)-F(z)}{1-F(z)},\\
 \E[\#\{i:z-x\le\theta_i<z\}\mid q_n=z]
 &=(k-1)\frac{F(z)-F(z-x)}{F(z)}.
 \end{aligned}
\]
Each mean is bounded above and below by positive multiples of $nx$.
Binomial concentration at $x=T^{-1/2}$ therefore proves
\eqref{eq:twosides}. For \eqref{eq:uppercounts}, apply the upper
bound at radii $2^j/\sqrt T$ up to a fixed neighborhood.
The total failure probability is bounded by
\[
 C\sum_{j\ge0}\exp\{-cn2^j/\sqrt T\}\longrightarrow0.
\]
Enlarging any intervening radius to the next dyadic radius proves
the bound for every $T^{-1/2}\le x\le\epsilon$. This establishes
Assumption~\ref{ass:counts} with probability tending to one;
it does not establish local balance of the realized array.

\subsection{Time-series dependence}\label{sec:dependence}
With serial dependence, the first-step standard deviation is determined
by the long-run variance of the influence contributions. Its estimator
may converge more slowly than required by Assumption~\ref{ass:first}.
The following result permits such slower convergence for $I_C$.

\begin{extensionprop}[Conservative inference with serial dependence]
\label{prop:weak}
Suppose Assumption~\ref{ass:counts} holds and unit samples are
independent, while observations within a unit may be serially dependent.
Each first-step estimate uses only that unit's observations. Suppose
$c\le\sigma_i,\widehat\sigma_i\le C$ and
\begin{equation}\label{eq:weakfirst}
 \begin{gathered}
 \sup_{i,z}\left|P\left\{
 \frac{\sqrt T(\widehat\theta_i-\theta_i)}{\sigma_i}\le z
 \right\}-\Phi(z)\right|\longrightarrow0,\\
 \sup_i \E\left|
 \frac{\sqrt T(\widehat\theta_i-\theta_i)}{\sigma_i}
 \right|^4\le C,\qquad
 \sup_i \E(\widehat\sigma_i-\sigma_i)^2\longrightarrow0.
 \end{gathered}
\end{equation}
Then $P(q_n\in I_C)\to1$, without Assumption~\ref{ass:balance}
or an upper growth restriction on $n$.
\end{extensionprop}

\begin{proof}
We obtain the local mean and variance bounds used in the proof of
Theorem~\ref{thm:bound}, now for the full first-step estimation error.
In this proof, write
\[
 V_i=\frac{\sqrt T(\widehat\theta_i-\theta_i)}{\sigma_i},\qquad
 \eta_T=\sup_{i,z}|P(V_i\le z)-\Phi(z)|,\qquad
 \varepsilon_T=\sup_i\{\E(\widehat\sigma_i/\sigma_i-1)^2\}^{1/2}.
\]
Both $\eta_T$ and $\varepsilon_T$ tend to zero.

\medskip\noindent
\textit{Step 1: prove that the local mean approximation error is
$o(T^{-1/2})$.}
The fourth-moment bound and Markov's inequality give
$|P(V_i\le z)-\Phi(z)|\le C|z|^{-4}$ for $|z|\ge1$.
Combining this with the uniform bound $\eta_T$, and splitting
at $|z|=\eta_T^{-1/4}$ when $\eta_T>0$, gives
\begin{equation}\label{eq:weaknonuniform}
 \sup_{i,z}(1+|z|^3)|P(V_i\le z)-\Phi(z)|
 \le C\eta_T^{1/4}\longrightarrow0.
\end{equation}
For any of the functions $g(v)=\Phi(v)$ or
$g(v)=\Phi(v)\pm2\Phi(-|v|)$, integration by parts as in
Lemma~\ref{lem:comparison} yields
\[
 |\E g(u-V_i)-\E g(u-Z)|
 \le \frac{C\eta_T^{1/4}}{1+|u|^3},\qquad Z\sim N(0,1).
\]
For the estimated standard deviation, the mean value theorem
and Cauchy-Schwarz give
\[
 \E\left|g\left(\frac{u-V_i}{\widehat\sigma_i/\sigma_i}\right)
              -g(u-V_i)\right|
 \le C\varepsilon_T
 \{e^{-cu^2}+\eta_T^{1/8}(1+|u|)^{-3/2}\}.
\]
The envelopes on the right average to $O(T^{-1/2})$ on every fixed
$KT^{-1/2}$ neighborhood of $q_n$, by Lemma~\ref{lem:envelope}.
Thus the mean approximations for $\widehat F$ and
$\widehat F\pm2\bar p$ have uniform error $o(T^{-1/2})$ there.

\medskip\noindent
\textit{Step 2: prove uniform estimation error $o_p(T^{-1/2})$.}
On $|V_i|\le|u|/2$, each contribution differs from
$\ind\{u\ge0\}$ by at most $Ce^{-cu^2}$, uniformly over the
bounded standard-deviation ratios. On the complement, use
\eqref{eq:weaknonuniform} and boundedness of $g$ to obtain
\[
 \E\left[g\left(\frac{u-V_i}{\widehat\sigma_i/\sigma_i}\right)
                 -\ind\{u\ge0\}\right]^2
 \le Ce^{-cu^2}+\frac{C\eta_T^{1/4}}{1+|u|^3}.
\]
Local counts and independence across units therefore give
\[
 \sup_{|x-q_n|\le K/\sqrt T}
 \Var\{\widehat F(x)\pm2\bar p(x)\}
 =O\{(n\sqrt T)^{-1}\}.
\]
The expected supremum of $\widehat F'$ over this neighborhood is
$O(1)$ by the same envelope bound. Also, the absolute one-sided
derivatives of $\widehat F\pm2\bar p$ are at most
$3\widehat F'$. A grid with
$M=\lfloor(n/\sqrt T)^{1/2}\rfloor$ cells has total Chebyshev
bound $O(M\sqrt T/n)=o(1)$ and interpolation error
$O_p(T^{-1/2}/M)=o_p(T^{-1/2})$. Hence
\[
 \sup_{|x-q_n|\le K/\sqrt T}
 |\widehat F(x)\pm2\bar p(x)
       -\E\{\widehat F(x)\pm2\bar p(x)\}|
 =o_p(T^{-1/2}).
\]
The same argument applies to $\widehat F$ alone.

\medskip\noindent
\textit{Step 3: prove $P(q_n\in I_C)\to1$.}
The normal-reference root bracket in the proof of
Proposition~\ref{prop:root} and the one-sided mean margins in
Lemma~\ref{lem:conservativebounds} have order $T^{-1/2}$. Steps 1-2 preserve them because both
approximation errors are smaller. In particular,
$|\widehat q-q_n|=O_p(T^{-1/2})$, and the segment argument in
the proof of Theorem~\ref{thm:bound} gives
\[
 P\{\,|\widehat F(x)-\tau|\le2\bar p(x)
       \text{ for every }x\text{ between }q_n\text{ and }\widehat q\,\}
 \longrightarrow1.
\]
The first outward endpoints therefore enclose $q_n$ with
probability tending to one.
\end{proof}

\end{document}